\documentclass[journal,onecolumn,twoside]{IEEEtran}
\normalsize
\ifCLASSINFOpdf
\else
\fi
\usepackage{setspace}
\usepackage{color}
\usepackage{cite}
\usepackage{amsmath}
\usepackage{amsfonts}
\usepackage{amssymb}
\usepackage{amsthm}
\usepackage{tikz}
\usepackage{cases}
\usepackage{bm}
\usepackage{graphicx}
    \graphicspath{{../}}
    \DeclareGraphicsExtensions{.pdf}
\usepackage[caption=true,font=footnotesize]{subfig}
\usepackage{multirow}
\usepackage{makecell}
\usepackage{mathdots}
\usepackage{booktabs}
\usepackage{url}
\usepackage{bm}
\usepackage{xtab}
\usepackage{pifont}
\usepackage{array}
\usepackage{algorithm}
\usepackage{algorithmic}
\usepackage{enumerate}
\usetikzlibrary{arrows}
\usepackage{caption}
\usepackage{soul}
\usepackage{xcolor}

\allowdisplaybreaks[4]
\usepackage[colorlinks,
            linkcolor=blue,
            anchorcolor=blue,
            citecolor=blue]{hyperref}
\theoremstyle{plain}

\newtheorem{theorem}{Theorem}
\newtheorem{lemma}[theorem]{Lemma}
\newtheorem{proposition}[theorem]{Proposition}
\newtheorem{definition}[theorem]{Definition}
\newtheorem{corollary}[theorem]{Corollary}
\theoremstyle{definition}

\newtheorem{remark}{Remark}
\newtheorem{conjecture}{Conjecture}

\begin{document}

\title{Some New Results on Sequence Reconstruction Problem for Deletion Channels}

\author{Xiang Wang,~Weijun Fang,~Han Li, and ~Fang-Wei Fu
\thanks{X. Wang is with the School of Mathematics, Statistics and Mechanics, Beijing University of Technology, Beijing, 100124, China (e-mail: xwang@bjut.edu.cn). Weijun Fang is with State Key Laboratory of Cryptography and Digital Economy Security, Shandong University, Qingdao, 266237, China, Key Laboratory of Cryptologic Technology and Information Security, Ministry of Education, Shandong University, Qingdao, 266237, China and School of Cyber Science and Technology, Shandong University, Qingdao, 266237, China (email: fwj@sdu.edu.cn). Han Li is with the Chern Institute of Mathematics and LPMC, Nankai University, Tianjin 300071, China (e-mail: hli@mail.nankai.edu.cn). F.-W. Fu is with the Chern Institute of Mathematics and LPMC, Nankai University, Tianjin 300071, China (e-mail: fwfu@nankai.edu.cn).}}



\maketitle

\begin{abstract}
Levenshtein first introduced the sequence reconstruction problem in $2001$. In the realm of combinatorics, the sequence reconstruction problem is equivalent to determining the value of $N(n,d,t)$, which represents the maximum size of the intersection of two metric balls of radius $t$, given that the distance between their centers is at least $d$ and the sequence length is $n$. In this paper, We present a lower bound on $N(n,3,t)$ for $n\geq \max\{13,t+8\}$ and $t \geq 4$. For $t=4$, we prove that this lower bound is tight. This settles an open question posed by Pham, Goyal, and Kiah, confirming that $N(n,3,4)=20n-166$ for all $n \geq 13$.
\end{abstract}


\section{Introduction}
\IEEEPARstart{T}{he} sequence reconstruction problem was initially proposed by Levenshtein~\cite{L1} in $2001$. In a communication model, a sequence $\mathbf{x}$ from a code $\mathcal{C}$, is transmitted through several noisy channels. A decoder then receives the distinct outputs from these noisy channels and reconstructs the original transmitted sequence $\mathbf{x}$. Levenshtein~\cite{L1,L2} determined the minimum number of transmission channels required for the exact reconstruction of the transmitted sequence. Originally motivated by applications in biology and chemistry, the sequence reconstruction problem has seen a resurgence of interest due to the advent of certain innovative data storage technologies. These include DNA-based data storage systems~\cite{Church,Yazdi,Lenz} and racetrack memories~\cite{Parkin,Chee}, which offer multiple inexpensive and noisy reads.

Let $S$ denote the set of all sequences, and let $\rho: S \times S \rightarrow \mathbb{N}$ represent a metric defined over the sequences in $S$. Consider the code $\mathcal{C}$, which is defined as a subset of $S$ and equipped with the distance metric $\rho$. Suppose the code $\mathcal{C}$ has a minimum distance of $d$, and each channel introduces at most $t$ errors, with the transmitted sequence being an element of $\mathcal{C}$.  In this framework, Levenshtein~\cite{L1} proved that the minimum number of transmission channels required exceeds the size of the largest intersection of two metric balls:
\begin{equation}
N(n,d,t)=\max\limits_{\mathbf{x}_1,\mathbf{x}_2\in S,\rho(\mathbf{x}_1,\mathbf{x}_2)\geq d }\{|B_t(\mathbf{x}_1)\cap B_t(\mathbf{x}_2)|\},\label{eq1}
\end{equation}
where $B_t(\mathbf{x})$ denotes the ball of radius $t$ centered at $\mathbf{x}$, and the length of any sequence is $n$. We refer to the task of determining $N(n,d,t)$ as the \emph{sequence reconstruction problem}.

Levenshtein~\cite{L1} investigated the sequence reconstruction problem, as defined in Equation~(\ref{eq1}), in the context of various channels, including those based on the Hamming distance, Johnson graphs, and several other metrics.  Subsequent research addressed this problem within the framework of permutations in~\cite{K1,Yaakobi,Wang1,Wang2,Wang3}, and it was further explored in relation to additional general error graphs in~\cite{L3,L4}. The problem has also been studied for the Grassmann graph in~\cite{Yaakobi} and deletions or insertions~\cite{Sala1,Lan,Sun1,Zhang,Song}. This variant addresses the dual of the sequence reconstruction problem \cite{Chee,Chrisnata,Cai,Sun2,Sun3,Sun4,Ye,Wu}, where codes are designed for a fixed number of channels to guarantee unique reconstruction at the receiver.

The problem of sequence reconstruction for the deletion channel has received attention in the academic literature. In this context, the deletion ball of a sequence $\mathbf{x}$ with radius $t$ is defined as the collection of all sequences resulting from up to $t$ deletions from $\mathbf{x}$. Levenshtein~\cite{L1,L2} first studied the sequence reconstruction problem for the deletion channel and determined the value of $N(n,1,t)$, where $S$ denotes the set of all nonbinary/binary sequences. In the realm of binary sequences, some authors further studied the sequence reconstruction problem for the deletion channel \cite{Gabrys,Pham,Pham2}. Specifically, Gabrys and Yaakobi~\cite{Gabrys} resolved this problem for the deletion channel by establishing the value of $N(n,2,t)$, under the condition that the Levenshtein distance between any two distinct binary sequences is at least two. Recently, Pham, Goyal, and Kiah~\cite{Pham,Pham2} presented an asymptotic solution of $N(n,d,t)$ for all $d\geq 2$, which is given by $N(n,d,t)=\frac{\binom{2d}{d}}{(t-d)!}n^{t-d}-O(n^{t-d-1})$ for $0\leq d\leq t<n$. In the specific case where $d=t$, they found that $N(n,d,t)=\binom{2t}{t}$. 
Furthermore, Zhang et al. \cite{Zhang} studied the sequence reconstruction problem for the binary $3$-deletion channel and characterized pairs of distinct binary sequences $\mathbf{x},\mathbf{y}$ for which $|D_3(\mathbf{x})\cap D_3(\mathbf{y})|\in \{19,20\}$, given that the distance between $\mathbf{x}$ and $\mathbf{y}$ is at least $3$. Here, $D_3(\mathbf{x})$ and $D_3(\mathbf{y})$ denote the deletion balls of radius $3$ centered at $\mathbf{x}, \mathbf{y}\in \{0,1\}^{n}$, respectively.

For the case where $d=3$ and $t=4$, they \cite{Pham2} gave an upper bound of $N(n,3,4)$, that is, $N(n,3,4)\leq 20n-150$ for any $n\geq 9$. 
In this paper, we study the sequence reconstruction problem for binary sequences over the deletion channel. We propose a lower bound on $N(n,3,t)$ for $n\geq 13$ and $t\geq 4$, and determine that $N(n,3,4)=20n-166$ holds for all $n\geq 13$. Furthermore, for any $n\geq 5$, we explicitly construct two length-$n$ binary sequences $\mathbf{x}$ and $\mathbf{y}$ with Levenshtein distance at least three, such that the cardinality of the intersection of their $4$-deletion balls is $N(n,3,4)$.

The rest of this paper is organized as follows. In Section~\ref{sec2}, we introduce definitions and notations related to the sequence reconstruction problem, along with some foundational results. In Section~\ref{sec3}, we prove a lower bound on $N(n,3,t)$ for any $n\geq 13$ and $t\geq 4$. In Section~\ref{sec4}, we show that $20n-166$ is an upper bound on $N(n,3,4)$ for any $n\geq 13$. Section~\ref{sec5} concludes the paper.

\section{Definitions and Preliminaries}
\label{sec2}

In this paper, we follow the same notation as stated in~\cite{Gabrys} and~\cite{Pham}. Let $\mathbb{F}_2$ be the set $\{0,1\}$ and $[n]=\{1,2,...,n-1,n\}$. Let $a\in \mathbb{F}_2$. Scalars are denoted by lowercase letters, while vectors or sequences are denoted by bold lowercase letters. Let $\mathbf{x}=x_1\,x_2\,\cdots\,x_n\in \mathbb{F}_2^n$. Let $|\mathbf{x}|$ denote the length of the sequence $\mathbf{x}$. Let $\mathbf{x}_{[i,j]}$ be the projection of $\mathbf{x}$ on the indices of one interval $[i,j]$, i.e., $\mathbf{x}_{[i,j]}=x_i\,\cdots\,x_j$. A run of $\mathbf{x}$ is a maximal interval which consists of the same symbol. The concatenation of two sequences $\mathbf{u}=u_1\, \cdots\, u_m$ and $\mathbf{v}=v_1\, \cdots\, v_n$ is denoted $\mathbf{u} \circ \mathbf{v}=u_1\, \cdots\, u_m\, v_1\, \cdots\, v_n$. The \emph{complement} of $a$ is defined to be $\overline{a}$. We also define the complement of a binary sequence $\mathbf{x}$, denoted $\overline{\mathbf{x}}$, to be the sequence obtained by taking the complement of each bit in $\mathbf{x}$.

Suppose $\mathcal{C} \subseteq \mathbb{F}^n_2$. For $\mathbf{u}\in \mathbb{F}_2^{m_1},\mathbf{v}=v_1\,\cdots\,v_{m_2}\in \mathbb{F}_2^{m_2}$ with $m_1,m_2\geq 0$, $\mathcal{C}_{\mathbf{v}}^{\mathbf{u}}$ is the set of all sequences in $\mathcal{C}$ that start with $\mathbf{u}$ and end with $v_{m_2}\,\cdots\,v_{1}$. For a sequence $\mathbf{u}=u_1\, \cdots\, u_m$, the set $\mathbf{u} \circ \mathcal{C}$ is defined as
\begin{equation}
\mathbf{u}\circ \mathcal{C}=\{\mathbf{u}\circ \mathbf{c}=u_1\,\cdots\,u_m\,c_1\,\cdots\,c_n|\mathbf{c}=c_1\,\cdots\,c_n\in \mathcal{C}\}.\nonumber
\end{equation}
For any sequence $x \in \mathbb{F}^n_2$, let $D_t(\mathbf{x})$ be the $t$-deletion ball centered at $\mathbf{x}$, that is,
\begin{equation}
D_t(\mathbf{x})=\{\mathbf{y}\in \mathbb{F}_2^{n-t}|\mathbf{y} ~\text{is a subsequence of}~ \mathbf{x}\}.\nonumber
\end{equation}
The \emph{Levenshtein distance} between any two sequences $\mathbf{x},\mathbf{y}\in \mathbb{F}_2^n$ is defined as
\begin{equation}
d_L(\mathbf{x},\mathbf{y})=\min\{t\geq 0| D_t(\mathbf{x})\cap D_t(\mathbf{y})\neq \emptyset\}.\nonumber
\end{equation}
For integers $1\leq d\leq t$, we define $N(n,d,t)$ to be the maximum possible size of the intersection of two deletion balls of radius $t$, centered at $\mathbf{x}, \mathbf{y} \in \mathbb{F}_2^n$ with distance at least $d$:
\begin{equation}
N(n,d,t)=\max
\{|D_t(\mathbf{x})\cap D_t(\mathbf{y})|:\mathbf{x},\mathbf{y}\in \mathbb{F}_2^n, d_L(\mathbf{x},\mathbf{y})\geq d\}.\nonumber
\end{equation}
For convenience, when $t<0$ or $t>n$, we denote $D_t(\mathbf{x})=\emptyset$.
Assuming a sequence from $\mathcal{C} \subseteq \mathbb{F}_2^n$ is transmitted over $N$ channels, with each channel experiencing exactly $t$ deletion errors, and all channel outputs are distinct, where $\mathcal{C}$ is a $(d-1)$-deletion-correcting code. Levenshtein~\cite{L1} proved that the minimum number of channels that guarantees the existence of a decoder that successfully decode any transmitted sequence is given by $N(n,d,t)+1$.

\begin{definition}
\label{def1}
A sequence $\mathbf{x}=x_1\,x_2\,\cdots\,x_n$ is $2$-periodic if $x_k=x_{k+2}$ for all $1\leq k\leq n-2$. We say that $\mathbf{x}$ is alternating if it is $2$-periodic and $x_1\neq x_2$. For convenience, length-one sequences are alternating.
\end{definition}

For convenience, we let $\mathbf{a}_n\in \mathbb{F}_2^n$ be the alternating sequence whose first bit is $1$. For $0<t<n$, let $D(n,t)$ denote the maximum cardinality of a deletion ball of radius $t$ in $\mathbb{F}_2^n$, which can be expressed as $D(n,t)=\max\{|D_t(\mathbf{x})|: \mathbf{x} \in \mathbb{F}_2^n\}$. As established in \cite{Calabi}, we have
\begin{equation}
D(n,t)=|D_t(\mathbf{a}_n)|=\sum\limits_{i=0}^{t}\binom{n-t}{i},\label{eq2}
\end{equation}
and it follows that
\begin{equation}
D(n,t)=D(n-1,t)+D(n-2,t-1).\label{eq3}
\end{equation}
For notational consistency, we define $D(n,t)=0$. in cases where $t>n, n<0$, or $t<0$.

The results of $N(n,d,t)$ have been discussed in the literature as follows. In the case where $d=1$, Levenshtein \cite{L1} presented the value of $N(n,1,t)$.
\begin{theorem}[Levenshtein \cite{L1}]
For $1 \leq t < n$,
\begin{equation}
N(n,1,t)=2D(n-2,t-1).\nonumber
\end{equation}
Furthermore, when $\mathbf{x}=\mathbf{a}_n$ and $\mathbf{y}$ is either $0\,1\,\mathbf{a}_{n-2}$ or $ 0\,\mathbf{a}_{n-1}$, we have $d_L(\mathbf{x},\mathbf{y})=1$ and $|D_t(\mathbf{x})\cap D_t(\mathbf{y})|=N(n,1,t)$.
\label{thm1}
\end{theorem}

For $d=2$, the expression for $N(n,2,t)$ was derived by Gabrys and Yaakobi \cite{Gabrys}.
\begin{theorem}[Gabrys and Yaakobi \cite{Gabrys}]
When $2 \leq t < n$ and $n \geq 8$,
\begin{align}
N(n,2,t)=2D(n-4,t-2)+2D(n-5,t-2)+2D(n-7,t-2)+D(n-6,t-3)+D(n-7,t-3).\label{eq4}
\end{align}
\label{thm2}
\end{theorem}
By Theorem $\ref{thm2}$, we have the following corollary.
\begin{corollary}
For any $n\geq 8$, we have
\begin{equation}
N(n,2,3)=6n-30.\label{eq5}
\end{equation}
\label{cor1}
\end{corollary}

Pham, Goyal, and Kiah \cite{Pham2} have provided asymptotically precise estimates of $N(n,d,t)$ for all $0 \leq d \leq t$.
\begin{theorem}[Pham, Goyal, and Kiah \cite{Pham2}]
For $0 \leq d \leq t < n$, we have that
\begin{align}
N(n,d,t)=\frac{\binom{2d}{d}}{(t-d)!}n^{t-d}-O(n^{t-d-1}).\nonumber
\end{align}
When $1\leq t=d$ and $n \geq 4t-2$, we obtain
\begin{equation}
N(n,d,d)=\binom{2d}{d}.\nonumber
\end{equation}
When $d=3$ and $t=4$, for any $n\geq 9$, we have
\begin{align}
N(n,3,4)\leq 20n-150.\nonumber
\end{align}
\label{thm3}
\end{theorem}

\begin{remark}
From Theorems \ref{thm1}-\ref{thm3}, we derive the following insights of $N(n,d,t)$:
\begin{itemize}
\item For $d=1,2$, the exact values of $N(n,d,t)$ have been established.
\item When $t=d=3$, it is evident that $N(n,3,3)=20$ for any $n \geq 10$.
\item For $d=3$ and $t=4$, Theorem \ref{thm3} implies $N(n,3,4)\leq 20n-150$ for any $n\geq 9$.
\end{itemize}
\end{remark}

For all $4\leq n\leq 14$, some values of $N(n,3,4)$ are given in Table $1$ of \cite{Pham2}. Specifically, $N(13,3,4)=94$ and $N(14,3,4)=114$ \cite{Pham2}. Building upon the definitions and preliminary results, we first establish that for $n\geq \max\{13,t+8\}$ and $t\geq 4$,
\begin{equation*}
N(n,3,t)\geq M(n,t),
\end{equation*}
where $M(n,t)$ is defined as 
\begin{equation*}
M(n,t)\triangleq 6D(n-6,t-3)+4D(n-8,t-3)+6D(n-9,t-3)+4D(n-11,t-3)+2D(n-13,t-5)+D(n-13,t-6).
\end{equation*}
Furthermore, for the specific case $t=4$, we determine the exact value
\begin{equation*}
N(n,3,4)=M(n,4)=20n-166
\end{equation*}
for all $n\geq 13$. For $5 \leq n\leq 12$, we used a computer search to give two length-$n$ sequences $\mathbf{x}$ and $\mathbf{y}$ such that the cardinality of the intersection of their $4$-deletion balls is $N(n,3,4)$ in Appendix \ref{APP-A}. In the following sections we only consider $n\geq 13$ and give our main result.

\begin{theorem}
For $n\geq \max\{13,t+8\}$ and $t\geq 4$, let $M(n,t)$ be defined as above. Then we have
\begin{equation*}
N(n,3,t)\geq M(n,t).
\end{equation*}
Furthermore, for $t=4$, the bound is tight and becomes 
\begin{equation*}
N(n,3,4)=M(n,4)=20n-166
\end{equation*}
for all $n\geq 13$.
\label{thm4}
\end{theorem}

\section{The lower bound}
\label{sec3}

The main objective of this section is to establish the lower bound  
\begin{equation*}
N(n,3,t) \geq M(n,t) \quad \text{for } n \geq \max\{13,t+8\} \text{ and } t \geq 4.
\end{equation*}
To this end, we construct explicit sequences $\mathbf{x}$ and $\mathbf{y}$ as follows:
\begin{itemize}
  \item If $n \geq 13$ is odd, let
  \begin{equation*}
    \mathbf{x}=1\,0\,1\,0\,1\,0\,\mathbf{a}_{n-8}\,1\,0,\qquad \mathbf{y}=0\,1\,1\,0\,0\,1\,\mathbf{a}_{n-8}\,0\,1;
  \end{equation*}  
  \item If $n \geq 13$ is even, let
  \begin{equation*}
    \mathbf{x}=1\,0\,1\,0\,1\,0\,\mathbf{a}_{n-8}\,0\,1,\qquad \mathbf{y}=0\,1\,1\,0\,0\,1\,\mathbf{a}_{n-8}\,1\,0.
  \end{equation*}  
\end{itemize}  
We will show that the Levenshtein distance between $\mathbf{x}$ and $\mathbf{y}$ is $3$, and that  
\begin{equation*}
|D_t(\mathbf{x})\cap D_t(\mathbf{y})|=M(n,t),
\end{equation*}
which directly leads to our desired lower bound. 

Before determining the values of $d_L(\mathbf{x},\mathbf{y})$ and $|D_t(\mathbf{x})\cap D_t(\mathbf{y})|$, we need some results obtained by Gabrys and Yaakobi \cite{Gabrys}.

\begin{proposition}[{\cite[Lemma 1]{Gabrys}}]
Let $n,m_1,m_2$ be non-negative integers such that $m_1+m_2\leq n-t$. For any $\mathbf{u}\in \mathbb{F}_2^n$ we have
\begin{equation}
|D_t(\mathbf{u})|=\sum\limits_{\mathbf{u}_1\in \mathbb{F}_2^{m_1}}\sum\limits_{\mathbf{u}_2\in \mathbb{F}_2^{m_2}}|D_t(\mathbf{u})_{\mathbf{u}_2}^{\mathbf{u}_1}|.\nonumber
\end{equation}
\label{prp1}
\end{proposition}

\begin{proposition}[{\cite[Lemma 2]{Gabrys}}]
Let $n,m_1,m_2,t$ be non-negative integers such that $m_1+m_2\leq n-t$, and $\mathbf{u}=u_1\,\cdots\,u_n\in \mathbb{F}_2^n, \mathbf{u}_1\in \mathbb{F}_2^{m_1}, \mathbf{u}_2=v_1\,\cdots\,v_{m_2}\in \mathbb{F}_2^{m_2}$.
Assume that $k_1$ is the smallest integer such that $\mathbf{u}_1$ is a subsequence of $u_1\,\cdots\,u_{k_1}$ and $k_2$ is the largest integer where $v_{m_2}\,\cdots\,v_{1}$ is a subsequence of $u_{k_2}\,\cdots\,u_n$. If $k_1<k_2$ then
\begin{equation}
D_t(\mathbf{u})_{\mathbf{u}_2}^{\mathbf{u}_1}=\mathbf{u}_1\circ D_{t^{*}}(u_{k_1+1}\,\cdots\,u_{k_2-1})\circ v_{m_2}\,\cdots\,v_{1},\nonumber
\end{equation}
where $t^*=t-(k_1-m_1)-(n-k_2+1-m_2)$. In particular, $|D_t(\mathbf{u})_{\mathbf{u}_2}^{\mathbf{u}_1}|=|D_{t^{*}}(u_{k_1+1}\,\cdots\,u_{k_2-1})|$.
\label{prp2}
\end{proposition}

We prove that $d_L(\mathbf{x},\mathbf{y})=3$ as follows.

\begin{lemma}
Let $n\geq \max\{13,t+8\}$ and $t\geq 4$. Define the sequences $\mathbf{x}$ and $\mathbf{y}$ as follows:
\begin{itemize}
    \item If $n \geq 13$ is odd, let  
          $\mathbf{x}=1\,0\,1\,0\,1\,0\,\mathbf{a}_{n-8}\,1\,0$ and $\mathbf{y}=0\,1\,1\,0\,0\,1\,\mathbf{a}_{n-8}\,0\,1$.
    \item If $n \geq 13$ is even, let
          $\mathbf{x}=1\,0\,1\,0\,1\,0\,\mathbf{a}_{n-8}\,0\,1$ and $\mathbf{y}=0\,1\,1\,0\,0\,1\,\mathbf{a}_{n-8}\,1\,0$.
\end{itemize}
Then
\begin{equation*}
d_L(\mathbf{x},\mathbf{y})=3.
\end{equation*}
\label{lm1}
\end{lemma}
\begin{proof}
First assume $n$ is odd. For convenience, let $\mathbf{x}=x_1\,x_2\,\cdots\,x_n$ and $\mathbf{y}=y_1\,y_2\,\cdots\,y_n$. We can obtain a subsequence $1\,1\,0\,1\,\mathbf{a}_{n-8}\,1$ by deleting the $2$-th, $6$-th, $n$-th elements in $\mathbf{x}$, or deleting the $1$-th, $5$-th, $(n-1)$-th elements in $\mathbf{y}$. Thus, $d_L(\mathbf{x},\mathbf{y})\leq 3$.

On the other hand, we will show that $d_L(\mathbf{x},\mathbf{y})>2$. That is, $D_2(\mathbf{x})\cap D_2(\mathbf{y})=\emptyset$. Let $\mathcal{X}=D_2(\mathbf{x})\cap D_2(\mathbf{y})$. By Propositions $\ref{prp1}$ and $\ref{prp2}$, we have  $|\mathcal{X}|=|\mathcal{X}^{00}|+|\mathcal{X}^{01}|+|\mathcal{X}^{10}|+|\mathcal{X}^{11}|$, where
$$\mathcal{X}^{00}=00\circ \big(D_0(1\,0\,\mathbf{a}_{n-8}\,1\,0)\cap D_0(0\,1\,\mathbf{a}_{n-8}\,0\,1)\big), $$
$$\mathcal{X}^{01}=01\circ \big(D_1(0\,1\,0\,\mathbf{a}_{n-8}\,1\,0)\cap D_2(1\,0\,0\,1\,\mathbf{a}_{n-8}\,0\,1)\big), $$
$$\mathcal{X}^{10}=10\circ \big(D_2(1\,0\,1\,0\,\mathbf{a}_{n-8}\,1\,0)\cap D_0(0\,1\,\mathbf{a}_{n-8}\,0\,1)\big), $$
$$\mathcal{X}^{11}=11\circ \big(D_1(0\,1\,0\,\mathbf{a}_{n-8}\,1\,0)\cap D_1(0\,0\,1\,\mathbf{a}_{n-8}\,0\,1)\big).$$

Consider the case of $\mathcal{X}^{00}$. Since $n\geq 13$, the first element of $1\,0\,\mathbf{a}_{n-8}\,1\,0$ and $0\,1\,\mathbf{a}_{n-8}\,0\,1$  are $1$ and $0$, respectively. So, we have that $\mathcal{X}^{00}=\emptyset$. Similarly, we also prove that $\mathcal{X}^{01}=\mathcal{X}^{10}=\mathcal{X}^{11}=\emptyset$. Consequently, the lemma is verified for all odd $n\geq 13$. An analogous argument covers the case of even $n$, thereby establishing the lemma.
\end{proof}

To compute $|D_t(\mathbf{x}) \cap D_t(\mathbf{y})|$ for the sequences defined in Lemma $\ref{lm1}$, we first establish the following auxiliary result.
\begin{lemma}
Then we have
\begin{equation}
|D_t(1\,\mathbf{a}_{n-1})|=D(n-1,t)+D(n-3,t-2),\quad\text{for}~ n\geq t+1,\nonumber
\end{equation}
\begin{equation}
|D_t(0\,1\,\mathbf{a}_{n-2})|=D(n-2,t)+D(n-2,t-1)+D(n-4,t-2),\quad \quad\text{for}~ n\geq t+2.\nonumber
\end{equation}
\label{lm2}
\end{lemma}
\begin{proof} 
For convenience, let $\mathbf{u}=1\,\mathbf{a}_{n-1}$. By Propositions \ref{prp1} and \ref{prp2}, we have $|D_t(\mathbf{u})|=|D_t(\mathbf{u})^{0}|+|D_t(\mathbf{u})^{1}|$, where
$$|D_t(\mathbf{u})^{0}|=|0\circ D_{t-2}(\mathbf{a}_{n-3})|=D(n-3,t-2), \qquad |D_t(\mathbf{u})^{1}|=|1\circ D_{t}(\mathbf{a}_{n-1})|=D(n-1,t).$$
Therefore, $|D_t(1\,\mathbf{a}_{n-1})|=D(n-1,t)+D(n-3,t-2)$ for any $n\geq t+1$. 

Let $\mathbf{v}=0\,1\,\mathbf{a}_{n-2}$. By Propositions \ref{prp1} and \ref{prp2}, we have $|D_t(\mathbf{v})|=|D_t(\mathbf{v})^{0}|+|D_t(\mathbf{v})^{1}|$, where
$$|D_t(\mathbf{v})^{0}|=|0\circ D_{t}(1\,\mathbf{a}_{n-2})|=|D_{t}(1\,\mathbf{a}_{n-2})|=D(n-2,t)+D(n-4,t-2),$$
$$|D_t(\mathbf{v})^{1}|=|1\circ D_{t-1}(\mathbf{a}_{n-2})|=D(n-2,t-1).$$
Thus, we have $|D_t(0\,1\,\mathbf{a}_{n-1})|=D(n-2,t)+D(n-2,t-1)+D(n-4,t-2)$ for any $n\geqslant t+2$. This completes the proof of the lemma.
\end{proof}

By Lemma $\ref{lm2}$, we have the following theorem which gives the value of $|D_t(\mathbf{x})\cap D_t(\mathbf{y})|$ for the sequences defined in Lemma $\ref{lm1}$.

\begin{theorem}
Let $n\geq \max\{13,t+8\}$ and $t\geq 4$. Then we have
\begin{equation}
N(n,3,t)\geq M(n,t).\nonumber
\end{equation}
In particular, $N(n,3,4)\geq M(n,4)=20n-166$.
\end{theorem}
\begin{proof}
Assume $n$ is odd. Let $\mathbf{x}=1\,0\,1\,0\,1\,0\,\mathbf{a}_{n-8}\,1\,0$ and  $\mathbf{y}=0\,1\,1\,0\,0\,1\,\mathbf{a}_{n-8}\,0\,1$. We denote $\mathcal{X}=D_t(\mathbf{x})\cap D_t(\mathbf{y})$. Then, by Proposition $\ref{prp1}$, we have $|\mathcal{X}|=|\mathcal{X}^{0}_0|+|\mathcal{X}^{0}_1|+|\mathcal{X}^{1}_0|+|\mathcal{X}^{1}_1|$,
where
\begin{align*}
|\mathcal{X}^0_0|=&|D_{t-1}(1\,0\,1\,0\,\mathbf{a}_{n-8}\,1)\cap D_{t-1}(1\,1\,0\,0\,1\,\mathbf{a}_{n-8})|,\\
|\mathcal{X}^0_1|=&|D_{t-2}(1\,0\,1\,0\,\mathbf{a}_{n-8})\cap D_{t}(1\,1\,0\,0\,1\,\mathbf{a}_{n-8}\,0)|,\\
|\mathcal{X}^1_0|=&|D_{t}(0\,1\,0\,1\,0\,\mathbf{a}_{n-8}\,1)\cap D_{t-2}(1\,0\,0\,1\,\mathbf{a}_{n-8})|,\\
|\mathcal{X}^1_1|=&|D_{t-1}(0\,1\,0\,1\,0\,\mathbf{a}_{n-8})\cap D_{t-1}(1\,0\,0\,1\,\mathbf{a}_{n-8}\,0)|.
\end{align*}

First we compute the value of $|\mathcal{X}^{0}_0|$. We continue to decompose the set of $\mathcal{X}^{0}_0$ as follows. Since $n$ is odd and $n\geq \max\{13,t+8\}$, then we have $|\mathcal{X}^0_0|=|\mathcal{X}^{00}_{00}|+|\mathcal{X}^{000}_{010}|+|\mathcal{X}^{000}_{011}|+|\mathcal{X}^{001}_{01}|+|\mathcal{X}^{010}_{00}|
+|\mathcal{X}^{011}_{00}|+|\mathcal{X}^{010}_{010}|+|\mathcal{X}^{010}_{011}|+|\mathcal{X}^{0110}_{010}|+|\mathcal{X}^{0110}_{011}|+
|\mathcal{X}^{0111}_{01}|$ with
\begin{small}
\begin{align*}
|\mathcal{X}^{00}_{00}|=&|D_{t-4}(1\,0\,\mathbf{a}_{n-10})\cap D_{t-4}(0\,1\,\mathbf{a}_{n-10})|\overset{(a)}{=}2D(n-10,t-5),\\
|\mathcal{X}^{000}_{010}|=&|D_{t-4}(\mathbf{a}_{n-10})\cap D_{t-3}(1\,\mathbf{a}_{n-10})|=|D_{t-4}(\mathbf{a}_{n-10})|\overset{(b)}{=}D(n-10,t-4),\\
|\mathcal{X}^{000}_{011}|=&|D_{t-3}(\mathbf{a}_{n-10}\,0)\cap D_{t-4}(1\,\mathbf{a}_{n-11})|=|D_{t-4}(1\,\mathbf{a}_{n-11})|\overset{(c)}{=}D(n-11,t-4)+D(n-13,t-6),\\
|\mathcal{X}^{001}_{01}|=&|D_{t-2}(0\,\mathbf{a}_{n-8})\cap D_{t-4}(\mathbf{a}_{n-9})|=|D_{t-4}(\mathbf{a}_{n-9})|\overset{(b)}{=}D(n-9,t-4),\\
|\mathcal{X}^{010}_{00}|=&|D_{t-3}(1\,0\,\mathbf{a}_{n-10})\cap D_{t-3}(0\,1\,\mathbf{a}_{n-10})|\overset{(a)}{=}2D(n-10,t-4),\\
|\mathcal{X}^{011}_{00}|=&|D_{t-4}(0\,\mathbf{a}_{n-10})\cap D_{t-2}(0\,0\,1\,\mathbf{a}_{n-10})|=|D_{t-4}(0\,\mathbf{a}_{n-10})|\overset{(b)}{=}D(n-9,t-4),\\
|\mathcal{X}^{010}_{010}|=&|D_{t-2}(1\,0\,\mathbf{a}_{n-10})\cap D_{t-2}(0\,1\,\mathbf{a}_{n-10})|\overset{(a)}{=}2D(n-10,t-3),\\
|\mathcal{X}^{010}_{011}|=&|D_{t-1}(1\,0\,\mathbf{a}_{n-9})\cap D_{t-3}(0\,1\,\mathbf{a}_{n-11})|=|D_{t-3}(0\,1\,\mathbf{a}_{n-11})|\overset{(d)}{=}D(n-11,t-3)+D(n-11,t-4)+D(n-13,t-5),\\
|\mathcal{X}^{0110}_{010}|=&|D_{t-3}(\mathbf{a}_{n-10})\cap D_{t-1}(0\,1\,\mathbf{a}_{n-10})|=|D_{t-3}(\mathbf{a}_{n-10})|\overset{(b)}{=}D(n-10,t-3),\\
|\mathcal{X}^{0110}_{011}|=&|D_{t-2}(\mathbf{a}_{n-9})\cap D_{t-2}(0\,1\,\mathbf{a}_{n-11})|\overset{(a)}{=}2D(n-11,t-3),\\
|\mathcal{X}^{0111}_{01}|=&|D_{t-3}(0\,\mathbf{a}_{n-10})\cap D_{t-3}(\mathbf{a}_{n-9})|\overset{(a)}{=}2D(n-11,t-4),
\end{align*}
\end{small}
where $(a)$ follows from Theorem $\ref{thm1}$, $(b)$ follows by applying Equation $(\ref{eq2})$, $(c)$ follows from Lemma $\ref{lm2}$, $(d)$ follows from Lemma $\ref{lm2}$. Thus, we have
\begin{align}
|\mathcal{X}^0_0|&=3D(n-10,t-3)+3D(n-11,t-3)+2D(n-9,t-4)+3D(n-10,t-4)+4D(n-11,t-4)\nonumber\\
&~~~~~~~~~~~+2D(n-10,t-5)+D(n-13,t-5)+D(n-13,t-6).\nonumber
\end{align}

Next consider the set $\mathcal{X}^{0}_1$. Then, for $n\geq \max\{13,t+8\}$, we have $|\mathcal{X}^0_1|=|\mathcal{X}^{00}_{10}|+|\mathcal{X}^{00}_{11}|+|\mathcal{X}^{01}_{10}|+|\mathcal{X}^{010}_{11}|+|\mathcal{X}^{011}_{11}|$ with
\begin{align*}
|\mathcal{X}^{00}_{10}|=&|D_{t-4}(1\,0\,\mathbf{a}_{n-10})\cap D_{t-2}(0\,1\,\mathbf{a}_{n-8})|=|D_{t-4}(1\,0\,\mathbf{a}_{n-10})|\overset{(a)}{=}D(n-8,t-4),\\
|\mathcal{X}^{00}_{11}|=&|D_{t-3}(1\,0\,\mathbf{a}_{n-9})\cap D_{t-3}(0\,1\,\mathbf{a}_{n-9})|\overset{(b)}{=}2D(n-9,t-4),\\
|\mathcal{X}^{01}_{10}|=&|D_{t-3}(0\,1\,0\,\mathbf{a}_{n-10})\cap D_{t}(1\,0\,0\,1\,\mathbf{a}_{n-8})|=|D_{t-3}(0\,1\,0\,\mathbf{a}_{n-10})|\overset{(a)}{=}D(n-7,t-3),\\
|\mathcal{X}^{010}_{11}|=&|D_{t-2}(1\,0\,\mathbf{a}_{n-9})\cap D_{t-2}(0\,1\,\mathbf{a}_{n-9})|\overset{(b)}{=}2D(n-9,t-3),\\
|\mathcal{X}^{011}_{11}|=&|D_{t-3}(0\,\mathbf{a}_{n-9})\cap D_{t-1}(0\,0\,1\,\mathbf{a}_{n-9})|=|D_{t-3}(0\,\mathbf{a}_{n-9})|\overset{(a)}{=}D(n-8,t-3),
\end{align*}
where $(a)$ follows by applying Equation $(\ref{eq2})$, $(b)$ follows from Theorem $\ref{thm1}$. Thus, we have
\begin{align}
|\mathcal{X}^0_1|&=D(n-7,t-3)+D(n-8,t-3)+2D(n-9,t-3)+D(n-8,t-4)+2D(n-9,t-4).\nonumber
\end{align}

We discuss the value of  $|\mathcal{X}^{1}_0|$. Then, for $n\geq \max\{13,t+8\}$, we have $|\mathcal{X}^1_0|=|\mathcal{X}^{10}_{00}|+|\mathcal{X}^{10}_{01}|+|\mathcal{X}^{110}_{00}|+|\mathcal{X}^{111}_{00}|+
|\mathcal{X}^{110}_{010}|+|\mathcal{X}^{110}_{011}|+|\mathcal{X}^{111}_{01}|$ with
\begin{small}
\begin{align*}
|\mathcal{X}^{10}_{00}|=&|D_{t-2}(1\,0\,1\,0\,\mathbf{a}_{n-10})\cap D_{t-4}(0\,1\,\mathbf{a}_{n-10})|=|D_{t-4}(0\,1\,\mathbf{a}_{n-10})|\overset{(a)}{=}D(n-10,t-4)+D(n-10,t-5)+D(n-12,t-6),\\
|\mathcal{X}^{10}_{01}|=&|D_{t}(1\,0\,1\,0\,\mathbf{a}_{n-8})\cap D_{t-3}(0\,1\,\mathbf{a}_{n-9})|=|D_{t-3}(0\,1\,\mathbf{a}_{n-9})|\overset{(a)}{=}D(n-9,t-3)+D(n-9,t-4)+D(n-11,t-5),\\
|\mathcal{X}^{110}_{00}|=&|D_{t-3}(1\,0\,\mathbf{a}_{n-10})\cap D_{t-3}(0\,1\,\mathbf{a}_{n-10})|\overset{(b)}{=}2D(n-10,t-4),\\
|\mathcal{X}^{111}_{00}|=&|D_{t-4}(0\,\mathbf{a}_{n-10})\cap D_{t-5}(\mathbf{a}_{n-10})|=|D_{t-5}(\mathbf{a}_{n-10})|\overset{(c)}{=}D(n-10,t-5),\\
|\mathcal{X}^{110}_{010}|=&|D_{t-2}(1\,0\,\mathbf{a}_{n-10})\cap D_{t-2}(0\,1\,\mathbf{a}_{n-10})|\overset{(b)}{=}2D(n-10,t-3),\\
|\mathcal{X}^{110}_{011}|=&|D_{t-1}(1\,0\,\mathbf{a}_{n-9})\cap D_{t-3}(0\,1\,\mathbf{a}_{n-11})|=|D_{t-3}(0\,1\,\mathbf{a}_{n-11})|\overset{(a)}{=}D(n-11,t-3)+D(n-11,t-4)+D(n-13,t-5),\\
|\mathcal{X}^{111}_{01}|=&|D_{t-2}(0\,\mathbf{a}_{n-8})\cap D_{t-4}(\mathbf{a}_{n-9})|=|D_{t-4}(\mathbf{a}_{n-9})|\overset{(c)}{=}D(n-9,t-4),
\end{align*}
\end{small}
where $(a)$ follows from Lemma $\ref{lm2}$, $(b)$ follows from Theorem $\ref{thm1}$, $(c)$ follows by applying Equation $(\ref{eq2})$. Thus, we have
\begin{align}
|\mathcal{X}^1_0|&=D(n-9,t-3)+2D(n-10,t-3)+D(n-11,t-3)+2D(n-9,t-4)+3D(n-10,t-4)\nonumber\\
&~~~~~~~~~~~+D(n-11,t-4)+2D(n-10,t-5)+D(n-11,t-5)+D(n-13,t-5)+D(n-12,t-6).\nonumber
\end{align}

Finally consider the set $\mathcal{X}^{1}_1$. Then, for $n\geq \max\{13,t+8\}$, we have $|\mathcal{X}^1_1|=|\mathcal{X}^{10}_{10}|+|\mathcal{X}^{10}_{11}|+|\mathcal{X}^{11}_{10}|+|\mathcal{X}^{110}_{11}|+
|\mathcal{X}^{111}_{11}|$ with
\begin{small}
\begin{align*}
|\mathcal{X}^{10}_{10}|=&|D_{t-2}(1\,0\,\mathbf{a}_{n-8})\cap D_{t-2}(0\,1\,\mathbf{a}_{n-8})|\overset{(a)}{=}2D(n-8,t-3),\\
|\mathcal{X}^{10}_{11}|=&|D_{t-1}(1\,0\,1\,0\,\mathbf{a}_{n-9})\cap D_{t-3}(0\,1\,\mathbf{a}_{n-9})|=|D_{t-3}(0\,1\,\mathbf{a}_{n-9})|\overset{(b)}{=}D(n-9,t-3)+D(n-9,t-4)+D(n-11,t-5),\\
|\mathcal{X}^{11}_{10}|=&|D_{t-3}(0\,1\,0\,\mathbf{a}_{n-10})\cap D_{t-1}(0\,0\,1\,\mathbf{a}_{n-8})|=|D_{t-3}(0\,1\,0\,\mathbf{a}_{n-10})|\overset{(c)}{=}D(n-7,t-3),\\
|\mathcal{X}^{110}_{11}|=&|D_{t-2}(1\,0\,\mathbf{a}_{n-9})\cap D_{t-2}(0\,1\,\mathbf{a}_{n-9})|\overset{(a)}{=}2D(n-9,t-3),\\
|\mathcal{X}^{111}_{11}|=&|D_{t-3}(0\,\mathbf{a}_{n-9})\cap D_{t-4}(\mathbf{a}_{n-9})|=|D_{t-4}(\mathbf{a}_{n-9})|\overset{(c)}{=}D(n-9,t-4),
\end{align*}
\end{small}
where $(a)$ follows from Theorem $\ref{thm2}$, $(b)$ follows from Lemma $\ref{lm2}$, $(c)$ follows by applying Equation $(\ref{eq2})$. Thus, we have
\begin{align}
|\mathcal{X}^1_1|=D(n-7,t-3)+2D(n-8,t-3)+3D(n-9,t-3)+2D(n-9,t-4)+D(n-11,t-5).\nonumber
\end{align}

Therefore, for any $n\geqslant \max\{13,t+8\}$ we have
\begin{align}
|\mathcal{X}|&=|\mathcal{X}^{0}_0|+|\mathcal{X}^{0}_1|+|\mathcal{X}^{1}_0|+|\mathcal{X}^{1}_1|\nonumber\\
&=2D(n-7,t-3)+3D(n-8,t-3)+6D(n-9,t-3)+5D(n-10,t-3)+4D(n-11,t-3)\nonumber\\
&~~~+D(n-8,t-4)+8D(n-9,t-4)+6D(n-10,t-4)+5D(n-11,t-4)\nonumber\\
&~~~+4D(n-10,t-5)+2D(n-11,t-5)+2D(n-13,t-5)+D(n-12,t-6)+D(n-13,t-6)\nonumber\\
&\overset{(a)}{=}6D(n-6,t-3)+4D(n-8,t-3)+6D(n-9,t-3)+4D(n-11,t-3)+2D(n-13,t-5)+D(n-13,t-6),\nonumber
\end{align}
where $(a)$ follows by applying Equation $(\ref{eq3})$.

When $n$ is even, let $\mathbf{x}=1\,0\,1\,0\,1\,0\,\mathbf{a}_{n-8}\,0\,1$ and  $\mathbf{y}=0\,1\,1\,0\,0\,1\,\mathbf{a}_{n-8}\,1\,0$. We denote $\mathcal{Y}=D_t(\mathbf{x})\cap D_t(\mathbf{y})$. Then, by Proposition $\ref{prp1}$, we have $|\mathcal{Y}|=|\mathcal{Y}^{0}_0|+|\mathcal{Y}^{0}_1|+|\mathcal{Y}^{1}_0|+|\mathcal{Y}^{1}_1|$,
where
\begin{align*}
|\mathcal{Y}^0_0|=&|D_{t-2}(1\,0\,1\,0\,\mathbf{a}_{n-8})\cap D_{t}(1\,1\,0\,0\,1\,\mathbf{a}_{n-8}\,1)|,\\
|\mathcal{Y}^0_1|=&|D_{t-1}(1\,0\,1\,0\,\mathbf{a}_{n-8}\,0)\cap D_{t-1}(1\,1\,0\,0\,1\,\mathbf{a}_{n-8})|,\\
|\mathcal{Y}^1_0|=&|D_{t-1}(0\,1\,0\,1\,0\,\mathbf{a}_{n-8})\cap D_{t-1}(1\,0\,0\,1\,\mathbf{a}_{n-8}\,1)|,\\
|\mathcal{Y}^1_1|=&|D_{t}(0\,1\,0\,1\,0\,\mathbf{a}_{n-8}\,0)\cap D_{t-2}(1\,0\,0\,1\,\mathbf{a}_{n-8})|.
\end{align*}
Here, we decompose $\mathcal{Y}^0_0$ using the similar decomposition method of $\mathcal{X}^0_1$ and view the rightmost $1$ in $\mathcal{Y}^0_0$ as the rightmost $0$ in $\mathcal{X}^0_1$. Since $n$ is even, then we have $|\mathcal{Y}^0_0|=|\mathcal{Y}^{00}_{01}|+|\mathcal{Y}^{00}_{00}|+|\mathcal{Y}^{01}_{01}|+|\mathcal{Y}^{010}_{00}|+|\mathcal{Y}^{011}_{00}|$ with
\begin{align*}
|\mathcal{Y}^{00}_{01}|=&|D_{t-4}(1\,0\,\mathbf{a}_{n-10})\cap D_{t-2}(0\,1\,\mathbf{a}_{n-8})|=|D_{t-4}(1\,0\,\mathbf{a}_{n-10})|\overset{(a)}{=}D(n-8,t-4),\\
|\mathcal{Y}^{00}_{00}|=&|D_{t-3}(1\,0\,\mathbf{a}_{n-9})\cap D_{t-3}(0\,1\,\mathbf{a}_{n-9})|\overset{(b)}{=}2D(n-9,t-4),\\
|\mathcal{Y}^{01}_{01}|=&|D_{t-3}(0\,1\,0\,\mathbf{a}_{n-10})\cap D_{t}(1\,0\,0\,1\,\mathbf{a}_{n-8})|=|D_{t-3}(0\,1\,0\,\mathbf{a}_{n-10})|\overset{(a)}{=}D(n-7,t-3),\\
|\mathcal{Y}^{010}_{00}|=&|D_{t-2}(1\,0\,\mathbf{a}_{n-9})\cap D_{t-2}(0\,1\,\mathbf{a}_{n-9})|\overset{(b)}{=}2D(n-9,t-3),\\
|\mathcal{Y}^{011}_{00}|=&|D_{t-3}(0\,\mathbf{a}_{n-9})\cap D_{t-1}(0\,0\,1\,\mathbf{a}_{n-9})|=|D_{t-3}(0\,\mathbf{a}_{n-9})|\overset{(a)}{=}D(n-8,t-3),
\end{align*}
where $(a)$ follows by applying Equation $(\ref{eq2})$, $(b)$ follows from Theorem $\ref{thm1}$. Thus, we have
\begin{align}
|\mathcal{Y}^0_0|=|\mathcal{X}^0_1|&=D(n-7,t-3)+D(n-8,t-3)+2D(n-9,t-3)+D(n-8,t-4)+2D(n-9,t-4).\nonumber
\end{align}
Similarly, we also have $|\mathcal{Y}^0_1|=|\mathcal{X}^0_0|$, $|\mathcal{Y}^1_0|=|\mathcal{X}^1_1|$, and $|\mathcal{Y}^1_1|=|\mathcal{X}^1_0|$. Thus, when $n$ is even, we have $|\mathcal{Y}|=M(n,t)$. So, $N(n,3,t)\geq M(n,t)$ for any $n\geq \max\{13,t+8\}$ and $t\geqslant 4$. In particular, when $t=4$, then $M(n,4)=20n-166$ since $D(m,1)=m$. Thus, the theorem is proved.
\end{proof}

\section{The upper bound}
\label{sec4}
In this section, we will prove that the lower bound $M(n,4)=20n-166$ established in the previous section also serves as an upper bound for $N(n,3,4)$. For convenience, let $\mathbf{x}=x_1\,\cdots\,x_{n}, \mathbf{y}=y_1\,\cdots\,y_n\in \mathbb{F}_2^{n}$, and let $\mathcal{X}=D_t(\mathbf{x})\cap D_t(\mathbf{y})$ ($n \geq t \geq 1$), Proposition \ref{prp2} implies that for $\mathbf{u}_1\in \mathbb{F}_2^{m_1}$ and $\mathbf{u}_2=v_{1}\,\cdots\,v_{m_2}\in \mathbb{F}_2^{m_2}$,
\begin{align*}
\mathcal{X}_{\mathbf{u}_2}^{\mathbf{u}_1}=\mathbf{u}_1\circ \big(D_{t-(1+\ell-m_1)-(1+\ell_1-m_2)}(x_{2+\ell}\,\cdots\,x_{n-1-\ell_1})
\cap D_{t-(1+\ell^*-m_1)-(1+\ell_1^*-m_2)}(y_{2+\ell^*}\,\cdots\,y_{n-1-\ell_1^*})\big)\circ v_{m_2}\,\cdots\,v_{1}.
\end{align*}
Here, $\ell,\ell^*$ are defined as the smallest integer such that $\mathbf{u}_1$ is a subsequence of $x_1\,\cdots\,x_{1+l}$ and $y_1\,\cdots\,y_{1+l^*}$, respectively. Similarly, $\ell_1,\ell_1^*$ are defined as the largest integer such that $\mathbf{u}_2$ is a subsequence of $x_{n-\ell_1}\,\cdots\,x_n$ and $y_{n-\ell_1^*}\,\cdots\,y_n$, respectively.  Unless otherwise stated, $\ell,\ell^*,\ell_1,\ell_1^*$ are defined this way during the decomposition of $\mathcal{X}$, with $\ell, \ell^*,\ell_1,\ell_1^* \geq 0$ for simplicity. First, we give the main theorem in the following. 

\begin{theorem}
\label{thm5}
For any $n\geq 13$, let $\mathbf{x}=x_1\,x_2\,\cdots\,x_n\in \mathbb{F}_2^n, \mathbf{y}=y_1\,y_2\,\cdots\,y_n\in \mathbb{F}_2^n$ such that $d_L(\mathbf{x}, \mathbf{y})\geq 3$.
Then, $|D_4(\mathbf{x})\cap D_4(\mathbf{y})|\leq 20n-166$.
\end{theorem}
\begin{proof}
We will prove the above result by induction on the lengths of the sequences $\mathbf{x,y}$. The base cases for $n=13$ or $14$ were verified using a computerized search.

Let $\mathbf{x}=x_1\,x_2\,\cdots\,x_n=x_1\,\mathbf{x}', \mathbf{y}=y_1\,y_2\,\cdots\,y_n=y_1\,\mathbf{y}'$ and denote $\mathcal{X}=D_4(\mathbf{x})\cap D_4(\mathbf{y})$. There are two cases: case $1)$ $x_1= y_1$; case $2)$ $x_1\neq y_1$. Suppose the result holds for all $m\leq n-1$.

First, we consider case $1)$. By Proposition $\ref{prp1}$, we have
$$|\mathcal{X}|=|\mathcal{X}^{x_1}|+|\mathcal{X}^{\overline{x_1}}|$$
with $\mathcal{X}^{x_1}=x_1\circ \big(D_4(\mathbf{x}')\cap D_4(\mathbf{y}')\big)$ and $\mathcal{X}^{\overline{x_1}}=\overline{x_1}\circ\big(D_{3-\ell}(\mathbf{x}'')\cap D_{3-\ell^*}(\mathbf{y}'')\big)$, where $\mathbf{x}''=x_{3+\ell}\,\cdots\,x_n$, $\mathbf{y}''=y_{3+\ell^*}\,\cdots\,y_n$, and $\ell,\ell^*\geq 0$.

Then  $d_L(\mathbf{x}', \mathbf{y}')\geq 3$ and from the inductive hypothesis we get
$$|\mathcal{X}^{x_1}|=|D_4(\mathbf{x}')\cap D_4(\mathbf{y}')|\leq 20(n-1)-166.$$

Moreover, by Theorem $\ref{thm3}$ we have
$$|\mathcal{X}^{\overline{x_1}}|=|D_{3-\ell}(\mathbf{x}'')\cap D_{3-\ell^*}(\mathbf{y}'')|\leq |D_{3}(\mathbf{x}')\cap D_{3}(\mathbf{y}')|=20$$ for any $n\geq 13$ since $D_{3-\ell}(\mathbf{x}'')\subset D_{3}(\mathbf{x}')$ and $D_{3-\ell^*}(\mathbf{y}'')\subset D_{3}(\mathbf{y}')$. Therefore, $|\mathcal{X}|\leq 20n-166$ which completes the proof for the case where $x_1=y_1$.

Next, we consider case $2):$ $x_1\neq y_1$. We further assume that $x_n \neq y_n$, since otherwise, if $x_n = y_n$, we can reverse both sequences, which reduces the problem to case~$1)$. Moreover, the conditions $x_1\neq y_1$ and $x_n\neq y_n$ can be divided into the following subcases:
\begin{itemize}
\item (a) The subcase where $x_i \neq y_i$ for $i=1,n$, and at least one of the following holds: $x_1=x_2, x_{n-1}=x_n, y_1=y_2$, or $y_{n-1}=y_n$;
\item (b)  The subcase with $x_i \neq y_i$ (for $i=1, n$), $x_j \neq x_{j+1}$ and $y_j \neq y_{j+1}$ (for $j = 1, n-1$), and with either ($x_3 \neq y_3~ \text{and}~ x_2=x_3$) or ($x_{n-2} \neq y_{n-2}~ \text{ and }~ x_{n-1} = x_{n-2}$);
\item (c) The subcase with $x_i \neq y_i$ (for $i =1, n$), $x_j \neq x_{j+1}$ and $y_j \neq y_{j+1}$ (for $j = 1, n-1$), and with either ($x_3\neq y_3, x_2\neq x_3, x_{n-2}\neq y_{n-2}$) or ($x_{n-2}\neq y_{n-2}, x_{n-1}\neq x_{n-2}, x_3\neq y_3$); 
\item (d) The subcase with $x_i \neq y_i$ (for $i =1, n$), $x_j \neq x_{j+1}$ and $y_j \neq y_{j+1}$ (for $j = 1, n-1$), and with either ($x_3\neq y_3, x_2\neq x_3, x_{n-2}=y_{n-2}$) or ($x_{n-2}\neq y_{n-2}, x_{n-1}\neq x_{n-2}, x_3=y_3$); 
\item (e) The subcase with $x_i \neq y_i$ (for $i = 1, n$), $x_j \neq x_{j+1}$ and $y_j \neq y_{j+1}$ (for $j = 1, n-1$), and with $x_l = y_l$ for $l = 3, n-2$, and at least one of the following: $x_4 \neq y_4$, $x_4 = y_4 = x_3$, $x_{n-3} \neq y_{n-3}$, or $x_{n-3} = y_{n-3} = x_{n-2}$;
\item (f) The subcase where $x_i \neq y_i$ (for $i = 1, n$), $x_j \neq x_{j+1}$ and $y_j \neq y_{j+1}$ (for $j = 1, 3,n-3, n-1$), and $x_l = y_l$ for ($l = 3, 4, n-3, n-2$).
\end{itemize}  
For the subcases $(a)$-$(f)$, we can prove that $|D_4(\mathbf{x})\cap D_4(\mathbf{y})|\leq 20n-166$ by using Lemmas $\ref{lm8}$, $\ref{lm9}$, $\ref{lm10}$, $\ref{lm11}$, $\ref{lm12}$, $\ref{lm17}$, respectively. So, the theorem holds.

\end{proof}

To facilitate understanding that subcases $(a)$–$(f)$ exhaust all possible cases, the logical flow is illustrated in Fig.~$1$. For convenience, let $\mathbf{x}=x_1\,\cdots\,x_n, \mathbf{y}=y_1\,\cdots\,y_n\in \mathbb{F}_2^n$ and let $a,b\in \mathbb{F}_2$. Without loss of generality, we consider the subcases $(a)$-$(f)$ by setting $x_1=1,y_1=0,x_n=a,y_n=\overline{a}$.

\begin{figure}[ht]
\centering
\includegraphics[scale=0.45]{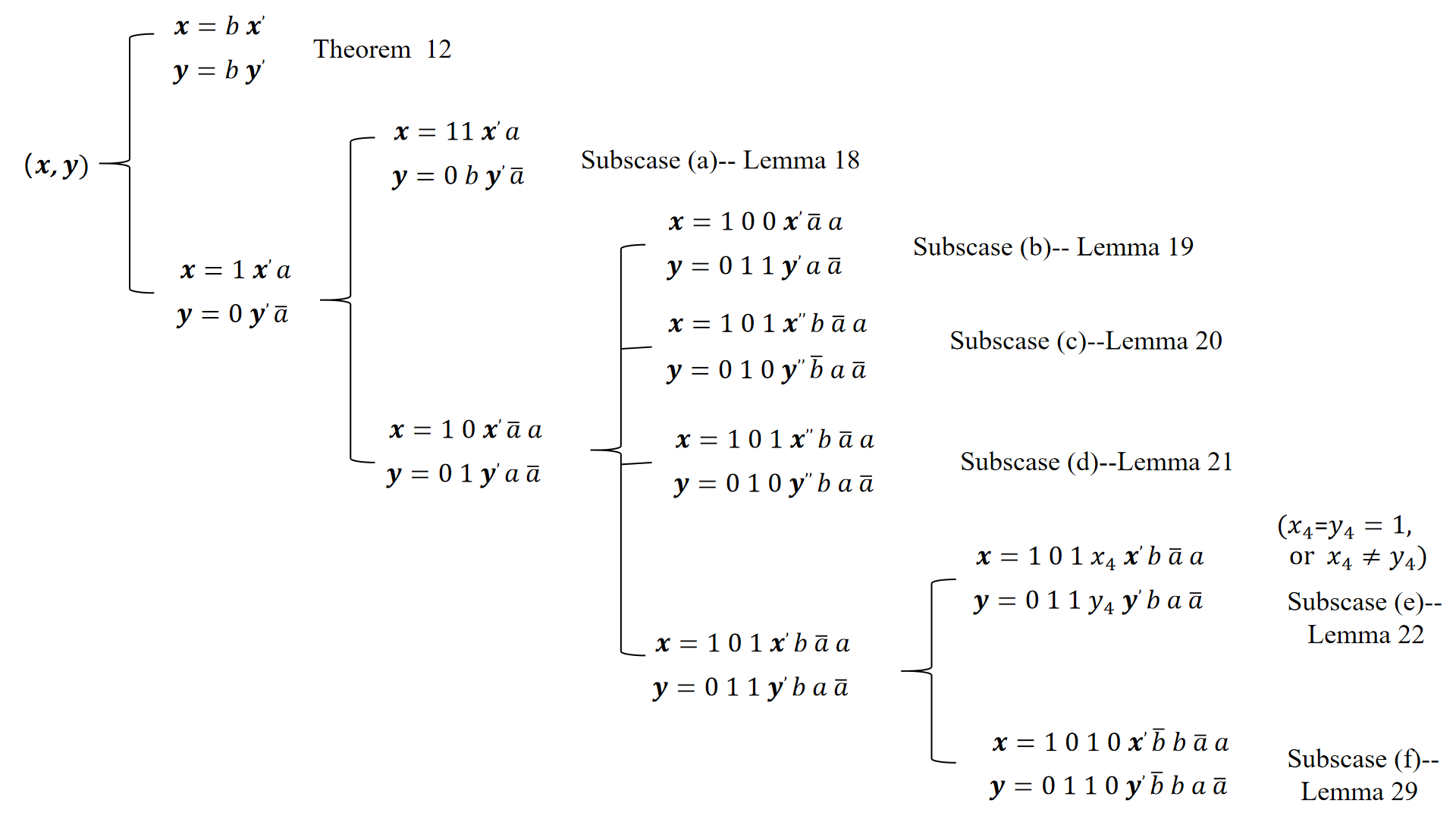}
\caption{Outline of proof of upper bound.}
\end{figure}

Now, we examine the logical flow shown in Fig. $1$. A basic observation is that analyzing the structure of $\mathbf{x}$ and $\mathbf{y}$ from either the left or the right side is equivalent. Here, the cases $x_1 = y_1$ and $x_n = y_n$ can be analyzed using the same approach. Moreover, the analysis of $\mathbf{x}$ and $\mathbf{y}$ is symmetric. For example, the approach used to analyze the case $x_1=x_2$ is exactly the same as that for the case $y_1=y_2$. Therefore, we only discuss some cases of $\mathbf{x,y}$ from the left side in the following lemmas which is also suitable for other cases from the right side. Unless otherwise specified, we assume 
$d_L(\mathbf{x},\mathbf{y})\geq 3$ in the following.

$d_L(\mathbf{x},\mathbf{y})\geq 3$ in the following.

\begin{enumerate}
    \item \textbf{Subcase (a):} For the case where $x_1 \neq y_1$ and $x_n \neq y_n$, Lemma $\ref{lm8}$ demonstrates that if $x_1= x_2$, then $|D_4(\mathbf{x})\cap D_4(\mathbf{y})|\leq 20n-166$. Similarly, the conclusion also holds for $x_{n-1}=x_{n}$, or $y_1=y_2$, or $y_{n-1}=y_n$.
    \item \textbf{Subcase (b):} For the case where $x_i \neq y_i$ (for $i=1, n$), $x_j \neq x_{j+1}$ and $y_j \neq y_{j+1}$ (for $j=1, n-1$), Lemma $\ref{lm9}$ demonstrates that if $x_3\neq  y_3$ and $x_2 = x_3$, then $|D_4(\mathbf{x})\cap D_4(\mathbf{y})|\leq 20n-166$. Similarly, the conclusion also holds for $x_{n-2}\neq y_{n-2}$ and $x_{n-1}=x_{n-2}$ by reversing $\mathbf{x}$ and $\mathbf{y}$. 
    \item \textbf{Subcase (c):} For the case where $x_i \neq y_i$ (for $i=1, n$), $x_j \neq x_{j+1}$ and $y_j \neq y_{j+1}$ (for $j=1, n-1$), Lemma $\ref{lm10}$ states that if $x_3\neq  y_3, x_2\neq x_3$, and $x_{n-2}\neq x_{n-1}, x_{n-2}\neq y_{n-2}$, then $|D_4(\mathbf{x})\cap D_4(\mathbf{y})|\leq 20n-166$, and Lemma $\ref{lm9}$ states that if $x_{n-2}=x_{n-1}, x_{n-2}\neq y_{n-2}$, then $|D_4(\mathbf{x})\cap D_4(\mathbf{y})|\leq 20n-166$. Thus, the conclusion holds for $x_3\neq  y_3, x_2\neq x_3$, and $x_{n-2}\neq y_{n-2}$. Similarly, the conclusion also holds for $x_{n-2}\neq y_{n-2}$, $x_{n-1}\neq x_{n-2}$, and $x_3\neq y_3$ by reversing $\mathbf{x}$ and $\mathbf{y}$.
    \item \textbf{Subcase (d):} For the case where $x_i \neq y_i$ (for $i=1, n$), $x_j \neq x_{j+1}$ and $y_j \neq y_{j+1}$ (for $j=1, n-1$), Lemma $\ref{lm11}$ states that if $x_3\neq  y_3, x_2\neq x_3$, and $x_{n-2}=y_{n-2}$, then $|D_4(\mathbf{x})\cap D_4(\mathbf{y})|\leq 20n-166$. Similarly, the conclusion also holds for $x_{n-2}\neq y_{n-2}$, $x_{n-1}\neq x_{n-2}$, and $x_3=y_3$ by reversing $\mathbf{x}$ and $\mathbf{y}$. In summary, for the case where $x_i \neq y_i$ (for $i=1, n$), $x_j \neq x_{j+1}$ and $y_j \neq y_{j+1}$ (for $j=1, n-1$), Lemmas $\ref{lm9}$, $\ref{lm10}$, $\ref{lm11}$ together demonstrate that if $x_3\neq y_3$ or $x_{n-2}\neq y_{n-2}$ then $|D_4(\mathbf{x})\cap D_4(\mathbf{y})|\leq 20n-166$.
    \item  \textbf{Subcase (e):} For the case where $x_i \neq y_i$ (for $i = 1, n$), $x_j \neq x_{j+1}$ and $y_j \neq y_{j+1}$ (for $j = 1, n-1$), and with $x_l = y_l$ for $l = 3, n-2$, Lemma $\ref{lm12}$ shows that if $x_4\neq y_4$ or $x_4=y_4=x_3$, then $|D_4(\mathbf{x})\cap D_4(\mathbf{y})|\leq 20n-166$. Similarly, the conclusion also holds for $x_{n-3} \neq y_{n-3}$ or $x_{n-3} = y_{n-3} = x_{n-2}$ by reversing $\mathbf{x}$ and $\mathbf{y}$.
    \item  \textbf{Subcase (f):} For the subcase where $x_i \neq y_i$ (for $i = 1, n$), $x_j \neq x_{j+1}$ and $y_j \neq y_{j+1}$ (for $j = 1, 3, n-3, n-1$), and $x_l = y_l$ (for $l = 3, 4, n-3, n-2$),  Lemma $\ref{lm17}$ shows that $|D_4(\mathbf{x})\cap D_4(\mathbf{y})|\leq 20n-166$.
\end{enumerate}

In order to prove Lemmas  $\ref{lm8}$, $\ref{lm9}$, $\ref{lm10}$, $\ref{lm11}$, $\ref{lm12}$, $\ref{lm17}$, we need some lemmas (i.e., Lemmas $\ref{lm3}$-$\ref{lm7}$) as follows.
\begin{lemma}
\label{lm3}
Let $\mathbf{x}\in \mathbb{F}_2^{n_1}$ and $\mathbf{y}\in \mathbb{F}_2^{n_2}$. Then we have the following results:
\begin{enumerate}
    \item If $n_1=n_2+1$ and $\mathbf{y}\notin D_1(\mathbf{x})$, then $|D_2(\mathbf{x})\cap D_1(\mathbf{y})|\leq 3$, and $|D_3(\mathbf{x})\cap D_2(\mathbf{y})|\leq 3n_2-8$ for any $n_2\geq 4$.
    \item If $n_1=n_2+2$ and $\mathbf{y}\notin D_2(\mathbf{x})$, then $|D_3(\mathbf{x})\cap D_1(\mathbf{y})|\leq 4$, and $|D_4(\mathbf{x})\cap D_2(\mathbf{y})|\leq 4n_2-13$ for any $n_2\geq 5$.
\end{enumerate}
\end{lemma}

The proof of Lemma $\ref{lm3}$ is given in Appendix \ref{APP-C}.

\begin{lemma}[{\cite[Proposition 16]{Chrisnata}}]
\label{lm4}
Let $\mathbf{x}=\mathbf{u}\,\mathbf{v}\,\mathbf{w}$ and $\mathbf{y}=\mathbf{u}\,\mathbf{\widetilde{v}}\,\mathbf{w}$ such that $\mathbf{u},\mathbf{v},\mathbf{w},\mathbf{\widetilde{v}}$ are binary sequences. Then
\begin{equation}
D_t(\mathbf{x})\cap D_s(\mathbf{y})=\bigcup \limits_{p+q\leq \min\{t,s\}} D_p(\mathbf{u})\circ \big(D_{t-p-q}(\mathbf{v})\cap D_{s-p-q}(\mathbf{\widetilde{v}})\big)\circ D_q(\mathbf{w}).\nonumber
\end{equation}
\end{lemma}

\begin{lemma}[{\cite[Lemma 23]{Chrisnata}}]
\label{lm5}
Let $\mathbf{x}=\mathbf{u}\,\mathbf{v}\,\mathbf{w}$ and $\mathbf{y}=\mathbf{u}\,\mathbf{\widetilde{v}}\,\mathbf{w}$ such that $\mathbf{u},\mathbf{v},\mathbf{w},\mathbf{\widetilde{v}}$ are binary sequences. If $|\mathbf{v}|=|\mathbf{\widetilde{v}}|$ and $d_L(\mathbf{v},\mathbf{\widetilde{v}})\geq t$, then
\begin{equation}
D_t(\mathbf{x})\cap D_t(\mathbf{y})=\mathbf{u}\circ \big(D_{t}(\mathbf{v})\cap D_{t}(\mathbf{\widetilde{v}})\big)\circ \mathbf{w}.\nonumber
\end{equation}
\end{lemma}

\begin{lemma}
\label{lm6}
Let $n\geq 10$ and let $\mathbf{x,y}\in \mathbb{F}_2^n$ with $\mathbf{x}=\mathbf{u}\,\mathbf{v}\,\mathbf{w}$ and  $\mathbf{y}=\mathbf{u}\,\mathbf{\widetilde{v}}\,\mathbf{w}$, $|\mathbf{u}|=s$, $|\mathbf{w}|=t$, $|\mathbf{v}|=|\mathbf{\widetilde{v}}|=l$, $l\geq2$, and $s+t\geq 1$, where $d_L(\mathbf{x},\mathbf{y})\geq 2$, $\mathbf{u}$ and $\mathbf{w}$ are the longest common prefix and suffix of $\mathbf{x}$ and $\mathbf{y}$, respectively. When $s,t\geq 1$, we have \begin{equation}
|D_3(\mathbf{x})\cap D_3(\mathbf{y})|\leq 6n-34.\nonumber
\end{equation}
Specifically, when $s,t\geq 1$, we have the following results:
\begin{enumerate}
    \item If $l=6$ then $|D_3(\mathbf{x})\cap D_3(\mathbf{y})|\leq 6n-34$. Moreover, the equality holds only when $\mathbf{u}$, $\mathbf{w}$ are alternating sequences, $d_L(\mathbf{v},\mathbf{\widetilde{v}})=2$, and $|D_3(\mathbf{v})\cap D_3(\mathbf{\widetilde{v}})|=8$;
    \item If $l=7$ then $|D_3(\mathbf{x})\cap D_3(\mathbf{y})|\leq 6n-35$.  Moreover, the equality holds only when $\mathbf{u}$, $\mathbf{w}$ are alternating sequences, $d_L(\mathbf{v},\mathbf{\widetilde{v}})=2$, and $|D_3(\mathbf{v})\cap D_3(\mathbf{\widetilde{v}})|=13$;
    \item  If $l\geq 8$ then $|D_3(\mathbf{x})\cap D_3(\mathbf{y})|\leq 6n-36$;
    \item  If $l=5$ then $|D_3(\mathbf{x})\cap D_3(\mathbf{y})|\leq 4n-18$;
    \item  If $l=4$ then $|D_3(\mathbf{x})\cap D_3(\mathbf{y})|\leq 4n-16$ for any $n\geq 11$; When $n=10$, we have $|D_3(\mathbf{x})\cap D_3(\mathbf{y})|\leq 22$;
    \item  If $l=3$ then $|D_3(\mathbf{x})\cap D_3(\mathbf{y})|\leq 2n-5$;
    \item  If $l=2$ then $|D_3(\mathbf{x})\cap D_3(\mathbf{y})|\leq n-2$.
\end{enumerate}
Consider $s=0, t\geq 1$ or $s\geq 1, t=0$, we have $|D_3(\mathbf{x})\cap D_3(\mathbf{y})|\leq 6n-31$; if $l=7$ then we have $|D_3(\mathbf{x})\cap D_3(\mathbf{y})|\leq 6n-32$; if $l\geq 8$ then we have $|D_3(\mathbf{x})\cap D_3(\mathbf{y})|\leq 6n-33$. When $s,t\geq 1$, if $\mathbf{u}$ and $\mathbf{w}$ are both not alternating sequences then $|D_3(\mathbf{x})\cap D_3(\mathbf{y})|\leq 6n-40$. Furthermore, when $s=0, t\geq 1$ or $s\geq 1, t=0$, if $\mathbf{u}$ or $\mathbf{w}$ is not an alternating sequence then $|D_3(\mathbf{x})\cap D_3(\mathbf{y})|\leq 6n-37$.
\end{lemma}

\begin{lemma}
\label{lm7}
Let $n\geq 9$ and let $\mathbf{x}\in \mathbb{F}_2^n,\mathbf{y}\in \mathbb{F}_2^{n+2}$ with $\mathbf{x}=\mathbf{u}\,\mathbf{v}\,\mathbf{w}$, $\mathbf{y}=\mathbf{u}\,\mathbf{\widetilde{v}}\,\mathbf{w}$, $|\mathbf{u}|=s$, $|\mathbf{w}|=t$,
$|\mathbf{v}|=l, |\mathbf{\widetilde{v}}|=l+2$, $l\geq2$, and $s+t\geq 1$, where $\mathbf{x}\notin D_2(\mathbf{y})$, $\mathbf{u}$ and $\mathbf{w}$ are the longest common prefix and suffix of $\mathbf{x}$ and $\mathbf{y}$, respectively. When $s,t\geq 1$, we have
\begin{equation}
|D_2(\mathbf{x})\cap D_4(\mathbf{y})|\leq 4n-14.\nonumber
\end{equation}
Specifically, when $s,t\geq 1$, we have the following results:
\begin{enumerate}
    \item  If $l=4$ then $|D_2(\mathbf{x})\cap D_4(\mathbf{y})|\leq 4n-14$. Moreover, the equality holds only when $l=4$;
    \item  If $l\geq 5$ then $|D_2(\mathbf{x})\cap D_4(\mathbf{y})|\leq 4n-15$;
    \item  If $l=3$ then $|D_2(\mathbf{x})\cap D_4(\mathbf{y})|\leq 3n-7$;
    \item  If $l=2$ then $|D_2(\mathbf{x})\cap D_4(\mathbf{y})|\leq 2n-3$.
\end{enumerate}
Consider $s=0, t\geq 1$ or $s\geq 1, t=0$, if $l\geq 5$ then we have $|D_2(\mathbf{x})\cap D_4(\mathbf{y})|\leq 4n-14$. When $s,t\geq 1$, if  $\mathbf{u}$ and $\mathbf{w}$ are both not alternating sequences then we have that $|D_2(\mathbf{x})\cap D_4(\mathbf{y})|\leq 4n-17$ for any $l\geq 5$, $|D_2(\mathbf{x})\cap D_4(\mathbf{y})|\leq 4n-16$ for $l=4$. Furthermore, when $s=0, t\geq 1$ or $s\geq 1, t=0$, if  $\mathbf{u}$ or $\mathbf{w}$ is not an alternating sequence then $|D_2(\mathbf{x})\cap D_4(\mathbf{y})|\leq 4n-16$ for any $l\geq 5$, $|D_2(\mathbf{x})\cap D_4(\mathbf{y})|\leq 4n-15$ for $l=4$.
\end{lemma}

The proof of Lemmas $\ref{lm6}$ and $\ref{lm7}$ is given in Appendix \ref{APP-D}.

In the following, we will give the proofs for all subcases $(a)$–$(f)$.
\subsection{The proof of the subcase $(a)$}
Consider the subcase with $x_1 \neq y_1$, $x_n \neq y_n$, and suppose that at least one of the four equalities $x_1=x_2$, $x_{n-1}=x_n$, $y_1=y_2$, or $y_{n-1}=y_n$ holds. 

\begin{lemma}
\label{lm8}
Let $n\geq 13$ and let $\mathbf{x}=x_1\,\cdots\,x_n, \mathbf{y}=y_1\,\cdots\,y_n\in \mathbb{F}_2^n$ with $x_1=x_2, x_1\neq y_1, x_n\neq y_n$. If $d_L(\mathbf{x},\mathbf{y})\geq 3$, then
\begin{equation}
|D_4(\mathbf{x})\cap D_4(\mathbf{y})|\leq 20n-166.\nonumber
\end{equation}
\end{lemma}
\begin{proof}
For convenience, we let $x_1=1, x_n=a,y_n=\overline{a}$ with $a\in \mathbb{F}_2$. Then we have $\mathbf{x}=1\,1\,x_3\,\cdots\,x_{n-1}\,a$ and  $\mathbf{y}=0\,y_2\,y_3\,\cdots\,y_{n-1}\,\overline{a}$. We denote $\mathcal{X}=D_4(\mathbf{x})\cap D_4(\mathbf{y})$. By Proposition $\ref{prp1}$, we have $|\mathcal{X}|=|\mathcal{X}^{0}_{a}|+|\mathcal{X}^{0}_{\overline{a}}|+|\mathcal{X}^{1}_{a}|+|\mathcal{X}^1_{\overline{a}}|$.

First, we compute the value of $|\mathcal{X}^{0}_{a}|$. By Proposition $\ref{prp2}$, we have
$$\mathcal{X}^{0}_{a}=0\circ \big(D_{2-\ell}(x_{4+\ell}\,\cdots\,x_{n-1})\cap D_{3-\ell^*}(y_2\,\cdots\,y_{n-2-\ell^*})\big) \circ a,$$
where $\ell,\ell^*\geq 0$ are integers such that $x_{3+l}$ is the leftmost element $0$ of $\mathbf{x}$ and $y_{n-1-\ell^*}$ is the rightmost element $a$ of $\mathbf{y}$. Specifically, we discuss the value of $|\mathcal{X}^{0}_{a}|$ based on the different values of $\ell$ and $\ell^*$.

\begin{itemize}
\item When $\ell=0,\ell^*=0$, then $\mathbf{x}=1\,1\,0\,x_4\,\cdots\,x_{n-1}\,a$ and $\mathbf{y}=0\,y_2\,\cdots\,y_{n-2}\,a\,\overline{a}$. If $x_4\,\cdots\,x_{n-1}\in D_1(y_2\,\cdots\,y_{n-2})$ then $d_L(\mathbf{x},\mathbf{y})\leq 2$. Thus, $x_4\,\cdots\,x_{n-1}\notin D_1(y_2\,\cdots\,y_{n-2})$. By Lemma $\ref{lm3}$, it follows that $|D_2(x_4\,\cdots\,x_{n-1})\cap D_3(y_2\,\cdots\,y_{n-2})|\leq 3(n-4)-8=3n-20.$
\item When $\ell=0,\ell^*=1$, we easily obtain $x_4\,\cdots\,x_{n-1}\neq y_2\,\cdots\,y_{n-3}$. By Theorem $\ref{thm1}$, we have $|D_2(x_4\,\cdots\,x_{n-1})\cap D_2(y_2\,\cdots\,y_{n-3})|\leq N(n-4,1,2)=2(n-6)=2n-12.$
\item When $\ell\geq 1$ or $\ell^*\geq 2$ then $2-\ell\leq 1$ or $3-\ell^*\leq 1$. Thus,  $|\mathcal{X}^{0}_{a}|\leq \min\{|D_1(x_5\,\cdots\,x_{n-1})|,|D_1(y_2\,\cdots\,y_{n-4})|\}\leq (n-5).$
\end{itemize}
Therefore, for any $n\geq 13$, by the above discussion it follows that
\begin{equation}
|\mathcal{X}^{0}_{a}|\leq 3n-20.\nonumber
\end{equation}

Second, we estimate the value of $|\mathcal{X}^{0}_{\overline{a}}|$. By Proposition $\ref{prp2}$, we have
$$|\mathcal{X}^{0}_{\overline{a}}|=|0\circ \big(D_{1-\ell-\ell^*}(x_{4+\ell}\,\cdots\,x_{n-2-\ell^*})\cap D_{4}(y_2\,\cdots\,y_{n-1})\big)\circ \overline{a}|\leq |D_1(x_{4}\,\cdots\,x_{n-2})|\leq (n-5),$$
where $\ell,\ell^*\geq 0$ are integers such that $x_{3+l}$ is the leftmost element $0$ of $\mathbf{x}$ and $x_{n-1-\ell^*}$ is the rightmost element $\overline{a}$ of $\mathbf{x}$.

Next, we compute the value of $|\mathcal{X}^{1}_{a}|$. By Proposition $\ref{prp2}$, we have
$$\mathcal{X}^{1}_{a}=1\circ \big(D_{4}(1\,x_{3}\,\cdots\,x_{n-1})\cap D_{2-\ell-\ell^*}(y_{3+\ell}\,\cdots\,y_{n-2-\ell^*})\big) \circ a,$$
where $\ell,\ell^*\geq 0$ are integers such that $y_{2+l}$ is the leftmost element $1$ of $\mathbf{y}$ and $y_{n-1-\ell^*}$ is the rightmost element $a$ of $\mathbf{y}$. Specifically, we discuss the value of $|\mathcal{X}^{1}_{a}|$ based on the different values of $\ell$ and $\ell^*$.
\begin{itemize}
\item When $\ell=0,\ell^*=0$ then $\mathbf{x}=1\,1\,x_3\,\cdots\,x_{n-1}\,a$ and $\mathbf{y}=0\,1\,y_3\,\cdots\,y_{n-2}\,a\,\overline{a}$. If $y_3\,\cdots\,y_{n-2}\in D_2(1\,x_3\,\cdots\,x_{n-1})$ then $d_L(\mathbf{x},\mathbf{y})\leq 2$. Thus, $y_3\,\cdots\,y_{n-2}\notin D_2(1\,x_3\,\cdots\,x_{n-1})$. By Lemma $\ref{lm3}$, it follows that $|D_{4}(1\,x_{3}\,\cdots\,x_{n-1})\cap D_{2}(y_{3}\,\cdots\,y_{n-2})|\leq 4(n-4)-13=4n-29.$
\item When $\ell+\ell^*\geq 1$, then we have $|\mathcal{X}^{1}_{a}|\leq |D_{2-\ell-\ell^*}(y_{3+\ell}\,\cdots\,y_{n-2-\ell^*})|\leq  (n-5).$
\end{itemize}
Therefore, for any $n\geq 13$, we have
\begin{equation}
|\mathcal{X}^{1}_{a}|\leq 4n-29.\nonumber
\end{equation}

Finally, we discuss the value of $|\mathcal{X}^{1}_{\overline{a}}|$. By Proposition $\ref{prp2}$, we have
$$\mathcal{X}^{1}_{\overline{a}}=1\circ \big(D_{3-\ell}(1\,x_{3}\,\cdots\,x_{n-2-\ell})\cap D_{3-\ell^*}(y_{3+\ell^*}\,\cdots\,y_{n-1})\big) \circ \overline{a},$$
where $\ell,\ell^*\geq 0$ are integers such that $y_{2+l^*}$ is the leftmost element $1$ of $\mathbf{y}$ and $x_{n-1-\ell}$ is the rightmost element $\overline{a}$ of $\mathbf{x}$. Specifically, we discuss the value of $|\mathcal{X}^{1}_{\overline{a}}|$ based on the different values of $\ell$ and $\ell^*$.
\begin{itemize}
\item When $\ell=0,\ell^*=0$ then $\mathbf{x}=1\,1\,x_3\,\cdots\,x_{n-2}\,\overline{a}\,a$ and $\mathbf{y}=0\,1\,y_3\,\cdots\,y_{n-1}\,\overline{a}$. If $D_1(1\,x_3\,\cdots\,x_{n-2})\cap D_1(y_3\,\cdots\,y_{n-1})\neq \emptyset$ then $d_L(\mathbf{x},\mathbf{y})\leq 2$. Thus, $d_L(1\,x_3\,\cdots\,x_{n-2}, y_3\,\cdots\,y_{n-1})\geq 2$. By Corollary $\ref{cor1}$, since $n-3\geq 8$ it follows that $|D_{3}(1\,x_{3}\,\cdots\,x_{n-2})\cap D_{3}(y_{3}\,\cdots\,y_{n-1})|\leq N(n-3,2,3)=6(n-3)-30=6n-48$.
\item When $\ell=0,\ell^*=1$ or $\ell=1,\ell^*=0$, by Lemma $\ref{lm3}$ we have $|\mathcal{X}^{1}_{\overline{a}}|\leq 3n-20.$
\item When $\ell=1,\ell^*=1$, by Theorem $\ref{thm1}$ we have $|\mathcal{X}^{1}_{\overline{a}}|=|D_{2}(1\,x_{3}\,\cdots\,x_{n-3})\cap D_{2}(y_{4}\,\cdots\,y_{n-1})|\leq N(n-4,1,2)=2n-12.$
\item When $\ell\geq 2$, or $\ell^*\geq 2$, it follows that $|\mathcal{X}^{1}_{\overline{a}}|\leq (n-5).$
\end{itemize}
Therefore, for any $n\geq 13$, it follows that
\begin{equation}
|\mathcal{X}^{1}_{\overline{a}}|\leq 6n-48.\nonumber
\end{equation}

Based on the above discussion, we have
\begin{align*}
|\mathcal{X}|&=|\mathcal{X}^{0}_{a}|+|\mathcal{X}^{0}_{\overline{a}}|+|\mathcal{X}^{1}_{a}|+|\mathcal{X}^1_{\overline{a}}|\leq (3n-20)+(n-5)+(4n-29)+(6n-48)=14n-102\leq 20n-166,
\end{align*}
for any $n\geq 13$. Thus, the lemma follows.
\end{proof}

\subsection{The proof of the subcase $(b)$}
Consider the subcase where $x_i \neq y_i$ (for $i=1, n$), $x_j \neq x_{j+1}$ and $y_j \neq y_{j+1}$ (for $j = 1, n-1$), and with either ($x_3 \neq y_3~ \text{and}~ x_2=x_3$) or ($x_{n-2} \neq y_{n-2}~ \text{ and }~ x_{n-1} = x_{n-2}$). 

\begin{lemma}
\label{lm9}
Let $n\geq 13$ and let $\mathbf{x}=x_1\,\cdots\,x_n, \mathbf{y}=y_1\,\cdots\,y_n\in \mathbb{F}_2^n$ with $x_1\neq y_1, x_n\neq y_n, x_1\neq x_2, y_1\neq y_2, x_{n-1}\neq x_n, 
y_{n-1}\neq y_n$, and $x_3\neq y_3$, $x_2=x_3$. If $d_L(\mathbf{x},\mathbf{y})\geq 3$, then
\begin{equation}
|D_4(\mathbf{x})\cap D_4(\mathbf{y})|\leq 20n-166.\nonumber
\end{equation}
\end{lemma}

\begin{proof}
Without loss of generality, we let $x_1=1$ and $x_n=a$ with $a\in \mathbb{F}_2$. Then, $\mathbf{x}=1\,0\,0\,x_4\,\cdots\,x_{n-2}\,\overline{a}\,a$, and $\mathbf{y}=0\,1\,1\,y_4\,\cdots\,y_{n-2}\,a\,\overline{a}$. We denote $\mathcal{X}=D_4(\mathbf{x})\cap D_4(\mathbf{y})$. Then, by Proposition $\ref{prp1}$, we have $|\mathcal{X}|=|\mathcal{X}^{0}_{a}|+|\mathcal{X}^{0}_{\overline{a}}|+|\mathcal{X}^{1}_{a}|+|\mathcal{X}^1_{\overline{a}}|$. In the following, we compute the values of $|\mathcal{X}^{i}_{j}|$ for all $i=0,1$ and $j=a,\overline{a}$, respectively.

First, we consider the value of $|\mathcal{X}^{0}_{a}|$. By Propositions $\ref{prp1}$ and $\ref{prp2}$, we have $|\mathcal{X}^{0}_{a}|=|\mathcal{X}^{01}_{a}|+|\mathcal{X}^{00}_{a}|$, $\mathcal{X}^{01}_{a}=0\,1\circ \big(D_{2-\ell}(x_{5+l}\,\cdots\,x_{n-2}\,\overline{a})\cap D_{3}(1\,\cdots\,y_{n-2})\big)\circ a,$ and $\mathcal{X}^{00}_{a}=0\,0\circ \big(D_{3}(x_{4}\,\cdots\,x_{n-2}\,\overline{a})\cap D_{1-\ell^*}(y_{5+\ell^*}\,\cdots\,y_{n-2})\big)\circ a,$ where $\ell,\ell^*$ are integers, $0\leq \ell\leq 2$, and $0\leq \ell^*\leq 1$.

If $\ell=0$ by Lemma $\ref{lm3}$ then $|\mathcal{X}^{01}_{a}|\leq 3n-23$ because of $(x_{5}\,\cdots\,x_{n-2}\,\overline{a})\notin D_1(1\,\cdots\,y_{n-2})$; if $\ell=1$ then $|\mathcal{X}^{01}_{a}|\leq |D_1(x_{6}\,\cdots\,x_{n-2}\,\overline{a})|\leq n-6$; if $\ell=2$ then $|\mathcal{X}^{01}_{a}|\leq 1$. Thus, for any $n\geq 13$ we have $|\mathcal{X}^{01}_{a}|\leq 3n-23$. Similarly, for any $n\geq 13$ we have $|\mathcal{X}^{00}_{a}|\leq n-6$. Therefore, for any $n\geq 13$, it follows that
\begin{equation}
|\mathcal{X}^{0}_{a}|\leq 4n-29.\nonumber
\end{equation}

Next, we compute the value of $|\mathcal{X}^{0}_{\overline{a}}|$. By Propositions $\ref{prp1}$ and $\ref{prp2}$, we have
 $|\mathcal{X}^{0}_{\overline{a}}|=|\mathcal{X}^{01}_{\overline{a}}|+|\mathcal{X}^{00}_{\overline{a}}|$,
$\mathcal{X}^{01}_{\overline{a}}=0\,1\circ \big(D_{1-\ell}(x_{5+l}\,\cdots\,x_{n-2})\cap D_{4}(1\,\cdots\,y_{n-2}\,a)\big) \circ \overline{a},$ and
$\mathcal{X}^{00}_{\overline{a}}=0\,0\circ \big(D_{2}(x_{4}\,\cdots\,x_{n-2})\cap D_{2-\ell^*}(y_{5+\ell^*}\,\cdots\,y_{n-2}\,a)\big) \circ \overline{a},$
where $\ell,\ell^*$ are integers, $0\leq \ell\leq 1$, and $0\leq \ell^*\leq 2$.

If $\ell=0$ then $|\mathcal{X}^{01}_{\overline{a}}|\leq n-6$;  if $\ell=1$ then $|\mathcal{X}^{01}_{\overline{a}}|\leq 1$. Thus, for any $n\geq 13$ we have $|\mathcal{X}^{01}_{\overline{a}}|\leq n-6$. By Theorem $\ref{thm1}$, if $\ell^*=0$ then we have $|\mathcal{X}^{00}_{\overline{a}}|=|D_{2}(x_{4}\,\cdots\,x_{n-2})\cap D_{2}(y_{5}\,\cdots\,y_{n-2}\,a)|\leq N(n-5,1,2)=2(n-7)=2n-14$. Therefore, for any $n\geq 13$, it follows that
\begin{equation}
|\mathcal{X}^{0}_{\overline{a}}|\leq 3n-20.\nonumber
\end{equation}

Moreover, we discuss the value of $|\mathcal{X}^{1}_{a}|$. By Propositions $\ref{prp1}$ and $\ref{prp2}$, we have $|\mathcal{X}^{1}_{a}|=|\mathcal{X}^{11}_{a}|+|\mathcal{X}^{10}_{a}|$.
$\mathcal{X}^{11}_{a}=1\,1\circ \big(D_{2-\ell}(x_{5+\ell}\,\cdots\,x_{n-2}\,\overline{a})\cap D_{2}(1\,y_4\,\cdots\,y_{n-2})\big)\circ a,$ and
$\mathcal{X}^{10}_{a}=1\,0\circ \big(D_{4}(0\,x_{4}\,\cdots\,x_{n-2}\,\overline{a})\cap D_{1-\ell^*}(y_{5+\ell^*}\,\cdots\,y_{n-2})\big)\circ a,$
where $\ell,\ell^*$ are integers, $0\leq \ell\leq 2$, and $0\leq \ell^*\leq 1$. Similarly, by using the method of computing the value of $|\mathcal{X}^{0}_{\overline{a}}|$, we have
\begin{equation}
|\mathcal{X}^{1}_{a}|\leq 3n-20.\nonumber
\end{equation}

Finally, Consider the value of $|\mathcal{X}^{1}_{\overline{a}}|$. By Propositions $\ref{prp1}$ and $\ref{prp2}$, we have
 $|\mathcal{X}^{1}_{\overline{a}}|=|\mathcal{X}^{11}_{\overline{a}}|+|\mathcal{X}^{10}_{\overline{a}}|$,
$\mathcal{X}^{11}_{\overline{a}}=1\,1\circ \big(D_{1-\ell}(x_{5+\ell}\,\cdots\,x_{n-2})\cap D_{3}(y_4\,\cdots\,y_{n-2}\,a)\big)\circ \overline{a},$ and
$\mathcal{X}^{10}_{\overline{a}}=1\,0\circ \big(D_{3}(0\,x_{4}\,\cdots\,x_{n-2})\cap D_{2-\ell^*}(y_{5+\ell^*}\,\cdots\,y_{n-2}\,a)\big)\circ \overline{a},$
where $\ell,\ell^*$ are integers, $0\leq \ell\leq 1$, and $0\leq \ell^*\leq 2$. Similarly, by using the method of computing the value of $|\mathcal{X}^{0}_{a}|$, we can obtain
\begin{equation}
|\mathcal{X}^{1}_{\overline{a}}|\leq 4n-29.\nonumber
\end{equation}

By the above discussion, we have
\begin{align*}
|\mathcal{X}|=|\mathcal{X}^{0}_{a}|+|\mathcal{X}^{0}_{\overline{a}}|+|\mathcal{X}^{1}_{a}|+|\mathcal{X}^1_{\overline{a}}|\leq 2(3n-20)+2(4n-29)=14n-98\leq 20n-166,
\end{align*}
for any $n\geq 13$. So, the lemma follows.
\end{proof}

\subsection{The proof of the subcase $(c)$}
Consider the subcase with $x_i \neq y_i$ (for $i = 1, n$), $x_j \neq x_{j+1}$ and $y_j \neq y_{j+1}$ (for $j=1, n-1$), and with either $(x_3 \neq y_3, x_2 \neq x_3, \text{ and } x_{n-2} \neq y_{n-2}$) or ($x_{n-2} \neq y_{n-2}, x_{n-1} \neq x_{n-2}, \text{ and } x_3 \neq y_3$). 

 

\begin{lemma}
\label{lm10}
Let $n\geq 13$ and let $\mathbf{x}=x_1\,\cdots\,x_n, \mathbf{y}=y_1\,\cdots\,y_n\in \mathbb{F}_2^n$ with $x_1\neq y_1, x_n\neq y_n, x_1\neq x_2, x_{n-1}\neq x_n, y_1\neq y_2, y_{n-1}\neq y_n$, and $x_3\neq y_3$, $x_2\neq x_3$, $x_{n-2}\neq y_{n-2}$, $x_{n-2}\neq x_{n-1}$. If $d_L(\mathbf{x},\mathbf{y})\geq 3$, then
\begin{equation}
|D_4(\mathbf{x})\cap D_4(\mathbf{y})|\leq 20n-166.\nonumber
\end{equation}
\end{lemma}

\begin{proof}
Without loss of generality, we let $x_1=1$ and $x_n=a$ with $a\in \mathbb{F}_2$. Then, $\mathbf{x}=1\,0\,1\,x_4\,\cdots\,x_{n-3}\,a\,\overline{a}\,a$, and $\mathbf{y}=0\,1\,0\,y_4\,\cdots\,y_{n-3}\,\overline{a}\,a\,\overline{a}$. We denote $\mathcal{X}=D_4(\mathbf{x})\cap D_4(\mathbf{y})$. Then, by Proposition $\ref{prp1}$, we have $|\mathcal{X}|=|\mathcal{X}^{0}_{a}|+|\mathcal{X}^{0}_{\overline{a}}|+|\mathcal{X}^{1}_{a}|+|\mathcal{X}^1_{\overline{a}}|$.

In the following, we compute the values of $|\mathcal{X}^{i}_{j}|$ for all $i=0,1$ and $j=a,\overline{a}$, respectively. By Proposition $\ref{prp2}$, we have
\begin{align*}
\mathcal{X}^{0}_{a}&=0\circ (D_3(\mathbf{x}^{0}_{a})\cap D_3(\mathbf{y}^{0}_{a}))\circ a=0\circ \big(D_{3}(1\,x_{4}\,\cdots\,x_{n-3}\,a\,\overline{a})\cap D_{3}(1\,0\,y_4\,\cdots\,y_{n-3}\,\overline{a})\big)\circ a,\\
\mathcal{X}^{0}_{\overline{a}}&=0\circ (D_2(\mathbf{x}^{0}_{\overline{a}})\cap D_4(\mathbf{y}^{0}_{\overline{a}}))\circ \overline{a}=0 \big(D_{2}(1\,x_{4}\,\cdots\,x_{n-3}\,a)\cap D_{4}(1\,0\,y_4\,\cdots\,y_{n-2}\,\overline{a}\,a)\big) \circ \overline{a},\\
\mathcal{X}^{1}_{a}&=1\circ (D_4(\mathbf{x}^{1}_{a})\cap D_2(\mathbf{y}^{1}_{a}))\circ a=1\circ \big(D_{4}(0\,1\,x_{4}\,\cdots\,x_{n-3}\,a\,\overline{a})\cap D_{2}(0\,y_4\,\cdots\,y_{n-3}\,\overline{a})\big)\circ a,\\
\mathcal{X}^{1}_{\overline{a}}&=1\circ (D_3(\mathbf{x}^{1}_{\overline{a}})\cap D_3(\mathbf{y}^{1}_{\overline{a}}))\circ \overline{a}=1\circ \big(D_{3}(0\,1\,x_{4}\,\cdots\,x_{n-3}\,a)\cap D_{3}(0\,y_4\,\cdots\,y_{n-3}\,\overline{a}\,a)\big)\circ \overline{a},
\end{align*}
where $d_L(\mathbf{x}^{0}_{a},\mathbf{y}^{0}_{a})\geq 2,\mathbf{x}^{0}_{\overline{a}}\notin D_2(\mathbf{y}^{0}_{\overline{a}}),\mathbf{y}^{1}_{a}\notin D_2(\mathbf{x}^{1}_{a}),$ and $d_L(\mathbf{x}^{1}_{\overline{a}},\mathbf{y}^{1}_{\overline{a}})\geq 2$.

Since $n\geq 13$ we have $|\mathbf{x}^{0}_{a}|=|\mathbf{x}^{1}_{\overline{a}}|\geq 10$ and $|\mathbf{x}^{0}_{\overline{a}}|=|\mathbf{y}^{1}_{a}|\geq 9$. By Lemmas $\ref{lm6}$ and $\ref{lm7}$, we have
$$|\mathcal{X}^{0}_{a}|=|D_3(\mathbf{x}^{0}_{a})\cap D_3(\mathbf{y}^{0}_{a})|\overset{(I_1)}{\leq} 6(n-3)-34=6n-52,$$
$$|\mathcal{X}^{0}_{\overline{a}}|=|D_2(\mathbf{x}^{0}_{\overline{a}})\cap D_4(\mathbf{y}^{0}_{\overline{a}})|\overset{(I_2)}{\leq} 4(n-4)-14=4n-30,$$
$$|\mathcal{X}^{1}_{a}|=|D_4(\mathbf{x}^{1}_{a})\cap D_2(\mathbf{y}^{1}_{a})|\overset{(I_3)}{\leq} 4(n-4)-14=4n-30,$$
$$|\mathcal{X}^{1}_{\overline{a}}|=|D_3(\mathbf{x}^{1}_{\overline{a}})\cap D_3(\mathbf{y}^{1}_{\overline{a}})|\overset{(I_4)}{\leq} 6(n-3)-34=6n-52,$$
for any $n\geq 13$.

Thus, we have $|\mathcal{X}|=|\mathcal{X}^{0}_{a}|+|\mathcal{X}^{0}_{\overline{a}}|+|\mathcal{X}^{1}_{a}|+|\mathcal{X}^{1}_{\overline{a}}|\leq 20n-164$ for any $n\geq 13$. Assume that $|\mathcal{X}|> 20n-166$ for any $n\geq 13$. Then at least three of the above four inequalities (Inequations $(I_1)-(I_4)$) hold with equality. Thus, these three inequalities must contain Inequations $(I_1)$ and $(I_2)$, or Inequations $(I_3)$ and $(I_4)$. Since Inequations $(I_1)-(I_2)$ and Inequations $(I_3)-(I_4)$ are similar, thus we only consider one case of Inequations $(I_1)$ and $(I_2)$.

Consider the case of $|\mathcal{X}^{0}_{a}|=6n-52$ and $|\mathcal{X}^{0}_{\overline{a}}|=4n-30$. Since $|\mathcal{X}^{0}_{a}|=6n-52$, by Lemma $\ref{lm6}$, we have $\mathbf{x}^{0}_{a}=\mathbf{u}\,\mathbf{v}\,\mathbf{w},~\mathbf{y}^{0}_{a}=\mathbf{u}\,\mathbf{\widetilde{v}}\,\mathbf{w},$
where $|\mathbf{v}|=|\mathbf{\widetilde{v}}|=6$, $\mathbf{u}=u_1\,\cdots\,u_s=\mathbf{a}_{s}$, $\mathbf{w}=w_1\,\cdots\,w_t$, $s,t\geq 1$, $d_L(\mathbf{v},\mathbf{\widetilde{v}})=2$, and the first element and the last element of $\mathbf{v}$ and $\mathbf{\widetilde{v}}$ are both different,  $\mathbf{w}$ is an alternating sequence that ends with $\overline{a}$. For convenience, $\mathbf{v}=v_1\,v_2\,\cdots\,v_5\,v_6$ and $\mathbf{\widetilde{v}}=\widetilde{v}_1\,\widetilde{v}_2\,\cdots\,\widetilde{v}_5\,\widetilde{v}_6$.
Then, we have
$\mathbf{x}=1\,0\,\mathbf{a}_{s}\,\mathbf{v}\,\mathbf{w}\,a,~\mathbf{y}=0\,\mathbf{a}_{s}\,\mathbf{\widetilde{v}}\,\mathbf{w}\,a\,\overline{a},$
and
$\mathbf{x}_{[2,s+2]}=\mathbf{y}_{[1,s+1]}$, $x_{s+3}\neq y_{s+2}$, $x_{s+8}\neq y_{s+7}$, $\mathbf{x}_{[s+9,s+t+9]}=\mathbf{y}_{[s+8,s+t+8]}$,
where $s+t+9=n$. Next, we discuss the values of $\mathbf{x}^{0}_{\overline{a}}$ and $\mathbf{y}^{0}_{\overline{a}}$. By the definition of $\mathbf{x}^{0}_{\overline{a}}$ and $\mathbf{y}^{0}_{\overline{a}}$, it follows that
$\mathbf{x}^{0}_{\overline{a}}=\mathbf{u}\,\mathbf{v}\,\mathbf{w}_{[1,t-1]},~\mathbf{y}^{0}_{\overline{a}}=\mathbf{u}\,\mathbf{\widetilde{v}}\,\mathbf{w}\,a.$
Since $|\mathcal{X}^{0}_{\overline{a}}|=4n-30$, by Lemma $\ref{lm7}$,  we have $\mathbf{x}^{0}_{\overline{a}}=\mathbf{u}_1\,\mathbf{v}_1\,\mathbf{w}_1$ and $\mathbf{y}^{0}_{\overline{a}}=\mathbf{u}_1\,\mathbf{\widetilde{v}}_1\,\mathbf{w}_1$, where $|\mathbf{v}_1|=4,|\mathbf{\widetilde{v}}_1|=6$, $\mathbf{w}_{1
[t-1]}=a$, $\mathbf{v}_1\notin D_2(\mathbf{\widetilde{v}})$, and the first element and the last element of $\mathbf{v}_1$ and $\mathbf{\widetilde{v}}_1$ are both different. By comparison, we can obtain that
$$\mathbf{u}=\mathbf{u}_1, \mathbf{v}_1=\mathbf{v}_{[1,4]}, \mathbf{w}_1=\mathbf{v}_{[5,6]}\,\mathbf{w}_{[1,t-1]}, \mathbf{\widetilde{v}}_1=\mathbf{\widetilde{v}}, \mathbf{w}_1=\mathbf{w}\,a.$$
Thus, $v_5=\mathbf{w}_{[1]}$, $v_6=\mathbf{w}_{[2]}$, and $v_4\neq \widetilde{v}_6$. Next, we compare $\mathbf{x}^{1}_{\overline{a}}$ and $\mathbf{y}^{1}_{\overline{a}}$ as follows. By the definition of $\mathbf{x}^{1}_{\overline{a}}$ and $\mathbf{y}^{1}_{\overline{a}}$, then we have
\begin{align*}
\mathbf{x}^{1}_{\overline{a}}=0\,\mathbf{u}\,\mathbf{v}\,\mathbf{w}_{[1,t-1]}=0\,\mathbf{a}_s\,v_1\,v_2\,v_3\,v_4\,\mathbf{w}_1=0\,\mathbf{a}_{s-2}\,\mathbf{a}_{s[s-1]}\,
\mathbf{a}_{s[s]}\,v_1\,v_2\,v_3\,v_4\,\mathbf{w}_1,
\end{align*}
$$\mathbf{y}^{1}_{\overline{a}}=\mathbf{u}_{[2,s]}\,\mathbf{\widetilde{v}}\,\mathbf{w}\,a=0\,\mathbf{a}_{s-2}\,\widetilde{v}_1\,\cdots\,\widetilde{v}_5\,\widetilde{v}_6\,
\mathbf{w}_1.$$

When $\mathbf{a}_{s[s-1]}= \widetilde{v}_1$, then the first $s$ elements of $\mathbf{x}^{1}_{\overline{a}}$ and $\mathbf{y}^{1}_{\overline{a}}$ are same. For convenience, let $\mathbf{x}^{1}_{\overline{a}}=\mathbf{u}_2\,\mathbf{v}_2\,\mathbf{w}_2$ and $\mathbf{y}^{1}_{\overline{a}}=\mathbf{u}_2\,\mathbf{\widetilde{v}}_2\,\mathbf{w}_2$, where the first and last elements of $\mathbf{v}_2$ and $\mathbf{\widetilde{v}}_2$ are both different. Then $|\mathbf{v}_2|\leq 5$. By Lemma $\ref{lm6}$, we have that
$$|\mathcal{X}^{1}_{\overline{a}}|\leq 6(n-3)-36=6n-54.$$
Therefore, we have
$$|\mathcal{X}|=|\mathcal{X}^{0}_{a}|+|\mathcal{X}^{0}_{\overline{a}}|+|\mathcal{X}^{1}_{a}|+|\mathcal{X}^1_{\overline{a}}|\leq 6n-52+4n-30+4n-30+6n-54=20n-166.$$

When $\mathbf{a}_{s[s-1]}\neq \widetilde{v}_1$, we let $\mathbf{u}_1=0\,\mathbf{a}_{s-2}$ and
$\mathbf{v}_3=\mathbf{a}_{s[s-1]}\,\mathbf{a}_{s[s]}\,v_1\,v_2\,v_3\,v_4$, then we have
$$\mathbf{x}^{1}_{\overline{a}}=\mathbf{u}_1\,\mathbf{a}_{s[s-1]}\,\mathbf{a}_{s[s]}\,v_1\,v_2\,v_3\,v_4\,\mathbf{w}_1=\mathbf{u}_1\,\mathbf{v}_3\,\mathbf{w}_1,~
\mathbf{y}^{1}_{\overline{a}}=\mathbf{u}_1\,\widetilde{v}_1\,\cdots\,\widetilde{v}_5\,\widetilde{v}_6\,\mathbf{w}_1=\mathbf{u}_1\,\mathbf{\widetilde{v}}\,\mathbf{w}_1,$$
where $|\mathbf{v}_3|=6$, $d_L(\mathbf{v}_3,\mathbf{\widetilde{v}})=2$ because of $d_L(\mathbf{x},\mathbf{y})=3$. For convenience, let $\mathbf{u}_{[s]}=\mathbf{u}_{1[s-1]}=c$, then $\mathbf{a}_{s[s-1]}=\overline{c},\mathbf{a}_{s[s]}=c$ because of $0\,\mathbf{u}=\mathbf{u}_1\,\mathbf{a}_{s[s-1]}\,\mathbf{a}_{s[s]}=0\,\mathbf{a}_s$. Since $\mathbf{a}_{s[s-1]}\neq \widetilde{v}_1$, thus we have $\widetilde{v}_1=c.$
Furthermore, we discuss the values of $\mathbf{x}^{1}_{a}$ and $\mathbf{y}^{1}_{a}$ as follows. Then, we have
$$\mathbf{x}^{1}_{a}=0\,\mathbf{u}\,\mathbf{v}\,\mathbf{w}=\mathbf{u}_1\,\overline{c}\,c\,v_1\,v_2\,v_3\,v_4\,v_5\,v_6\,\mathbf{w},~\mathbf{y}^{1}_{a}=\mathbf{u}_{[2,s]}\,
\mathbf{\widetilde{v}}\,\mathbf{w}=\mathbf{u}_1\,\widetilde{v}_1\,\cdots\,\widetilde{v}_5\,\widetilde{v}_6\,\mathbf{w},$$
where $\widetilde{v}_1\neq \overline{c}$, $v_6\neq \widetilde{v}_6$, and $\mathbf{u}_{1[s-1]}=\widetilde{v}_1=c$. By the similar method of Lemma $\ref{lm7}$, we decompose $\mathbf{x}^{1}_{a}$ and $\mathbf{y}^{1}_{a}$, continue the process until $\mathbf{x}$ and $\mathbf{y}$ are transformed into $\overline{c}\,c\,\mathbf{v}$ and $\mathbf{\widetilde{v}}$, respectively. The maximum contribution of the first $s-2$ left decompositions and the later $t-1$ right decompositions to the value of $|\mathcal{X}^{1}_{\overline{a}}|$ is $4(s-2)$ and $4(t-1)$, respectively. By Lemma $\ref{lm7}$, the $(s-1)$-th left decomposition contributes only $1$ because of $\mathbf{u}_{1[s-1]}=\widetilde{v}_1=c$. The $t$-th right decomposition contributes only $3$ because of $v_6\neq \widetilde{v}_6$. Thus,
we have $|\mathcal{X}^{1}_{\overline{a}}|\leq 4(s-2)+1+|D_4(\overline{c},c,\mathbf{v})\cap D_2(\mathbf{\widetilde{v}})|+4(t-1)+3.$
By using a computerized search, we have $|D_4(\overline{c}\,c\,\mathbf{v})\cap D_2(\mathbf{\widetilde{v}})|\leq 11$. So, we can obtain that
$$|\mathcal{X}^{1}_{\overline{a}}|\leq 4(s-2)+1+11+4(t-1)+3=4(s+t)+3=4(n-9)+3=4n-33.$$
Therefore, we have
$$|\mathcal{X}|=|\mathcal{X}^{0}_{a}|+|\mathcal{X}^{0}_{\overline{a}}|+|\mathcal{X}^{1}_{a}|+|\mathcal{X}^1_{\overline{a}}|\leq 6n-52+4n-30+6n-52+4n-33=20n-167,$$
which completes the proof of Lemma $\ref{lm10}$.

\end{proof}

\subsection{The proof of the subcase $(d)$}
Consider the subcase with $x_i \neq y_i$ (for $i = 1, n$), $x_j \neq x_{j+1}$ and $y_j \neq y_{j+1}$ (for $j=1, n-1$), and with either $(x_3 \neq y_3, x_2 \neq x_3, \text{ and } x_{n-2}=y_{n-2}$) or ($x_{n-2} \neq y_{n-2}, x_{n-1} \neq x_{n-2}, \text{ and } x_3=y_3$). 

\begin{lemma}
\label{lm11}
Let $n\geq 13$ and let $\mathbf{x}=x_1\,\cdots\,x_n, \mathbf{y}=y_1\,\cdots\,y_n \in \mathbb{F}_2^n$ with $x_1\neq y_1, x_n\neq y_n, x_1\neq x_2, x_{n-1}\neq x_n, y_1\neq y_2, y_{n-1}\neq y_n$, and $x_3\neq y_3$, $x_2\neq x_3$, $x_{n-2}=y_{n-2}$. If $d_L(\mathbf{x},\mathbf{y})\geq 3$, then
\begin{equation}
|D_4(\mathbf{x})\cap D_4(\mathbf{y})|\leq 20n-166.\nonumber
\end{equation}
\end{lemma}
\begin{proof}
Without loss of generality, we let $x_1=1$, $x_{n-2}=b, x_n=a$ with $a,b\in \mathbb{F}_2$. Then, $\mathbf{x}=1\,0\,1\,x_4\,\cdots\,x_{n-3}\,b\,\overline{a}\,a$, and $\mathbf{y}=0\,1\,0\,y_4\,\cdots\,y_{n-3}\,b\,a\,\overline{a}$. We denote $\mathcal{X}=D_4(\mathbf{x})\cap D_4(\mathbf{y})$. Let $b=a$. Then we have $\mathbf{x}=1\,0\,1\,x_4\,\cdots\,x_{n-3}\,a\,\overline{a}\,a$ and  $\mathbf{y}=0\,1\,0\,y_4\,\cdots\,y_{n-3}\,a\,a\,\overline{a}$. Thus, $\overline{\mathbf{x}}=0\,1\,0\,\overline{x_4}\,\cdots\,\overline{x_{n-3}}\,\overline{a}\,a\,\overline{a}$ and $\overline{\mathbf{y}}=1\,0\,1\,\overline{y_4}\,\cdots\,\overline{y_{n-3}}\,\overline{a}\,\overline{a}\,a$. It is easily verified that $|D_t(\mathbf{x})\cap D_t(\mathbf{y})|=|D_t(\overline{\mathbf{x}})\cap D_t(\overline{\mathbf{y}})|$. Assume that $b=\overline{a}$. Then we have  $\mathbf{x}=1\,0\,1\,x_4\,\cdots\,x_{n-3}\,\overline{a}\,\overline{a}\,a$ and  $\mathbf{y}=0\,1\,0\,y_4\,\cdots\,y_{n-3}\,\overline{a}\,a\,\overline{a}$. By comparison, we only discuss the case where $b=a$.

By Propositions $\ref{prp1}$ and $\ref{prp2}$, we have $|\mathcal{X}|=|\mathcal{X}^{0}_{a}|+|\mathcal{X}^{0}_{\overline{a}}|+|\mathcal{X}^{1}_{a}|+|\mathcal{X}^1_{\overline{a}}|$ and
\begin{align*}
\mathcal{X}^{0}_{a}&=0\circ (D_3(\mathbf{x}^{0}_{a})\cap D_3(\mathbf{y}^{0}_{a}))\circ a=0\circ \big(D_{3}(1\,x_{4}\,\cdots\,x_{n-3}\,a\,\overline{a})\cap D_{3}(1\,0\,y_4\,\cdots\,y_{n-3}\,a)\big)\circ a,\\
\mathcal{X}^{0}_{\overline{a}}&=0\circ (D_2(\mathbf{x}^{0}_{\overline{a}})\cap D_4(\mathbf{y}^{0}_{\overline{a}}))\circ \overline{a}=0 \big(D_{2}(1\,x_{4}\,\cdots\,x_{n-3}\,a)\cap D_{4}(1\,0\,y_4\,\cdots\,y_{n-3}\,a\,a)\big) \circ \overline{a},\\
\mathcal{X}^{1}_{a}&=1\circ (D_4(\mathbf{x}^{1}_{a})\cap D_2(\mathbf{y}^{1}_{a}))\circ a=1\circ \big(D_{4}(0\,1\,x_{4}\,\cdots\,x_{n-3}\,a\,\overline{a})\cap D_{2}(0\,y_4\,\cdots\,y_{n-3}\,a)\big)\circ a,\\
\mathcal{X}^{1}_{\overline{a}}&=1\circ (D_3(\mathbf{x}^{1}_{\overline{a}})\cap D_3(\mathbf{y}^{1}_{\overline{a}}))\circ \overline{a}=1\circ \big(D_{3}(0\,1\,x_{4}\,\cdots\,x_{n-3}\,a)\cap D_{3}(0\,y_4\,\cdots\,y_{n-3}\,a\,a)\big)\circ \overline{a},
\end{align*}
where $d_L(\mathbf{x}^{0}_{a},\mathbf{y}^{0}_{a})\geq 2,\mathbf{x}^{0}_{\overline{a}}\notin D_2(\mathbf{y}^{0}_{\overline{a}}),\mathbf{y}^{1}_{a}\notin D_2(\mathbf{x}^{1}_{a}),$ and $d_L(\mathbf{x}^{1}_{\overline{a}},\mathbf{y}^{1}_{\overline{a}})\geq 2$.

Since $n\geq 13$ we have $|\mathbf{x}^{0}_{a}|=|\mathbf{x}^{1}_{\overline{a}}|\geq 10$ and $|\mathbf{x}^{0}_{\overline{a}}|=|\mathbf{y}^{1}_{a}|\geq 9$. By Lemmas $\ref{lm6}$ and $\ref{lm7}$, we have
\begin{align}
|\mathcal{X}^{0}_{a}|=|D_3(\mathbf{x}^{0}_{a})\cap D_3(\mathbf{y}^{0}_{a})|\leq 6(n-3)-31=6n-49,\label{eq7}\\
|\mathcal{X}^{0}_{\overline{a}}|=|D_2(\mathbf{x}^{0}_{\overline{a}})\cap D_4(\mathbf{y}^{0}_{\overline{a}})|\leq 4(n-4)-14=4n-30,\label{eq8}\\
|\mathcal{X}^{1}_{a}|=|D_4(\mathbf{x}^{1}_{a})\cap D_2(\mathbf{y}^{1}_{a})|\leq 4(n-4)-13=4n-29,\label{eq9}\\
|\mathcal{X}^{1}_{\overline{a}}|=|D_3(\mathbf{x}^{1}_{\overline{a}})\cap D_3(\mathbf{y}^{1}_{\overline{a}})|\leq 6(n-3)-34=6n-52,\label{eq10}
\end{align}
for any $n\geq 13$.

Next, we will discuss the value of $|\mathcal{X}|$ based on the values of $x_4,y_4,x_{n-3}$, and $y_{n-3}$. Specifically, we will address this in five cases as follows: \textbf{Case 1}: $x_{n-3}=a$; \textbf{Case 2}: $x_{n-3}=\overline{a},y_{n-3}=a$; \textbf{Case 3}: $x_{n-3}=y_{n-3}=\overline{a},x_4=1$; \textbf{Case 4}: $x_{n-3}=y_{n-3}=\overline{a},x_4=0, y_4=0$; \textbf{Case 5}: $x_{n-3}=y_{n-3}=\overline{a},x_4=0,y_4=1$.

\textbf{Case 1}: When $x_{n-3}=a$, since the last two elements of $\mathbf{x}^{1}_{\overline{a}}$ and $\mathbf{y}^{1}_{\overline{a}}$ are both $(a,a)$, then by Lemma $\ref{lm6}$ we have $$|\mathcal{X}^{1}_{\overline{a}}|\leq 6(n-3)-40=6n-58.$$
Combing with Inequations $(\ref{eq7})-(\ref{eq9})$, we have $|\mathcal{X}|= |\mathcal{X}^{0}_{a}|+|\mathcal{X}^{0}_{\overline{a}}|+|\mathcal{X}^{1}_{a}|+|\mathcal{X}^{1}_{\overline{a}}|\leq 6n-49+4n-30+4n-29+6n-58=20n-166$, for any $n\geq 13$.

\textbf{Case 2}: When $x_{n-3}=\overline{a},y_{n-3}=a$, we discuss the value of $|\mathcal{X}^{0}_{a}|$ as follows. First we have $|\mathcal{X}^{0}_{a}|=|\mathcal{X}^{0}_{aa}|+|\mathcal{X}^{0}_{a\overline{a}}|$ such that
\begin{align*}
|\mathcal{X}^{0}_{aa}|&=|D_2(\mathbf{x}^{0}_{aa})\cap D_3(\mathbf{y}^{0}_{aa})|=|D_{2}(1\,x_{4}\,\cdots\,x_{n-4}\,\overline{a})\cap D_{3}(1\,0\,y_4\,\cdots\,y_{n-4}\,a)|,\\
|\mathcal{X}^{0}_{a\overline{a}}|&=|D_3(\mathbf{x}^{0}_{a\overline{a}})\cap D_3(\mathbf{y}^{0}_{a\overline{a}})|=|D_{3}(1\,x_{4}\,\cdots\,x_{n-4}\,\overline{a}\,a)\cap D_{1-\ell}(1\,0\,y_4\,\cdots\,y_{n-5-l})|,
\end{align*}
where $\mathbf{x}^{0}_{aa}\notin D_1(\mathbf{y}^{0}_{aa})$ and $\ell\geq 0$. By Lemma $\ref{lm3}$, we have $|\mathcal{X}^{0}_{aa}|\leq 3(n-5)-8=3n-23$. If $\ell\geq 0$ then $|\mathcal{X}^{0}_{aa}|\leq |D_1(1\,0\,y_4\,\cdots\,y_{n-5})|\leq n-6$. Thus, $|\mathcal{X}^{0}_{a}|\leq 4n-29$ for any $n\geq 13$. Combing with Inequations $(\ref{eq8})$-$(\ref{eq10})$, we have $|\mathcal{X}|= |\mathcal{X}^{0}_{a}|+|\mathcal{X}^{0}_{\overline{a}}|+|\mathcal{X}^{1}_{a}|+|\mathcal{X}^{1}_{\overline{a}}|\leq 4n-29+4n-30+4n-29+6n-58=18n-140\leq 20n-166$, for any $n\geq 13$.

\textbf{Case 3}: When $x_{n-3}=y_{n-3}=\overline{a},x_4=1$, we discuss the values of $|\mathcal{X}^{0}_{a}|$ and $|\mathcal{X}^{0}_{\overline{a}}|$ as follows. Since the second and last bits of $\mathbf{x}^{0}_{a}$ and $\mathbf{y}^{0}_{a}$ are different and $n\geq 13$, by Lemma $\ref{lm6}$, we have
$$|\mathcal{X}^{0}_{a}|\leq 6(n-3)-33=6n-51.$$
Since the first and second bits of $\mathbf{x}^{0}_{\overline{a}}$ are $1$, we have $|\mathcal{X}^{00}_{\overline{a}}|\leq 1$. Then, by Lemma $\ref{lm7}$, it follows that
$$|\mathcal{X}^{0}_{\overline{a}}|=|\mathcal{X}^{00}_{\overline{a}}|+|\mathcal{X}^{01}_{\overline{a}}|\leq 1+4(n-5)-14=4n-33,$$
for any $n\geq 13$. Combing with Inequations $(\ref{eq9})$-$(\ref{eq10})$, we have $|\mathcal{X}|= |\mathcal{X}^{0}_{a}|+|\mathcal{X}^{0}_{\overline{a}}|+|\mathcal{X}^{1}_{a}|+|\mathcal{X}^{1}_{\overline{a}}|\leq 6n-51+4n-33+4n-29+6n-52=20n-165$, for any $n\geq 13$. Assume that $|\mathcal{X}|> 20n-166$ for any $n\geq 13$. Then the above two inequalities (Inequations $(\ref{eq9})$-$(\ref{eq10})$) hold with equality. Equation $(\ref{eq9})$ holds only when the condition of $l$ in Lemma $\ref{lm7}$ satisfies $l=4$. Thus, $y_4=y_5=1$ which makes the left common subsequence of $\mathbf{x}^{1}_{\overline{a}}$ and $\mathbf{y}^{1}_{\overline{a}}$ is not a alternating sequence. Therefore, the above two inequalities cannot hold simultaneously. So, $|\mathcal{X}|\leq 20n-166$ for any $n\geq 13$.

\textbf{Case 4}: When $x_{n-3}=y_{n-3}=\overline{a},x_4=y_4=0$, we discuss the values of $|\mathcal{X}^{0}_{a}|$, $|\mathcal{X}^{0}_{\overline{a}}|$, and $|\mathcal{X}^{1}_{\overline{a}}|$ as follows. Since $y_4=0$ then the second bits of $\mathbf{x}^{1}_{\overline{a}}$ and $\mathbf{y}^{1}_{\overline{a}}$ are not same. By Lemma $\ref{lm6}$, we have
$$|\mathcal{X}^{1}_{\overline{a}}|\leq 6(n-3)-36=6n-54,$$
for any $n\geq 13$.

When $x_5=1$, then the third bits of $\mathbf{x}^{0}_{a}$ and $\mathbf{y}^{0}_{a}$ are not same which makes $l\geq 8$ because of $n\geq 13$. By Lemma $\ref{lm6}$, we have
$$|\mathcal{X}^{0}_{a}|\leq 6(n-3)-33=6n-51,$$
for any $n\geq 13$. 
Since $y_4=0$ and $y_{n-3}=\overline{a}$ then the first and second bits of $\mathbf{y}^{1}_{a}$ are same and the condition of $l\geq 5$ in Lemma $\ref{lm7}$ holds. Thus, 
$$|\mathcal{X}^{1}_{a}|\leq 4(n-4)-16=4n-32,$$
for any $n\geq 13$. 
Combing with Inequation $(\ref{eq9})$, we have $|\mathcal{X}|= |\mathcal{X}^{0}_{a}|+|\mathcal{X}^{0}_{\overline{a}}|+|\mathcal{X}^{1}_{a}|+|\mathcal{X}^{1}_{\overline{a}}|\leq 6n-51+4n-30+4n-32+6n-54=20n-167$ for any $n\geq 13$.

When $x_5=0$, then the left common subsequence of $\mathbf{x}^{0}_{a}$ and $\mathbf{y}^{0}_{a}$ is not an alternating sequence. By Lemma $\ref{lm6}$, we have
$$|\mathcal{X}^{0}_{a}|\leq 6(n-3)-40=6n-58,$$
for any $n\geq 13$.
Combing with Inequations $(\ref{eq8})$-$(\ref{eq9})$, we have $|\mathcal{X}|= |\mathcal{X}^{0}_{a}|+|\mathcal{X}^{0}_{\overline{a}}|+|\mathcal{X}^{1}_{a}|+|\mathcal{X}^{1}_{\overline{a}}|\leq 6n-58+4n-30+4n-29+6n-54=20n-171$ for any $n\geq 13$.

\textbf{Case 5}: When $x_{n-3}=y_{n-3}=\overline{a},x_4=0, y_4=1$, we discuss all the values of $|\mathcal{X}^{0}_{a}|$, $|\mathcal{X}^{0}_{\overline{a}}|$, $|\mathcal{X}^{1}_{a}|$, and $|\mathcal{X}^{1}_{\overline{a}}|$ as follows. First, we let $\mathbf{x}^{0}_{a}=\mathbf{u}\,\mathbf{v},~\mathbf{y}^{0}_{a}=\mathbf{u}\,\mathbf{\widetilde{v}},$
where $|\mathbf{v}|=|\mathbf{\widetilde{v}}|=l$, $\mathbf{u}=u_1\,\cdots\,u_s=\mathbf{a}_s$, $s\geq 2$, $s+l=n-3$, $d_L(\mathbf{v},\mathbf{\widetilde{v}})\geq 2$, and the first element and the last element of $\mathbf{v}$ and $\mathbf{\widetilde{v}}$ are both different. By comparison, we have
$\mathbf{x}^{0}_{\overline{a}}=\mathbf{u}\,\mathbf{v}_{[1,l-2]}\,a,~\mathbf{y}^{0}_{\overline{a}}=\mathbf{u}\,\mathbf{\widetilde{v}}\,a,$
where $v_{l-2}=\overline{a}$ and $\widetilde{v}_l=a$. Next, we will study the value of $|\mathcal{X}|$ based on the value of $l$ and the different sequence $\mathbf{u}$.

When $\mathbf{u}$ is not an alternating sequence, by Lemma $\ref{lm6}$, we have
$$|\mathcal{X}^{0}_{a}|\leq 6(n-3)-37=6n-55,$$
for any $n\geq 13$. Combing with Inequations $(\ref{eq8})$-$(\ref{eq10})$, we have $|\mathcal{X}|= |\mathcal{X}^{0}_{a}|+|\mathcal{X}^{0}_{\overline{a}}|+|\mathcal{X}^{1}_{a}|+|\mathcal{X}^{1}_{\overline{a}}|\leq 6n-55+4n-30+4n-29+6n-52=20n-166$ for any $n\geq 13$.

Consider $\mathbf{u}$ is an alternating sequence. Since $\mathbf{u}$ is an alternating sequence $\mathbf{a}_s$, we have $x_i=(i\mod 2), y_j=(j+1\mod 2),$
for any $3\leq i\leq s+2$ and $2\leq j\leq s+1$. For convenience, we let
$$\mathbf{x}^{1}_{a}=\mathbf{p}\,\mathbf{z},~\mathbf{y}^{1}_{a}=\mathbf{p}\,\mathbf{\widetilde{z}};
~\mathbf{x}^{1}_{\overline{a}}=\mathbf{p}\,\mathbf{z}_{[1,l_1]},a,~\mathbf{y}^{1}_{\overline{a}}=\mathbf{p}\,\mathbf{\widetilde{z}}\,a,$$
where $|\mathbf{\widetilde{z}}|=l_1$, $|\mathbf{z}|=l_1+2$, $|\mathbf{p}|=n-4-l_1$, $z_{l_1}=\overline{a}$, and $\widetilde{z}_{l_1}=a$. By comparison, if $s\geq 3$ then we have $p_i=x_{i+1}=y_{i+2}=(i+1\mod 2)$ for any $1\leq i\leq s-1$. Thus, if $s\geq 3$ then $|\mathbf{p}|=n-4-l_1\geq s-1$ which makes $l_1\leq l$. Moreover, if $x_{s+3}\neq x_{s+2}$ then $y_{s+2}=x_{s+2}\neq x_{s+3}$ because of $v_1=x_{s+3}$ and $\widetilde{v}_1=y_{s+2}$. So, $x_{s+1}\neq y_{s+2}$ because of $x_{s+1}\neq x_{s+2}$. Therefore, when $s\geq 3$, if $x_{s+3}\neq x_{s+2}$ then $l_1=l$; otherwise $l_1<l$.

When $s=2$, then we have $x_5=0$ and $l=n-3-2\geq 8$ for any $n\geq 13$. Moreover, the third bit of $\mathbf{x}^{0}_{b}$ and $\mathbf{y}^{0}_{b}$ are not same for all $b\in \{a,\overline{a}\}$. By Lemmas $\ref{lm6}$ and $\ref{lm7}$, we have
$$|\mathcal{X}^{0}_{a}|\leq 6(n-3)-33=6n-51,$$ and
$$|\mathcal{X}^{0}_{\overline{a}}|\leq 4(n-4)-17=4n-33,$$ 
because of $x_4=x_5=0$. Combing with Inequations $(\ref{eq9})$-$(\ref{eq10})$, we have $|\mathcal{X}|= |\mathcal{X}^{0}_{a}|+|\mathcal{X}^{0}_{\overline{a}}|+|\mathcal{X}^{1}_{a}|+|\mathcal{X}^{1}_{\overline{a}}|\leq 6n-51+4n-33+4n-29+6n-52=20n-165$ for any $n\geq 13$. Inequations $(\ref{eq9})$ and $(\ref{eq10})$ hold only when $l_1\leq 4$ and $l_1=6$ which causes a contradiction. So, when $s=2$, we have $|\mathcal{X}|\leq 20n-166$ for any $n\geq 13$.

When $s\geq 3$, we will discuss the value of $|\mathcal{X}|$ based on the value of $l$. Specifically, we will divide in the following three cases: \textbf{Case A}: $l\geq 7$; \textbf{Case B}: $l=6$; \textbf{Case C}: $l\leq 5$.

\textbf{Case A}: Consider $l\geq 7$. If $x_{s+3}\neq x_{s+2}$ then $l_1=l\geq 7$ and $y_{s+2}=x_{s+2}=y_{s+1}\neq x_{s+1}$. Then $p_{s-1}=\widetilde{z}_1$. Since $l=l_1\geq 7$, by Lemmas $\ref{lm6}$ and $\ref{lm7}$, we have
$|\mathcal{X}^{0}_{a}|\leq 6(n-3)-32=6n-50, |\mathcal{X}^{0}_{\overline{a}}|\leq 4(n-4)-15=4n-31,|\mathcal{X}^{1}_{a}|\leq 4(n-4)-16=4n-32,
|\mathcal{X}^{1}_{\overline{a}}|\leq 6(n-3)-35=6n-53,$ for any $n\geq 13$. Therefore, we have $|\mathcal{X}|= |\mathcal{X}^{0}_{a}|+|\mathcal{X}^{0}_{\overline{a}}|+|\mathcal{X}^{1}_{a}|+|\mathcal{X}^{1}_{\overline{a}}|\leq 6n-50+4n-31+4n-32+6n-53=20n-166$ for any $n\geq 13$.

If $x_{s+3}=x_{s+2}$ then $l_1\leq l-1$. For convenience, let $x_{s+2}=c$ with $c\in \mathbb{F}_2$. Then we have $x_{s+1}=y_{s+2}=\overline{c}$ and $x_{s+2}=x_{s+3}=c$. In the following, we divide the two cases: $l\geq 8$ and $l=7$.

When $l\geq 8$, since $u_{s}=v_1$ by Lemmas $\ref{lm6}$ and $\ref{lm7}$, we have
$|\mathcal{X}^{0}_{a}|\leq 6(n-3)-33=6n-51, |\mathcal{X}^{0}_{\overline{a}}|\leq 4(n-4)-17=4n-33,|\mathcal{X}^{1}_{a}|\leq 4(n-4)-13=4n-29,
|\mathcal{X}^{1}_{\overline{a}}|\leq 6(n-3)-34=6n-52,$ for any $n\geq 13$. Therefore, we have $|\mathcal{X}|= |\mathcal{X}^{0}_{a}|+|\mathcal{X}^{0}_{\overline{a}}|+|\mathcal{X}^{1}_{a}|+|\mathcal{X}^{1}_{\overline{a}}|\leq 6n-51+4n-33+4n-29+6n-52=20n-165$ for any $n\geq 13$. The last two inequations hold only when $l_1\leq 4$ and $l_1=6$ which causes a contradiction. So, when $l\geq 8$, we have $|\mathcal{X}|\leq 20n-166$ for any $n\geq 13$.

When $l=7$, if $y_{s+3}=\overline{c}$ then $x_{s+2}\neq y_{s+3}$ which makes $l_1=l-1=6$; otherwise $l_1\leq 5$. Consider $y_{s+3}=\overline{c}$. Then $y_{s+3}\neq x_{s+2}$ and $y_{s+3}=y_{s+2}$ which makes $p_{s-1}=\widetilde{z}_1$. By Lemmas $\ref{lm6}$ and $\ref{lm7}$, we have
$|\mathcal{X}^{0}_{a}|\leq 6(n-3)-32=6n-50, |\mathcal{X}^{0}_{\overline{a}}|\leq 4(n-4)-17=4n-33,|\mathcal{X}^{1}_{a}|\leq 4(n-4)-16=4n-32,
|\mathcal{X}^{1}_{\overline{a}}|\leq 6(n-3)-34=6n-52,$ for any $n\geq 13$. Therefore, $|\mathcal{X}|\leq 20n-167$ for any $n\geq 13$. Consider $y_{s+3}=c$. Then $l_1\leq 5$. By Lemma $\ref{lm6}$, we have
$$|\mathcal{X}^{1}_{\overline{a}}|\leq 4(n-3)-14=4n-26,$$
for any $n\geq 14$. Combing with Inequation $(\ref{eq9})$, we have $|\mathcal{X}|= |\mathcal{X}^{0}_{a}|+|\mathcal{X}^{0}_{\overline{a}}|+|\mathcal{X}^{1}_{a}|+|\mathcal{X}^{1}_{\overline{a}}|\leq 6n-50+4n-33+4n-29+4n-26=18n-138\leq 20n-166$ for any $n\geq 14$. When $n=13$, it have been verified that $|\mathcal{X}|\leq 20n-166$ \cite{Pham2}.

\textbf{Case B}: Consider $l=6$. If $x_{s+3}=x_{s+2}$ then $l_1\leq l-1$. By Lemmas $\ref{lm6}$ and $\ref{lm7}$, we have
we have
$$|\mathcal{X}^{0}_{\overline{a}}|\leq 4(n-4)-16=4n-32,|\mathcal{X}^{1}_{\overline{a}}|\leq 4(n-3)-14=4n-26,$$
for any $n\geq 14$. Combing with Inequations $(\ref{eq7})$ and $(\ref{eq9})$, we have $|\mathcal{X}|= |\mathcal{X}^{0}_{a}|+|\mathcal{X}^{0}_{\overline{a}}|+|\mathcal{X}^{1}_{a}|+|\mathcal{X}^{1}_{\overline{a}}|\leq 6n-49+4n-32+4n-29+4n-26=18n-136\leq 20n-166$ for any $n\geq 15$. When $n=13$ or $14$, it have been verified that $|\mathcal{X}|\leq 20n-166$ \cite{Pham2}.

If $x_{s+3}\neq x_{s+2}$ and $l=6$ then $\mathbf{x}=\mathbf{a}_{s+2}\,x_{n-6}\,x_{n-5}\,x_{n-4}\,\overline{a}\,a\,\overline{a}\,a$ and $\mathbf{y}=0\,\mathbf{a}_{s}\,y_{n-7}\,y_{n-6}\,y_{n-5}\,y_{n-4}\,\overline{a}\,a\,a\,\overline{a}$, where $x_{n-6}=x_{n-8}\neq y_{n-7}$. For convenience, let $x_{n-8}=c$ with $c\in \mathbb{F}_2$. Thus if $n$ is odd then $c=1$; otherwise $c=0$. Moreover, we have $\mathbf{x}^{1}_{a}=0\,\mathbf{a}_{n-13}\,\overline{c}\,c\,\overline{c}\,c\,x_{n-5}\,x_{n-4}\,\overline{a}\,a\,\overline{a}$, $\mathbf{y}^{1}_{a}=0\,\mathbf{a}_{n-13}\,\overline{c}\,\overline{c}\,y_{n-6}\,y_{n-5}\,y_{n-4}\,\overline{a}\,a$, $\mathbf{x}^{1}_{\overline{a}}=0\,\mathbf{a}_{n-13}\,\overline{c}\,c\,\overline{c}\,c\,x_{n-5}\,x_{n-4}\,\overline{a}\,a$, and $\mathbf{y}^{1}_{\overline{a}}=0\,\mathbf{a}_{n-13}\,\overline{c}\,\overline{c}\,y_{n-6}\,y_{n-5}\,y_{n-4}\,\overline{a}\,a\,a$.
Since $y_{n-7}=y_{n-8}=\overline{c}$ thus by Lemma $\ref{lm7}$ we have
$$|\mathcal{X}^{1}_{a}|\leq 4(n-4)-16=4n-32,$$
for any $n\geq 13$. For convenience, let $\mathbf{v}=c\,\overline{c}\,c\,x_{n-5}\,x_{n-4}\,\overline{a}$, $\mathbf{\widetilde{v}}=\overline{c}\,y_{n-6}\,y_{n-5}\,y_{n-4}\,\overline{a}\,a$. Thus, $d_L(\mathbf{v},\mathbf{\widetilde{v}})\geq 2$. By decompose $\mathbf{x}^{1}_{\overline{a}}$ and $\mathbf{y}^{1}_{\overline{a}}$ on both the left and right sides, continue the process until $\mathbf{x}^{1}_{\overline{a}}$ and $\mathbf{y}^{1}_{\overline{a}}$ are transformed into $\mathbf{v}$ and $\mathbf{\widetilde{v}}$, respectively. Then, we have
$$|\mathcal{X}^{1}_{\overline{a}}|\leq (n-3-8)|D_2(\mathbf{v})\cap D_2(\mathbf{\widetilde{v}})|+6+|D_3(\mathbf{v})\cap D_3(\mathbf{\widetilde{v}})|,$$
where $\mathbf{v}=c\,\overline{c}\,c\,x_{n-5}\,x_{n-4}\,\overline{a}$, $\mathbf{\widetilde{v}}=\overline{c}\,y_{n-6}\,y_{n-5}\,y_{n-4}\,\overline{a}\,a$, and $d_L(\mathbf{v},\mathbf{\widetilde{v}})\geq 2$. By using a computerized search, we can obtain $|D_2(\mathbf{v})\cap D_2(\mathbf{\widetilde{v}})|\leq 6$ and $|D_3(\mathbf{v})\cap D_3(\mathbf{\widetilde{v}})|\leq 8$. Moreover, taking into account the structure of $\mathbf{v}$ and $\mathbf{\widetilde{v}}$, $|D_2(\mathbf{v})\cap D_2(\mathbf{\widetilde{v}})|=6$ if and only if $\mathbf{v}=1\,0\,1\,0\,1\,0, \mathbf{\widetilde{v}}=0\,1\,1\,0\,0\,1$ or $\mathbf{v}=0\,1\,0\,1\,0\,1, \mathbf{\widetilde{v}}=1\,0\,0\,1\,1\,0$.

If $|D_2(\mathbf{v})\cap D_2(\mathbf{\widetilde{v}})|\leq 5$ then
$$|\mathcal{X}^{1}_{\overline{a}}|\leq 5(n-3-8)+6+8=5n-41,$$
for any $n\geq 13$.
Combing with Inequations $(\ref{eq7})-(\ref{eq8})$, we have $|\mathcal{X}|= |\mathcal{X}^{0}_{a}|+|\mathcal{X}^{0}_{\overline{a}}|+|\mathcal{X}^{1}_{a}|+|\mathcal{X}^{1}_{\overline{a}}|\leq 6n-49+4n-30+4n-32+5n-41=19n-152\leq 20n-166$ for any $n\geq 14$. When $n=13$ or $14$, it have been verified that $|\mathcal{X}|\leq 20n-166$ \cite{Pham2}.

When $|D_2(\mathbf{v})\cap D_2(\mathbf{\widetilde{v}})|=6$, if $n$ is odd then $\mathbf{v}=1\,0\,1\,0\,1\,0, \mathbf{\widetilde{v}}=0\,1\,1\,0\,0\,1$; otherwise $\mathbf{v}=0\,1\,0\,1\,0\,1, \mathbf{\widetilde{v}}=1\,0\,0\,1\,1\,0$. Consider $n$ is odd. By decompose $\mathbf{x}^{1}_{a}$ and $\mathbf{y}^{1}_{a}$ on both the left and right sides, continue the process until $\mathbf{x}^{1}_{\overline{a}}$ and $\mathbf{y}^{1}_{\overline{a}}$ are transformed into $\mathbf{v}\,1\,0$ and $\mathbf{\widetilde{v}}$, respectively. Then, we have
$$|\mathcal{X}^{1}_{\overline{a}}|\leq 4(n-4-1-6)+1+|D_4(\mathbf{v}\,1\,0)\cap D_2(\mathbf{\widetilde{v}})|,$$
for any $n\geq 13$. Moreover, by using a computerized search, we have $|D_4(\mathbf{v}\,1\,0)\cap D_2(\mathbf{\widetilde{v}})|\leq 8$.
Thus, we have
$$|\mathcal{X}^{1}_{\overline{a}}|\leq 4(n-4-1-6)+1+8=4n-35.$$
Combing with Inequations $(\ref{eq7}),(\ref{eq8}),(\ref{eq10})$, we have $|\mathcal{X}|= |\mathcal{X}^{0}_{a}|+|\mathcal{X}^{0}_{\overline{a}}|+|\mathcal{X}^{1}_{a}|+|\mathcal{X}^{1}_{\overline{a}}|\leq 6n-49+4n-30+4n-35+6n-52=20n-166$ for any $n\geq 13$. Similarly, when $n$ is even we also have $|\mathcal{X}|\leq 20n-166$ for any $n\geq 13$.

\textbf{Case C}: Consider $l\leq 5$. By Lemma $\ref{lm6}$, we have
$$|\mathcal{X}^{0}_{a}|\leq 4(n-3)-14=4n-26,|\mathcal{X}^{1}_{\overline{a}}|\leq 4(n-3)-14=4n-26,$$
for any $n\geq 14$. Combing with Inequations $(\ref{eq8})$-$(\ref{eq9})$, we have $|\mathcal{X}|= |\mathcal{X}^{0}_{a}|+|\mathcal{X}^{0}_{\overline{a}}|+|\mathcal{X}^{1}_{a}|+|\mathcal{X}^{1}_{\overline{a}}|\leq 4n-26+4n-30+4n-29+4n-26=16n-111\leq 20n-166$ for any $n\geq 14$. When $n=13$, it have been verified that $|\mathcal{X}|\leq 20n-166$ \cite{Pham2}.

Based on the above discussion of the five cases, the lemma follows.
\end{proof}

\subsection{The proof of the subcase $(e)$}
Consider the subcase with $x_i \neq y_i$ (for $i = 1, n$), $x_j \neq x_{j+1}$ and $y_j \neq y_{j+1}$ (for $j = 1, n-1$), and with $x_l = y_l$ for $l = 3, n-2$, and at least one of the following: $x_4 \neq y_4$, $x_4 = y_4 = x_3$, $x_{n-3} \neq y_{n-3}$, or $x_{n-3} = y_{n-3} = x_{n-2}$. Without loss of generality,  we only need to consider the first and second scenarios, $x_4 \neq y_4$ or $x_4 = y_4 = x_3$, because the third and fourth scenarios ($x_{n-3} \neq y_{n-3}$, or $x_{n-3} = y_{n-3} = x_{n-2}$) are reduced to the first and second by reversing $\mathbf{x}$ and $\mathbf{y}$. 
\begin{lemma}
\label{lm12}
Let $n\geq 13$ and let $\mathbf{x}=x_1\,\cdots\,x_n, \mathbf{y}=y_1\,\cdots\,y_n\in \mathbb{F}_2^n$ with $x_i \neq y_i$, $x_j \neq x_{j+1}$, $y_j \neq y_{j+1}$, $x_l = y_l$ for $i\in\{1,n\}, j\in\{1,n-1\}, l\in\{3,n-2\}$. If $x_4\neq y_4$ or $x_4=y_4=x_3$, then
\begin{equation}
|D_4(\mathbf{x})\cap D_4(\mathbf{y})|\leq 20n-166.\nonumber
\end{equation}
\end{lemma}
\begin{proof}
Without loss of generality, we let $x_1=1$, $x_3=c$, $x_{n-2}=b, x_n=a$ with $a,b,c\in \mathbb{F}_2$. Then, $\mathbf{x}=1\,0\,c\,x_4\,\cdots\,x_{n-3}\,b\,\overline{a}\,a$, and $\mathbf{y}=0\,1\,c\,y_4\,\cdots\,y_{n-3}\,b\,a\,\overline{a}$. We denote $\mathcal{X}=D_4(\mathbf{x})\cap D_4(\mathbf{y})$. Let $c=0$. Then we have $\mathbf{x}=1\,0\,0\,x_4\,\cdots\,x_{n-3}\,b\,\overline{a}\,a$ and  $\mathbf{y}=0\,1\,0\,y_4\,\cdots\,y_{n-3}\,b\,a\,\overline{a}$. Thus, $\overline{\mathbf{x}}=0\,1\,1\,\overline{x_4}\,\cdots\,\overline{x_{n-3}}\,\overline{b}\,a\,\overline{a}$ and $\overline{\mathbf{y}}=1\,0\,1\,\overline{y_4}\,\cdots\,\overline{y_{n-3}}\,\overline{b}\,\overline{a}\,a$. It is easily verified that $|D_t(\mathbf{x})\cap D_t(\mathbf{y})|=|D_t(\overline{\mathbf{x}})\cap D_t(\overline{\mathbf{y}})|$. Assume that $c=1$. Then we have  $\mathbf{x}=1\,0\,1\,x_4\,\cdots\,x_{n-3}\,b\,\overline{a}\,a$ and  $\mathbf{y}=0\,1\,1\,y_4\,\cdots\,y_{n-3}\,b\,a\,\overline{a}$. By comparison, we only discuss the case where $c=1$.

By Proposition $\ref{prp1}$, it follows that $|\mathcal{X}|=|\mathcal{X}^{0}_{a}|+|\mathcal{X}^{0}_{\overline{a}}|+|\mathcal{X}^{1}_{a}|+|\mathcal{X}^1_{\overline{a}}|$. Next, we discuss the value of $|\mathcal{X}|$ in the following three cases: $1)~x_4=0$ and $y_4=1$; $2)~x_4=1$ and $y_4=0$; $3)~x_4=1$ and $y_4=1$.


\textbf{Case $1)$:} When $x_4=0$ and $y_4=1$, then $\mathbf{x}=(1,0,1,0,\cdots,b,\overline{a},a)$ and  $\mathbf{y}=(0,1,1,1,\cdots,b,a,\overline{a})$. Thus,
\begin{align*}
|\mathcal{X}^{0}_{a}|&=|D_3(\mathbf{x}^{0}_{a})\cap D_3(\mathbf{y}^{0}_{a})|=|D_3(1\,0\,x_5\,\cdots\,x_{n-3}\,b\,\overline{a})\cap D_3(1\,1\,1\,y_5\,\cdots\,y_{n-3}\,b)|,\\
|\mathcal{X}^{0}_{\overline{a}}|&=|D_2(\mathbf{x}^{0}_{\overline{a}})\cap D_4(\mathbf{y}^{0}_{\overline{a}})|=|D_2(1\,0\,x_5\,\cdots\,x_{n-3}\,b)\cap D_4(1\,1\,1\,y_5\,\cdots\,y_{n-3}\,b\,a)|,\\
|\mathcal{X}^{1}_{a}|&=|D_4(\mathbf{x}^{1}_{a})\cap D_3(\mathbf{y}^{1}_{a})|=|D_4(0\,1\,0\,x_5\,\cdots\,x_{n-3}\,b\,\overline{a})\cap D_2(1\,1\,y_5\,\cdots\,y_{n-3}\,b)|,\\
|\mathcal{X}^{1}_{\overline{a}}|&=|D_3(\mathbf{x}^{1}_{\overline{a}})\cap D_3(\mathbf{y}^{1}_{\overline{a}})|=|D_3(0\,1\,0\,x_5\,\cdots\,x_{n-3}\,b)\cap D_3(1\,1\,y_5\,\cdots\,y_{n-3}\,b\,a)|.
\end{align*}
Moreover, if $d_L(\mathbf{x}^{0}_{a},\mathbf{y}^{0}_{a})\leq 1$ then $d_L(\mathbf{x},\mathbf{y})\leq 2$. Thus, $d_L(\mathbf{x}^{0}_{a},\mathbf{y}^{0}_{a})\geq 2$. Similarly, it is easily verified that $\mathbf{x}^{0}_{\overline{a}} \notin D_2(\mathbf{y}^{0}_{\overline{a}}), \mathbf{y}^{1}_{a} \notin D_2(\mathbf{x}^{1}_{a}),$ and $d_L(\mathbf{x}^{1}_{\overline{a}},\mathbf{y}^{1}_{\overline{a}})\geq 2$.

By Lemma $\ref{lm3}$, we have
$$|\mathcal{X}^{0}_{\overline{a}}|\leq 4(n-4)-13=4n-29,|\mathcal{X}^{1}_{a}|\leq 4(n-4)-13=4n-29,$$
where $|\mathbf{x}^{0}_{\overline{a}}|=|\mathbf{y}^{1}_{a}|=n-4$ and $n\geq 13$.
Next, we estimate the value of $|\mathcal{X}^{0}_{a}|$ for any $n\geq 13$. By Propositions $\ref{prp1}$ and $\ref{prp2}$, we have
$|\mathcal{X}^{0}_{a}|=|\mathcal{X}^{00}_{a}|+|\mathcal{X}^{010}_{a}|+|\mathcal{X}^{011}_{a}|,$ and
\begin{align*}
|\mathcal{X}^{00}_{a}|&=|D_2(\mathbf{x}^{00}_{a})\cap D_{0-\ell_1}(\mathbf{y}^{00}_{a})|=|D_2(x_5\,\cdots\,x_{n-3}\,b\,\overline{a})\cap D_{0-\ell_1}(y_{6+\ell_1}\,\cdots\,y_{n-3}\,b)|,\\
|\mathcal{X}^{010}_{a}|&=|D_3(\mathbf{x}^{010}_{a})\cap D_{1-\ell_1}(\mathbf{y}^{010}_{a})|=|D_3(x_5\,\cdots\,x_{n-3}\,b\,\overline{a})\cap D_{1-\ell_1}(y_{6+\ell_1}\,\cdots\,y_{n-3}\,b)|,\\
|\mathcal{X}^{011}_{a}|&=|D_{2-\ell}(\mathbf{x}^{011}_{a})\cap D_{3}(\mathbf{y}^{011}_{a})|=|D_{2-\ell}(x_{6+\ell}\,\cdots\,x_{n-3}\,b\,\overline{a})\cap D_{3}(1\,y_{5}\,\cdots\,y_{n-3}\,b)|.
\end{align*}
Thus, $|\mathcal{X}^{00}_{a}|\leq |D_{0-\ell_1}(y_{6+\ell_1}\,\cdots\,y_{n-3}\,b)|\leq |D_{0}(y_{6+\ell_1}\,\cdots\,y_{n-3}\,b)|=1$ and $|\mathcal{X}^{010}_{a}|\leq |D_1(y_{6}\,\cdots\,y_{n-3}\,b)|\leq n-7$. Consider the value of $|\mathcal{X}^{011}_{a}|$. If $\ell=0$ then $\mathbf{x}^{011}_{a}\notin D_1(\mathbf{y}^{011}_{a})$ because of $\mathbf{x}^{011}_{a}\in D_1(\mathbf{y}^{011}_{a})$ which will make $d_L(\mathbf{x},\mathbf{y})=2$ and causes a contradiction with $d_L(\mathbf{x},\mathbf{y})\geq 3$. Furthermore, by Lemma $\ref{lm3}$, we have $$|\mathcal{X}^{011}_{a}|\leq 3(n-6)-8=3n-26,$$
for $n\geq 13$. Therefore,
$$|\mathcal{X}^{0}_{a}|=|\mathcal{X}^{00}_{a}|+|\mathcal{X}^{010}_{a}|+|\mathcal{X}^{011}_{a}|\leq 1+(n-7)+(3n-26)=4n-32.$$
Similarly, we can also estimate the value of $|\mathcal{X}^{1}_{\overline{a}}|$. By Propositions $\ref{prp1}$ and $\ref{prp2}$, we have 
$|\mathcal{X}^{1}_{\overline{a}}|=|\mathcal{X}^{10}_{\overline{a}}|+|\mathcal{X}^{11}_{\overline{a}}|,$ and
\begin{align*}
|\mathcal{X}^{10}_{\overline{a}}|&=|D_3(\mathbf{x}^{10}_{\overline{a}})\cap D_{1-\ell^*}(\mathbf{y}^{10}_{\overline{a}})|=|D_3(1\,0\,x_5\,\cdots\,x_{n-3}\,b)\cap D_{1-\ell^*}(y_{6+\ell^*}\,\cdots\,y_{n-3}\,b\,a)|,\\
|\mathcal{X}^{11}_{\overline{a}}|&=|D_{2}(\mathbf{x}^{11}_{\overline{a}})\cap D_{3}(\mathbf{y}^{11}_{\overline{a}})|=|D_{2}(0\,x_5\,\cdots\,x_{n-3}\,b)\cap D_{3}(1\,y_{5}\,\cdots\,y_{n-3}\,b\,a)|,
\end{align*}
where $\mathbf{x}^{11}_{\overline{a}}\notin D_{1}(\mathbf{y}^{11}_{\overline{a}})$.
Thus, $|\mathcal{X}^{10}_{\overline{a}}|\leq |D_{1-\ell^*}(y_{6+\ell^*}\,\cdots\,y_{n-3}\,b\,a)|\leq n-6$. By Lemma $\ref{lm3}$, we have
$$|\mathcal{X}^{11}_{\overline{a}}|\leq 3(n-5)-8=3n-23.$$
Therefore,
$$|\mathcal{X}^{1}_{\overline{a}}|=|\mathcal{X}^{10}_{\overline{a}}|+|\mathcal{X}^{11}_{\overline{a}}|\leq n-6+3n-23=4n-29.$$ So, when $x_4=0$ and $y_4=1$, it follows that
$$|\mathcal{X}|=|\mathcal{X}^{0}_{a}|+|\mathcal{X}^{0}_{\overline{a}}|+|\mathcal{X}^{1}_{a}|+|\mathcal{X}^1_{\overline{a}}|\leq 4n-32+(4n-29)*3<20n-166,$$
for any $n\geq 13$.

\textbf{Case $2)$:} When $x_4=1$ and $y_4=0$, then $\mathbf{x}=1\,0\,1\,1\,\cdots\,b\,\overline{a}\,a$ and  $\mathbf{y}=0\,1\,1\,0\,\cdots\,b\,a\,\overline{a}$. Thus,
\begin{align*}
|\mathcal{X}^{0}_{a}|&=|D_3(\mathbf{x}^{0}_{a})\cap D_3(\mathbf{y}^{0}_{a})|=|D_3(1\,1\,x_5\,\cdots\,x_{n-3}\,b\,\overline{a})\cap D_3(1\,1\,0\,y_5\,\cdots\,y_{n-3}\,b)|,\\
|\mathcal{X}^{0}_{\overline{a}}|&=|D_2(\mathbf{x}^{0}_{\overline{a}})\cap D_4(\mathbf{y}^{0}_{\overline{a}})|=|D_2(1\,1\,x_5\,\cdots\,x_{n-3}\,b)\cap D_4(1\,1\,0\,y_5\,\cdots\,y_{n-3}\,b\,a)|,\\
|\mathcal{X}^{1}_{a}|&=|D_4(\mathbf{x}^{1}_{a})\cap D_3(\mathbf{y}^{1}_{a})|=|D_4(0\,1\,1\,x_5\,\cdots\,x_{n-3}\,b\,\overline{a})\cap D_2(1\,0\,y_5\,\cdots\,y_{n-3}\,b)|,\\
|\mathcal{X}^{1}_{\overline{a}}|&=|D_3(\mathbf{x}^{1}_{\overline{a}})\cap D_3(\mathbf{y}^{1}_{\overline{a}})|=|D_3(0\,1\,1\,x_5\,\cdots\,x_{n-3}\,b)\cap D_3(1\,0\,y_5\,\cdots\,y_{n-3}\,b\,a)|,
\end{align*}
where $d_L(\mathbf{x}^{0}_{a},\mathbf{y}^{0}_{a})\geq 2$, $\mathbf{x}^{0}_{\overline{a}} \notin D_2(\mathbf{y}^{0}_{\overline{a}}), \mathbf{y}^{1}_{a} \notin D_2(\mathbf{x}^{1}_{a}),$ and $d_L(\mathbf{x}^{1}_{\overline{a}},\mathbf{y}^{1}_{\overline{a}})\geq 2$. For convenience, we let $\mathbf{x}^{0}_{a}=\mathbf{u\,v\,w}$ and $\mathbf{y}^{0}_{a}=\mathbf{u\,\widetilde{v}\,w}$, where $\mathbf{u}$ and $\mathbf{w}$ are the longest common prefix and suffix of $\mathbf{x}^{0}_{a}$ and $\mathbf{y}^{0}_{a}$, respectively. Hence, $\mathbf{u}$ is also the longest common prefix of $\mathbf{x}^{0}_{\overline{a}}$ and $\mathbf{y}^{0}_{\overline{a}}$. We let $\mathbf{x}^{1}_{\overline{a}}=\mathbf{u_1\,v_1\,w_1}$ and $\mathbf{y}^{1}_{\overline{a}}=\mathbf{u_1\,\widetilde{v_1}\,w_1}$, where $\mathbf{u_1}$ and $\mathbf{w_1}$ are the longest common prefix and suffix of $\mathbf{x}^{1}_{\overline{a}}$ and $\mathbf{y}^{1}_{\overline{a}}$, respectively.

When $b=\overline{a}$, then the last bits of $\mathbf{x}^{0}_{a}$ and $\mathbf{y}^{0}_{a}$ are same such that $|\mathbf{w}|\geq 1$. Since the first and second bits of $\mathbf{x}^{0}_{c}$ and $\mathbf{y}^{0}_{c}$ are same for all $c\in \{a,\overline{a}\}$ thus $\mathbf{u}$ is not a alternating sequence. By Lemmas $\ref{lm6}$ and $\ref{lm7}$, we have
\begin{equation*}
|\mathcal{X}^{0}_{a}|\leq 6(n-3)-40=6n-58, |\mathcal{X}^{0}_{\overline{a}}|\leq 4(n-4)-15=4n-31.
\end{equation*}
By Lemma $\ref{lm3}$ and Corollary $\ref{cor1}$, we have 
\begin{equation*}
|\mathcal{X}^{1}_{a}|\leq 4(n-4)-13=4n-29, |\mathcal{X}^{1}_{\overline{a}}|\leq 6(n-3)-30=6n-48.
\end{equation*}
Therefore, when $b=\overline{a}$ we have $|\mathcal{X}|=|\mathcal{X}^{0}_{a}|+|\mathcal{X}^{0}_{\overline{a}}|+|\mathcal{X}^{1}_{a}|+|\mathcal{X}^1_{\overline{a}}|\leq 6n-58+4n-31+4n-29+6n-48=20n-166,$ for any $n\geq 13$. 

When $b=a$, then the last bits of $\mathbf{x}^{0}_{a}$ and $\mathbf{y}^{0}_{a}$ are not same such that $|\mathbf{w}|=0$. Since the first and second bits of $\mathbf{x}^{0}_{c}$ and $\mathbf{y}^{0}_{c}$ are same for all $c\in \{a,\overline{a}\}$ thus $\mathbf{u}$ is not an alternating sequence. By Lemmas $\ref{lm6}$ and $\ref{lm7}$, we have
\begin{equation*}
|\mathcal{X}^{0}_{a}|\leq 6(n-3)-37=6n-55, |\mathcal{X}^{0}_{\overline{a}}|\leq 4(n-4)-16=4n-32.
\end{equation*}
By Lemma $\ref{lm3}$, we have
\begin{equation*}
|\mathcal{X}^{1}_{a}|\leq 4(n-4)-13=4n-29.
\end{equation*}
If $x_{n-3}=a$ then $\mathbf{w}_1$ is not an alternating sequence. Thus, by Lemma $\ref{lm6}$ we have
\begin{equation*}
|\mathcal{X}^{1}_{\overline{a}}|\leq 6(n-3)-37=6n-55.
\end{equation*}
If $x_{n-3}\neq a$ then $|\mathbf{v}_1|\geq 8$. Thus, by Lemma $\ref{lm6}$ we have
\begin{equation*}
|\mathcal{X}^{1}_{\overline{a}}|\leq 6(n-3)-33=6n-51.
\end{equation*}
Therefore, when $b=a$, we have $|\mathcal{X}|=|\mathcal{X}^{0}_{a}|+|\mathcal{X}^{0}_{\overline{a}}|+|\mathcal{X}^{1}_{a}|+|\mathcal{X}^1_{\overline{a}}|\leq 6n-55+4n-32+4n-29+6n-51<20n-166,$ for any $n\geq 13$. 

By the above discussion, when $x_4=1$ and $y_4=0$, we have $|\mathcal{X}|\leq 20n-166$ for any $n\geq 13$. 

\textbf{Case $3)$:} When $x_4=1$ and $y_4=1$, then $\mathbf{x}=1\,0\,1\,1\,\cdots\,b\,\overline{a}\,a$ and  $\mathbf{y}=0\,1\,1\,1\,\cdots\,b\,a\,\overline{a}$. Thus,
\begin{align*}
|\mathcal{X}^{0}_{a}|&=|D_3(\mathbf{x}^{0}_{a})\cap D_3(\mathbf{y}^{0}_{a})|=|D_3(1\,1\,x_5\,\cdots\,x_{n-3}\,b\,\overline{a})\cap D_3(1\,1\,1\,y_5\,\cdots\,y_{n-3}\,b)|,\\
|\mathcal{X}^{0}_{\overline{a}}|&=|D_2(\mathbf{x}^{0}_{\overline{a}})\cap D_4(\mathbf{y}^{0}_{\overline{a}})|=|D_2(1\,1\,x_5\,\cdots\,x_{n-3}\,b)\cap D_4(1\,1\,1\,y_5\,\cdots\,y_{n-3}\,b\,a)|,\\
|\mathcal{X}^{1}_{a}|&=|D_4(\mathbf{x}^{1}_{a})\cap D_3(\mathbf{y}^{1}_{a})|=|D_4(0\,1\,1\,x_5\,\cdots\,x_{n-3}\,b\,\overline{a})\cap D_2(1\,1\,y_5\,\cdots\,y_{n-3}\,b)|,\\
|\mathcal{X}^{1}_{\overline{a}}|&=|D_3(\mathbf{x}^{1}_{\overline{a}})\cap D_3(\mathbf{y}^{1}_{\overline{a}})|=|D_3(0\,1\,1\,x_5\,\cdots\,x_{n-3}\,b)\cap D_3(1\,1\,y_5\,\cdots\,y_{n-3}\,b\,a)|,
\end{align*}
where $d_L(\mathbf{x}^{0}_{a},\mathbf{y}^{0}_{a})\geq 2$, $\mathbf{x}^{0}_{\overline{a}} \notin D_2(\mathbf{y}^{0}_{\overline{a}}), \mathbf{y}^{1}_{a} \notin D_2(\mathbf{x}^{1}_{a}),$ and $d_L(\mathbf{x}^{1}_{\overline{a}},\mathbf{y}^{1}_{\overline{a}})\geq 2$. We discuss the value of $|\mathcal{X}|$ based on $b=\overline{a}$ or $a$. Similarly, by using the method of proving Case $2$, when $x_4=1$ and $y_4=1$, we also have $|\mathcal{X}|\leq 20n-166$ for any $n\geq 13$.

By the above discussion, if $x_4\neq y_4$ or $x_4=y_4=1$, then we have $|D_4(\mathbf{x})\cap D_4(\mathbf{y})|\leq 20n-166,$ which completes the proof of Lemma $\ref{lm12}$.
\end{proof}

\subsection{The proof of the subcase $(f)$}
For the subcase where $x_i \neq y_i$ (for $i = 1, n$), $x_j \neq x_{j+1}$ and $y_j \neq y_{j+1}$ (for $j = 1, 3, n-2, n-1$), and $x_l = y_l$ (for $l = 3, 4, n-3, n-2$),  Lemma $\ref{lm17}$ shows that $|D_4(\mathbf{x})\cap D_4(\mathbf{y})|\leq 20n-166$. Before that, we introduce some preliminary lemmas and notations.


\begin{definition}
Two binary sequences $\mathbf{x}$ and $\mathbf{y}$ are Type-A-confusable if 
$$\mathbf{x}=\mathbf{u\,v\,w},~~ \text{and}~~ \mathbf{y}=\mathbf{u\,\overline{v}\,w},$$
for some substrings $\mathbf{u}$, $\mathbf{v}$, and $\mathbf{w}$ such that $|\mathbf{v}|\geq 2$, $\overline{\mathbf{v}}$ is the complement of $\mathbf{v}$, and $\mathbf{v}$ is an alternating sequence.
\end{definition}

\begin{definition}
Two binary sequences $\mathbf{x}$ and $\mathbf{y}$ are Type-B-confusable if
$$\mathbf{x}=\mathbf{u}\,a\,\overline{a}\,\mathbf{v}\,b\,\mathbf{w},~~ \text{and}~~ \mathbf{y}=\mathbf{u}\,\overline{a}\,\mathbf{v}\,b\,\overline{b}\,\mathbf{w},$$
or vice versa, for some substrings $\mathbf{u}$, $\mathbf{v}$, and $\mathbf{w}$, and $a,b\in \mathbb{F}_2$.
\end{definition}

Some properties of Type-A/B confusable sequences are given in \cite{Chrisnata}.

\begin{lemma}[{\cite[Lemma 2]{Chrisnata}}]
\label{lm13}
Let $\mathbf{x}$ and $\mathbf{y}$ be binary sequences. We have that $|D_1(\mathbf{x})\cap D_1(\mathbf{y})|=2$ if and only if $\mathbf{x}$ and $\mathbf{y}$ are Type-A confusable.
\end{lemma}

\begin{lemma}[{\cite[Lemma 4]{Chrisnata}}]
\label{lm14}
Let $\mathbf{x}$ and $\mathbf{y}$ be binary sequences. If $|D_1(\mathbf{x})\cap D_1(\mathbf{y})|=1$, then either the Hamming distance of $\mathbf{x}$ and $\mathbf{y}$ is one or $\mathbf{x}$ and $\mathbf{y}$ are Type-B confusable.
\end{lemma}

When $d_L(\mathbf{x},\mathbf{y})=1$, we discuss some properties of $|D_2(\mathbf{x})\cap D_2(\mathbf{y})|$ as follows. Since $d_L(\mathbf{x},\mathbf{y})=1$, we have $|D_1(\mathbf{x})\cap D_1(\mathbf{y})|\geq 1$. That is, $|D_1(\mathbf{x})\cap D_1(\mathbf{y})|=1,~\text{or}~2$.

\begin{lemma}[{\cite[Proposition 12]{Chrisnata}}]
\label{lm15}
Let $\mathbf{x}$ and $\mathbf{y}$ be binary sequences of length $n\geq 4$ that are Type-$B$-confusable. If $D_1(\mathbf{x})\cap D_1(\mathbf{y})=\{\mathbf{z}\}$ then we have
\begin{equation*}
\big|\big(D_2(\mathbf{x})\cap D_2(\mathbf{y})\big)\backslash D_1(\mathbf{z})\big|\leq 2.
\end{equation*}
Thus $|D_2(\mathbf{x})\cap D_2(\mathbf{y})|\leq n$.
\end{lemma}

\begin{lemma}
\label{lm16}
Let $\mathbf{x}$ and $\mathbf{y}$  be binary sequences of length $n\geq 7$, and $d_L(\mathbf{x},\mathbf{y})=1$.
There are some sequences $\mathbf{x}$ and $\mathbf{y}$ such that $|D_2(\mathbf{x}) \cap D_2(\mathbf{y})|=2n-4,2n-5,~\text{or}~2n-6$, then $\mathbf{x}$ and $\mathbf{y}$ are Type-A confusable. Specifically, we let $\mathbf{x}=\mathbf{u\,a\,v}$ and $\mathbf{y}=\mathbf{u\,\overline{a}\,v}$ with $|\mathbf{u}|=s,|\mathbf{a}|=l,|\mathbf{v}|=t$, where $s\geq 0,t\geq 0,l\geq 2, s+l+t=n$, and $\mathbf{a,u,v}$ are alternating sequences. Then we have the following results:
\begin{enumerate}
    \item  $|D_2(\mathbf{x}) \cap D_2(\mathbf{y})|=2n-4$ if and only if $l=n$, or $l=2$ and $\min\{s,t\}=0$;
    \item  $|D_2(\mathbf{x}) \cap D_2(\mathbf{y})|=2n-5$ if and only if $l=n-1$, $l=n-2$ and $\{s,t\}=\{0,2\}$, or $l=2$ and $s,t\geq 1$, or $n-2\geq l\geq 3$, $\min\{s,t\}=0$ and $\max\{s,t\}\geq 1$;
     \item  $|D_2(\mathbf{x}) \cap D_2(\mathbf{y})|=2n-6$ if and only if $n-2\geq l\geq 3$, $s\geq 1$, and $t\geq 1$;
\end{enumerate}
Moreover, for any $n\geq 4$, if $|D_1(\mathbf{x})\cap D_1(\mathbf{y})|\leq 1$ then $|D_2(\mathbf{x}) \cap D_2(\mathbf{y})|\leq n$.
\end{lemma}
The proof of Lemma $\ref{lm16}$ is given in Appendix \ref{APP-E}.

Now we discuss the subcase where $x_i \neq y_i$ (for $i = 1, n$), $x_j \neq x_{j+1}$ and $y_j \neq y_{j+1}$ (for $j = 1, 3, n-2, n-1$), and $x_l = y_l$ (for $l = 3, 4, n-3, n-2$). 
\begin{lemma}
\label{lm17}
Let $n\geq 13$ and let $\mathbf{x}=x_1\,\cdots\,x_n, \mathbf{y}=y_1\,\cdots\,y_n\in \mathbb{F}_2^n$ with $x_i \neq y_i$, $x_j \neq x_{j+1}$, $y_j \neq y_{j+1}$, and $x_l = y_l$ for $i\in \{1,n\}$, $j\in\{1,3,n-3,n-1\}$, $l\in \{3,4,n-3,n-2\}$. If $d_L(\mathbf{x},\mathbf{y})\geq 3$, then we have
\begin{equation}
|D_4(\mathbf{x})\cap D_4(\mathbf{y})|\leq 20n-166.\nonumber
\end{equation}
\end{lemma}
\begin{proof}
Without loss of generality, we let $x_1=1$, $x_3=c$, $x_{n-2}=b, x_n=a$ with $a,b,c\in \mathbb{F}_2$. Then, $\mathbf{x}=1\,0\,c\,\overline{c}\,x_5\,\cdots\,x_{n-4}\,\overline{b}\,b\,\overline{a}\,a$, and $\mathbf{y}=0\,1\,c\,\overline{c}\,y_5\,\cdots\,y_{n-4}\,\overline{b}\,b\,a\,\overline{a}$. We denote $\mathcal{X}=D_4(\mathbf{x})\cap D_4(\mathbf{y})$. Let $c=0$. Then we have $\mathbf{x}=1\,0\,0\,1\,x_5\,\cdots\,x_{n-4}\,\overline{b}\,b\,\overline{a}\,a$ and  $\mathbf{y}=0\,1\,0\,1\,y_5\,\cdots\,y_{n-4}\,\overline{b}\,b\,a\,\overline{a}$. Thus, $\overline{\mathbf{x}}=0\,1\,1\,0\,\overline{x_5}\,\cdots\,\overline{x_{n-4}}\,b\,\overline{b}\,a\,\overline{a}$ and $\overline{\mathbf{y}}=1\,0\,1\,0\,\overline{y_5}\,\cdots\,\overline{y_{n-4}}\,b\,\overline{b}\,\overline{a}\,a$. It is easily verified that $|D_t(\mathbf{x})\cap D_t(\mathbf{y})|=|D_t(\overline{\mathbf{x}})\cap D_t(\overline{\mathbf{y}})|$. Assume that $c=1$. Then we have  $\mathbf{x}=1\,0\,1\,0\,x_5\,\cdots\,x_{n-4}\,\overline{b}\,b\,\overline{a}\,a$ and  $\mathbf{y}=0\,1\,1\,0\,y_5\,\cdots\,y_{n-4}\,\overline{b}\,b\,a\,\overline{a}$. By comparison, we only discuss the case where $c=1$. For convenience, we let $\mathbf{x}=1\,0\,1\,0\,\mathbf{v}\,\overline{b}\,b\,\overline{a}\,a$ and  $\mathbf{y}=0\,1\,1\,0\,\mathbf{w}\,\overline{b}\,b\,a\,\overline{a}$ with $\mathbf{v}=v_1\,\cdots\,v_{n-8},\mathbf{w}=w_1\,\cdots\,w_{n-8}\in \mathbb{F}_2^{n-8}$.

First, we decompose $\mathcal{X}$ by Propositions $\ref{prp1}$ and $\ref{prp2}$ as follows: we have $|\mathcal{X}|=|\mathcal{X}^{0}_{a}|+|\mathcal{X}^{0}_{\overline{a}}|+|\mathcal{X}^{1}_{a}|+|\mathcal{X}^1_{\overline{a}}|$, where
\begin{align*}
|\mathcal{X}^{0}_{a}|&=|D_3(\mathbf{x}^{0}_{a})\cap D_3(\mathbf{y}^{0}_{a})|=|D_3(1\,0\,\mathbf{v}\,\overline{b}\,b\,\overline{a})\cap D_3(1\,1\,0\,\mathbf{w}\,\overline{b}\,b)|,\\
|\mathcal{X}^{0}_{\overline{a}}|&=|D_2(\mathbf{x}^{0}_{\overline{a}})\cap D_4(\mathbf{y}^{0}_{\overline{a}})|=|D_2(1\,0\,\mathbf{v}\,\overline{b}\,b)\cap D_4(1\,1\,0\,\mathbf{w}\,\overline{b}\,b\,a)|,\\
|\mathcal{X}^{1}_{a}|&=|D_4(\mathbf{x}^{1}_{a})\cap D_3(\mathbf{y}^{1}_{a})|=|D_4(0\,1\,0\,\mathbf{v}\,\overline{b}\,b\,\overline{a})\cap D_2(1\,0\,\mathbf{w}\,\overline{b}\,b)|,\\
|\mathcal{X}^{1}_{\overline{a}}|&=|D_3(\mathbf{x}^{1}_{\overline{a}})\cap D_3(\mathbf{y}^{1}_{\overline{a}})|=|D_3(0\,1\,0\,\mathbf{v}\,\overline{b}\,b)\cap D_3(1\,0\,\mathbf{w}\,\overline{b}\,b\,a)|.
\end{align*}
Moreover, if $d_L(\mathbf{x}^{0}_{a},\mathbf{y}^{0}_{a})\leq 1$ then $d_L(\mathbf{x},\mathbf{y})\leq 2$. Thus, $d_L(\mathbf{x}^{0}_{a},\mathbf{y}^{0}_{a})\geq 2$. Similarly, it is easily verified that $\mathbf{x}^{0}_{\overline{a}} \notin D_2(\mathbf{y}^{0}_{\overline{a}}), \mathbf{y}^{1}_{a} \notin D_2(\mathbf{x}^{1}_{a}),$ and $d_L(\mathbf{x}^{1}_{\overline{a}},\mathbf{y}^{1}_{\overline{a}})\geq 2$.

We discuss the value of $|D_4(\mathbf{x})\cap D_4(\mathbf{y})|$ in two cases: A)$b=\overline{a}$; B)$b=a$.

\textbf{Case A): $b=\overline{a}$}. Then $\mathbf{x}=1\,0\,1\,0\,\mathbf{v}\,a\,\overline{a}\,\overline{a}\,a$ and $\mathbf{y}=0\,1\,1\,0\,\mathbf{w}\,a\, \overline{a}\,a\,\overline{a}$. By Propositions $\ref{prp1}$ and $\ref{prp2}$, we can decompose $\mathcal{X}^{0}_{\overline{a}},\mathcal{X}^{1}_{a}$, and $\mathcal{X}^{1}_{\overline{a}}$, respectively. Then, $\mathcal{X}^{0}_{\overline{a}}$ can be decomposed as follows:
\begin{align}
|\mathcal{X}^{0}_{\overline{a}}|&=|\mathcal{X}^{00}_{\overline{a}a}|+|\mathcal{X}^{00}_{\overline{a}\overline{a}}|+|\mathcal{X}^{01}_{\overline{a}a}|+
|\mathcal{X}^{010}_{\overline{a}\overline{a}}|+|\mathcal{X}^{011}_{\overline{a}\overline{a}}|,\nonumber\\
|\mathcal{X}^{00}_{\overline{a}a}|&=|D_0(\mathbf{x}^{00}_{\overline{a}a})\cap D_{2}(\mathbf{y}^{00}_{\overline{a}a})|=|D_0(\mathbf{v})\cap D_{2}(\mathbf{w}\,a\,\overline{a})|,\nonumber\\
|\mathcal{X}^{00}_{\overline{a}\overline{a}}|&=|D_1(\mathbf{x}^{00}_{\overline{a}\overline{a}})\cap D_{1}(\mathbf{y}^{00}_{\overline{a}\overline{a}})|=|D_1(\mathbf{v}\,a)\cap D_{1}(\mathbf{w}\,a)|,\nonumber\\
|\mathcal{X}^{01}_{\overline{a}a}|&=|D_1(\mathbf{x}^{01}_{\overline{a}a})\cap D_{4}(\mathbf{y}^{01}_{\overline{a}a})|=|D_1(0\,\mathbf{v})\cap D_{4}(1\,0\,\mathbf{w}\,a\,\overline{a})|,\nonumber\\
|\mathcal{X}^{010}_{\overline{a}\overline{a}}|&=|D_2(\mathbf{x}^{010}_{\overline{a}\overline{a}})\cap D_{2}(\mathbf{y}^{010}_{\overline{a}\overline{a}})|=|D_2(\mathbf{v}\,a)\cap D_{2}(\mathbf{w}\,a)|,\nonumber\\
|\mathcal{X}^{011}_{\overline{a}\overline{a}}|=|D_{1-\ell}&(\mathbf{x}^{011}_{\overline{a}\overline{a}})\cap D_{3}(\mathbf{y}^{011}_{\overline{a}\overline{a}})|=|D_{1-\ell}(\mathbf{v}_{[2+\ell,n-8]}\,a)\cap D_{3}(0\,\mathbf{w}\,a)|,\label{eq11}
\end{align}
where $\ell\geq 0$. If $\ell=0$ then $v_{1}=1$; if $\ell=1$ then $v_{1}=0$ and $v_{2}=1$.

Then, $\mathcal{X}^{1}_{a}$ can be decomposed as follows:
\begin{align}
|\mathcal{X}^{1}_{a}|&=|\mathcal{X}^{10}_{aa}|+|\mathcal{X}^{10}_{a\overline{a}}|+|\mathcal{X}^{11}_{aa}|+|\mathcal{X}^{11}_{a\overline{a}a}|
+|\mathcal{X}^{11}_{a\overline{a}\,\overline{a}}|,\nonumber\\
|\mathcal{X}^{10}_{aa}|&=|D_2(\mathbf{x}^{10}_{aa})\cap D_{0}(\mathbf{y}^{10}_{aa})|=|D_2(1\,0\,\mathbf{v})\cap D_{0}(\mathbf{w})|,\nonumber\\
|\mathcal{X}^{10}_{a\overline{a}}|&=|D_4(\mathbf{x}^{10}_{a\overline{a}})\cap D_{1}(\mathbf{y}^{(1,0)}_{a\overline{a}})|=|D_4(1\,0\,\mathbf{v}\,a\,\overline{a})\cap D_{1}(\mathbf{w}\,a)|,\nonumber\\
|\mathcal{X}^{11}_{aa}|&=|D_1(\mathbf{x}^{11}_{aa})\cap D_{1}(\mathbf{y}^{11}_{aa})|=|D_1(0\,\mathbf{v})\cap D_{1}(0\,\mathbf{w})|,\nonumber\\
|\mathcal{X}^{11}_{a\overline{a}a}|&=|D_2(\mathbf{x}^{11}_{a\overline{a}a})\cap D_{2}(\mathbf{y}^{11}_{a\overline{a}a})|=|D_2(0\,\mathbf{v})\cap D_{2}(0\,\mathbf{w})|,\nonumber\\
|\mathcal{X}^{11}_{a\overline{a}\,\overline{a}}|=|D_3(\mathbf{x}&^{11}_{a\overline{a}\,\overline{a})})\cap D_{1-\ell^*}(\mathbf{y}^{11}_{a\overline{a}\,\overline{a})})|=|D_3(0\,\mathbf{v}\,a)\cap D_{1-\ell^*}(0\,\mathbf{w}_{[1,n-9-\ell^*]})|,\label{eq12}
\end{align}
where $\ell^*\geq 0$. If $\ell^*=0$ then $\mathbf{w}_{[n-8]}=\overline{a}$; if $\ell^*=1$ then $\mathbf{w}_{[n-8]}=a$ and $\mathbf{w}_{[n-9]}=\overline{a}$.

Moreover, $\mathcal{X}^{1}_{\overline{a}}$ can be decomposed as follows:
\begin{align}
|\mathcal{X}^{1}_{\overline{a}}|=&|\mathcal{X}^{10}_{\overline{a}a}|+|\mathcal{X}^{10}_{\overline{a}\,\overline{a}}|+
|\mathcal{X}^{11}_{\overline{a}a}|+|\mathcal{X}^{11}_{\overline{a}\,\overline{a}}|,\nonumber\\
|\mathcal{X}^{10}_{\overline{a}a}|=&|D_2(\mathbf{x}^{10}_{\overline{a}a})\cap D_{2}(\mathbf{y}^{10}_{\overline{a}a})|=|D_2(1\,0\,\mathbf{v})\cap D_{2}(\mathbf{w}\,a\,\overline{a})|,\nonumber\\
|\mathcal{X}^{10}_{\overline{a}\,\overline{a}}|=&|D_3(\mathbf{x}^{10}_{\overline{a}\,\overline{a}})\cap D_{1}(\mathbf{y}^{10}_{\overline{a}\,\overline{a}})|=|D_3(1\,0\,\mathbf{v}\,a)\cap D_{1}(\mathbf{w}\,a)|,\nonumber\\
|\mathcal{X}^{11}_{\overline{a}a}|=&|D_1(\mathbf{x}^{11}_{\overline{a}a})\cap D_{3}(\mathbf{y}^{11}_{\overline{a}a})|=|D_1(0\,\mathbf{v})\cap D_{3}(0\,\mathbf{w}\,a\,\overline{a})|,\nonumber\\
|\mathcal{X}^{11}_{\overline{a}\,\overline{a}}|=&|D_2(\mathbf{x}^{11}_{\overline{a}\,\overline{a}})\cap D_{2}(\mathbf{y}^{11}_{\overline{a}\,\overline{a}})|=|D_2(0\,\mathbf{v}\,a)\cap D_{2}(0\,\mathbf{w}\,a)|.\label{eq13}
\end{align}

By Lemma $\ref{lm6}$, for any $n\geq 13$, we have
\begin{equation}
|\mathcal{X}^{0}_{a}|\leq 6(n-3)-36=6n-54,\label{eq14}
\end{equation}
because of $|\mathbf{x}^{0}_{a}|=n-3$ and $\mathbf{x}^{0}_{a},\mathbf{y}^{0}_{a}$ satisfy one condition of $l\geq 8$ in Lemma $\ref{lm6}$. Since $d_L(\mathbf{x}^{1}_{\overline{a}},\mathbf{y}^{1}_{\overline{a}})\geq 2$, thus we have
\begin{equation}
|\mathcal{X}^{1}_{\overline{a}}|\leq N(n-3,2,3)=6(n-3)-30=6n-48,\label{eq15}
\end{equation}
for any $n\geq 13$.
By Lemma $\ref{lm7}$, we have
we have
\begin{equation}
|\mathcal{X}^{0}_{\overline{a}}|\leq 4(n-4)-14=4n-30,\label{eq16}
\end{equation}
because of $|\mathbf{x}^{0}_{\overline{a}}|=n-4$ and $\mathbf{x}^{0}_{\overline{a}},\mathbf{y}^{0}_{\overline{a}}$ satisfy one condition of $l\geq 5$ in Lemma $\ref{lm7}$.

Next, we analyze the value of $|D_1(0\,\mathbf{v})\cap D_1(0\,\mathbf{w})|$ in $|\mathcal{X}_{a}^{1}|$ (see Equation (\ref{eq12})). Since $d_L(\mathbf{x},\mathbf{y})\geq 3$ we have $d_L(\mathbf{v},\mathbf{w})\geq 1$. Thus, $|D_1(0\,\mathbf{v})\cap D_1(0\,\mathbf{w})|=0,1,$ or $2$.

When $|D_1(0\,\mathbf{v})\cap D_1(0\,\mathbf{w})|=0$, then $d_L(0\,\mathbf{v},0\,\mathbf{w})\geq 2$. We have $|D_2(0\,\mathbf{v})\cap D_2(0\,\mathbf{w})|\leq N(n-7,2,2)=6$ for any $n\geq 13$. Therefore, by $(\ref{eq12})$ we have
\begin{equation}
|\mathcal{X}^{1}_{a}|\leq 1+(n-7)+0+6+(n-8)=2n-8,\label{eq17}
\end{equation}
because of $|\mathcal{X}^{10}_{aa}|\leq |D_{0}(\mathbf{w})|=1$, $|\mathcal{X}^{10}_{a\overline{a}}|\leq |D_{1}(\mathbf{w}\,a)|\leq n-7$, $|\mathcal{X}^{11}_{aa}|=0$, $|\mathcal{X}^{11}_{a\overline{a}a}|\leq 6$, and $|\mathcal{X}^{11}_{a\overline{a}\,\overline{a}}|\leq |D_1(0\,\mathbf{w}_{[1,n-9]})|\leq n-8$.
for any $n\geq 13$. So, when $|D_1(0\,\mathbf{v})\cap D_1(0\,\mathbf{w})|=0$, by $(\ref{eq14})$-$(\ref{eq17})$ it follows that
\begin{align*}
|\mathcal{X}|=|\mathcal{X}^{0}_{a}|+|\mathcal{X}^{0}_{\overline{a}}|+|\mathcal{X}^{1}_{a}|+|\mathcal{X}^1_{\overline{a}}|\leq 6n-54+4n-30+2n-8+6n-48=18n-140\leq 20n-166,
\end{align*}
for any $n\geq 13$.

When $|D_1(0\,\mathbf{v})\cap D_1(0\,\mathbf{w})|=1$, then $d_L(0\,\mathbf{v}, 0\,\mathbf{w})=1$. By Lemma $\ref{lm16}$, we have $|D_2(0\,\mathbf{v})\cap D_2(0\,\mathbf{w})|\leq n-7$ for any $n\geq 13$. Moreover, $\mathbf{v}$ and $\mathbf{w}$ are not both alternating sequences. Thus, $|D_1(\mathbf{w}\,a)|=n-7$ and $|D_1(0\,\mathbf{v})|=n-7$ do not both hold true.
Therefore, by $(\ref{eq12})$ we have
\begin{equation}
|\mathcal{X}^{1}_{a}|\leq 1+(n-7)+1+(n-7)+(n-8)=3n-20,\label{eq18}
\end{equation}
because of $|\mathcal{X}^{10}_{aa}|\leq |D_{0}(\mathbf{w})|=1$, $|\mathcal{X}^{10}_{a\overline{a}}|\leq |D_{1}(\mathbf{w}\,a)|\leq n-7$, $|\mathcal{X}^{11}_{aa}|=1$, $|\mathcal{X}^{11}_{a\overline{a}a}|\leq n-7$, and $|\mathcal{X}^{11}_{a\overline{a}\,\overline{a}}|\leq |D_1(0\,\mathbf{w}_{[1,n-9]})|\leq n-8$. By $(\ref{eq13})$, we also have
\begin{equation}
|\mathcal{X}^{1}_{\overline{a}}|\leq 2(n-8)+(n-7)+(n-7)+(n-6)=5n-36,\label{eq19}
\end{equation}
because of $|\mathcal{X}^{10}_{\overline{a}a}|\leq N(n-6,1,2)=2(n-6-2)$, $|\mathcal{X}^{10}_{\overline{a}\,\overline{a}}|\leq |D_{1}(\mathbf{w}\,a)|\leq n-7$, $|\mathcal{X}^{11}_{\overline{a}a}|\leq |D_1(0\,\mathbf{v})|\leq n-7$, and $|\mathcal{X}^{11}_{\overline{a}\,\overline{a}}|\leq n-6$ by using Lemma $\ref{lm16}$.
So, when $|D_1(0\,\mathbf{v})\cap D_1(0\,\mathbf{w})|=1$, by $(\ref{eq14}),(\ref{eq16}),(\ref{eq18})$, and $(\ref{eq19})$, we have
\begin{align*}
|\mathcal{X}|=|\mathcal{X}^{0}_{a}|+|\mathcal{X}^{0}_{\overline{a}}|+|\mathcal{X}^{1}_{a}|+|\mathcal{X}^1_{\overline{a}}|\leq 6n-54+4n-30+3n-20+5n-36=18n-140\leq 20n-166,
\end{align*}
for any $n\geq 13$.

When $|D_1(0\,\mathbf{v})\cap D_1(0\,\mathbf{w})|=2$, by Lemma $\ref{lm13}$, then $0\,\mathbf{v}$ and $0\,\mathbf{w}$ are Type-A-confusable. Moreover, $\mathbf{v}$ and $\mathbf{w}$ are Type-A-confusable. For convenience, let $\mathbf{v}=\mathbf{v}_1\,\mathbf{u}\,\mathbf{w}_2$ and $\mathbf{w}=\mathbf{v}_1\,\overline{\mathbf{u}}\,\mathbf{w}_2$, where $|\mathbf{v}_1|=s$, $|\mathbf{u}|=l$, $|\mathbf{w}_2|=t$, and $s+l+t=n-8$, where $\mathbf{u}$ is an alternating sequence. Next, we discuss the value of $|\mathcal{X}|$ based on different values of $s$ and $t$.

\textbf{Case 1}: When $s\geq 1$ and $t\geq 1$, we have that $\mathbf{v}_1\,\mathbf{u}$ and $\mathbf{v}_1\,\overline{\mathbf{u}}$ are not both alternating sequences, and $\mathbf{u}\,\mathbf{w}_2$ and $\overline{\mathbf{u}}\,\mathbf{w}_2$ are not both alternating sequences. Hence, $|D_1(\mathbf{v})|+|D_1(\mathbf{w})|\leq 2(n-8)-2$. By $(\ref{eq11})$ and $(\ref{eq12})$, then we have
\begin{align*}
|\mathcal{X}^{1}_{a}|&\leq |D_{0}(\mathbf{w})|+|D_{1}(\mathbf{w}\,a)|+|D_1(0\,\mathbf{v})\cap D_{1}(0\,\mathbf{w})|+|D_2(0\,\mathbf{v})\cap D_{2}(0\,\mathbf{w})|+|D_{1-\ell^*}(0\,\mathbf{w}_{[1,n-9-\ell^*]})|\\
&\leq n-5+|D_{1}(\mathbf{w}\,a)|+|D_2(0\,\mathbf{v})\cap D_{2}(0\,\mathbf{w})|,\\
|\mathcal{X}^{0}_{\overline{a}}|&\leq |D_{0}(\mathbf{v})|+|D_1(\mathbf{v}\,a)\cap D_{1}(\mathbf{w}\,a)|+|D_{1}(0\,\mathbf{v})|+|D_2(\mathbf{v}\,a)\cap D_{2}(\mathbf{w}\,a)|+|D_{1-\ell}(\mathbf{v}_{[2+\ell,n-8]}\,a)|\\
&\leq n-5+|D_{1}(0\,\mathbf{v})|+|D_2(\mathbf{v}\,a)\cap D_{2}(\mathbf{w}\,a)|,
\end{align*}
because of $|D_1(\mathbf{v}\,a)\cap D_{1}(\mathbf{w}\,a)|=2$.
Since $|D_1(\mathbf{v})|+|D_1(\mathbf{w})|\leq 2(n-8)-2$, we have $|D_1(0\,\mathbf{v})|+|D_1(\mathbf{w}\,a)|\leq 2(n-8)$. Since $\mathbf{v}$ and $\mathbf{w}$ are Type-A-confusable, and $s,t\geq 1$, by Lemma $\ref{lm16}$, we have $|D_2(\mathbf{v}\,a)\cap D_{2}(\mathbf{w}\,a)|\leq 2(n-7)-5$ and $|D_2(0\,\mathbf{v})\cap D_{2}(0\,\mathbf{w})|\leq 2(n-7)-5$.
Therefore,
\begin{equation}
|\mathcal{X}^{1}_{a}|+|\mathcal{X}^{0}_{\overline{a}}|\leq 2(n-5)+2(n-8)+2(2n-14-5)=8n-64.
\label{eq20}
\end{equation}
So, by $(\ref{eq14})$, $(\ref{eq15})$, and $(\ref{eq20})$, it follows that
\begin{align*}
|\mathcal{X}|&=|\mathcal{X}^{0}_{a}|+|\mathcal{X}^{0}_{\overline{a}}|+|\mathcal{X}^{1}_{a}|+|\mathcal{X}^1_{\overline{a}}|\leq 6n-54+8n-64+6n-48=20n-166.
\end{align*}

\textbf{Case 2}: When $s\geq 1$ and $t=0$, then $\mathbf{v}=\mathbf{v}_1\,\mathbf{u}$ and $\mathbf{w}=\mathbf{v}_1\,\overline{\mathbf{u}}$ are not both alternating sequences. Thus $\mathbf{v}_1\,\mathbf{u}$ and $\mathbf{v}_1\,\overline{\mathbf{u}}$ are not both alternating sequences. Hence, $|D_1(\mathbf{v})|+|D_1(\mathbf{w})|\leq 2(n-8)-1$.
Since $|D_1(\mathbf{v})|+|D_1(\mathbf{w})|\leq 2(n-8)-1$, we have $|D_1(0\,\mathbf{v})|+|D_1(\mathbf{w}\,a)|\leq 2(n-8)+1$. Since $\mathbf{v}$ and $\mathbf{w}$ are Type-A-confusable, and $s\neq 1$ or $n-10$, by Lemma $\ref{lm16}$, we have $|D_2(\mathbf{v}\,a)\cap D_{2}(\mathbf{w}\,a)|\leq 2(n-7)-6$ and $|D_2(0\,\mathbf{v})\cap D_{2}(0\,\mathbf{w})|\leq 2(n-7)-5$.
Therefore,
\begin{equation}
|\mathcal{X}^{1}_{a}|+|\mathcal{X}^{0}_{\overline{a}}|\leq 2(n-5)+2(n-8)+1+2(2n-14)-11=8n-64.\label{eq21}
\end{equation}
So, if $s\neq 1 $ or $n-10$ then by $(\ref{eq14})$, $(\ref{eq15})$, and $(\ref{eq21})$, we have
\begin{align*}
|\mathcal{X}|=|\mathcal{X}^{0}_{a}|+|\mathcal{X}^{0}_{\overline{a}}|+|\mathcal{X}^{1}_{a}|+|\mathcal{X}^1_{\overline{a}}|\leq 6n-54+8n-64+6n-48=20n-166.
\end{align*}
If $s=1$ and $t=0$ then by Lemma $\ref{lm16}$ we have $|D_2(\mathbf{v}\,a)\cap D_{2}(\mathbf{w}\,a)|\leq 2(n-7)-6$ and $|D_2(0\,\mathbf{v})\cap D_{2}(0\,\mathbf{w})|\leq 2(n-7)-4$.
Therefore,
\begin{equation}
|\mathcal{X}^{1}_{a}|+|\mathcal{X}^{0}_{\overline{a}}|\leq 2(n-5)+2(n-8)+1+2(2n-14)-10=8n-63.\label{eq22}
\end{equation}
If $s=1$ and $t=0$ then $\mathbf{v}$ and $\mathbf{w}$ are not both alternating sequences. Thus, $|D_1(0\,\mathbf{v})|+|D_{1}(\mathbf{w}\,a)|\leq 2(n-7)-1$. So, by $(\ref{eq13})$, we have
\begin{equation}
|\mathcal{X}^{1}_{\overline{a}}|\leq 2(n-8)+2(n-7)-1+2(n-6)-6=6n-49,\label{eq23}
\end{equation}
because of $|\mathcal{X}^{10}_{\overline{a}a}|\leq N(n-6,1,2)=2(n-6-2)$ and $|\mathcal{X}^{11}_{\overline{a}\,\overline{a}}|\leq 2(n-6)-6$ by using Lemma $\ref{eq18}$. Furthermore, $|\mathcal{X}^{0}_{a}|\leq 6n-54$. So, if $s=1$ and $t=0$ then by $(\ref{eq14})$, $(\ref{eq22})$, and $(\ref{eq23})$, it follows that
\begin{align*}
|\mathcal{X}|=|\mathcal{X}^{0}_{a}|+|\mathcal{X}^{0}_{\overline{a}}|+|\mathcal{X}^{1}_{a}|+|\mathcal{X}^1_{\overline{a}}|\leq 6n-54+8n-63+6n-49=20n-166.
\end{align*}

If $s=n-10$ and $t=0$ then by $(\ref{eq12})$ we have
$$|\mathcal{X}^{1}_{a}|\leq 3n-22+|D_{1-\ell^*}(0\,\mathbf{w}_{[1,n-9-\ell^*]})|$$
because of $|\mathcal{X}^{10}_{aa}|\leq |D_{0}(\mathbf{w})|=1$, $|\mathcal{X}^{10}_{a\overline{a}}|\leq |D_{1}(\mathbf{w}\,a)|\leq n-7$, $|\mathcal{X}^{11}_{aa}|=|D_1(0\,\mathbf{v})\cap D_{1}(0\,\mathbf{w})|=2$, $|\mathcal{X}^{11}_{a\overline{a}a}|=|D_2(0\,\mathbf{v})\cap D_{2}(0\,\mathbf{w})|\leq 2(n-7)-4$, and $|\mathcal{X}^{11}_{a\overline{a}\,\overline{a}}|\leq |D_{1-\ell^*}(0\,\mathbf{w}_{[1,n-9-\ell^*]})|$.
When $w_{n-8}=a$ thus $\ell^*\geq 1$ and $|\mathcal{X}^{1}_{a}|\leq 3n-21$. So, if $w_{n-8}=a$ then $(\ref{eq14})$, $(\ref{eq15})$, and $(\ref{eq16})$, it follows that
\begin{align*}
|\mathcal{X}|=|\mathcal{X}^{0}_{a}|+|\mathcal{X}^{0}_{\overline{a}}|+|\mathcal{X}^{1}_{a}|+|\mathcal{X}^1_{\overline{a}}|\leq 6n-54+3n-21+4n-30+6n-48=19n-153\leq 20n-166,
\end{align*}
for any $n\geq 13$.
When $w_{n-8}=\overline{a}$, then we have $\ell^*=0$. Thus $$|\mathcal{X}^{1}_{a}|\leq 2n-15+|D_1(\mathbf{w},a)|+|D_{1}(0\,\mathbf{w}_{[1,n-9]})|.$$ Since $s=n-10$ and $t=0$, we have $v_{n-8}=a$. Then we have
$$|\mathcal{X}^{0}_{\overline{a}}|\leq 2n-16+|D_1(0\,\mathbf{v})|+|D_{1-\ell}(\mathbf{v}_{[2+\ell,n-8]}\,a)|$$
because of $|\mathcal{X}^{00}_{\overline{a}a}|\leq |D_{0}(\mathbf{v})|=1$, $|\mathcal{X}^{00}_{\overline{a}\,\overline{a}}|=|D_1(\mathbf{v}\,a)\cap D_{1}(\mathbf{w}\,a)|=2$, $|\mathcal{X}^{01}_{\overline{a}a}|\leq |D_{1}(0\,\mathbf{v})|$, $|\mathcal{X}^{010}_{\overline{a}\,\overline{a}}|=|D_2(\mathbf{v}\,a)\cap D_{2}(\mathbf{w}\,a)|\leq 2(n-7)-5$, and $|\mathcal{X}^{011}_{\overline{a}\,\overline{a}}|\leq |D_{1-\ell}(\mathbf{v}_{[2+\ell,n-8]}\,a)|$. Since $s=n-10$ and $t=0$, then $\mathbf{v}$ and $\mathbf{w}$ are not both alternating sequences, $\mathbf{w}_{[1,n-9]}$ and $\mathbf{v}_{[2,n-8]}$ are not both alternating sequences. Thus, $|D_1(\mathbf{w}\,a)|+|D_1(0\,\mathbf{v})|\leq 2(n-7)-1$, and $|D_{1}(0\,\mathbf{w}_{[1,n-9]})|+|D_{1}(\mathbf{v}_{[2,n-8]}\,a)|\leq 1+|D_{1}(\mathbf{w}_{[1,n-9]})|+|D_{1}(\mathbf{v}_{[2,n-8]})|\leq 2(n-9)-1+1=2(n-9)$ because of $v_{n-8}=a$. Therefore,
\begin{align}
|\mathcal{X}^{1}_{a}|+|\mathcal{X}^{0}_{\overline{a}}|&\leq (2n-15)+(2n-16)+|D_1(\mathbf{w}\,a)|+|D_1(0\,\mathbf{v})|+|D_{1}(0\,\mathbf{w}_{[1,n-9]})|+|D_{1}(\mathbf{v}_{[2,n-8]}\,a)|\nonumber\\
&\leq (2n-15)+(2n-16)+2(n-7)-1+2(n-9)=8n-64.\label{eq24}
\end{align}
So, if $w_{n-8}=\overline{a}$ then by $(\ref{eq14})$, $(\ref{eq15})$, and $(\ref{eq24})$, it follows that
\begin{align*}
|\mathcal{X}|=|\mathcal{X}^{0}_{a}|+|\mathcal{X}^{0}_{\overline{a}}|+|\mathcal{X}^{1}_{a}|+|\mathcal{X}^1_{\overline{a}}|\leq 6n-54+8n-64+6n-48=20n-166.
\end{align*}

\textbf{Case 3}: When $s=0$ and $t\geq 1$, by using the similar method of proving the case of $s\geq 1$ and $t=0$, we also obtain that $|\mathcal{X}|\leq 20n-166$.

\textbf{Case 4}: When $s=t=0$, if $\ell^*\geq 1$ or $\ell \geq 1$ then $|\mathcal{X}|\leq 19n-153\leq 20n-166$ for any $n\geq 13$. Then we only consider $\ell=\ell^*=0$. That is, $v_1=1$ and $w_{n-8}=\overline{a}$. Since $s=t=0$ then $\mathbf{v}$ is an alternating sequence with length of $n-8$ that starts $a$ and ends with $1$, and  $\mathbf{w}$ is also an alternating sequence with length of $n-8$ that starts $\overline{a}$ and ends with $0$.
Thus,
$$|\mathcal{X}^{0}_{\overline{a}}|\leq 2n-16+n-7+n-9=4n-32$$
because of $|\mathcal{X}^{00}_{\overline{a}a}|\leq |D_{0}(\mathbf{v})|=1$, $|\mathcal{X}^{00}_{\overline{a}\,\overline{a}}|=|D_1(\mathbf{v}\,a)\cap D_{1}(\mathbf{w}\,a)|=2$, $|\mathcal{X}^{01}_{\overline{a}a}|\leq |D_{1}(0\,\mathbf{v})|=n-7$, $|\mathcal{X}^{010}_{\overline{a}\,\overline{a}}|=|D_2(\mathbf{v}\,a)\cap D_{2}(\mathbf{w}\,a)|\leq 2(n-7)-5$, and $|\mathcal{X}^{011}_{\overline{a}\,\overline{a}}|\leq |D_{1}(\mathbf{v}_{[2,n-8]}\,a)|=n-9$. We also have
$$|\mathcal{X}^{1}_{a}|\leq 1+n-7+2+2(n-7)-5+n-9=4n-32$$
because of $|\mathcal{X}^{10}_{aa}|\leq |D_{0}(\mathbf{w})|=1$, $|\mathcal{X}^{10}_{a\overline{a}}|\leq |D_{1}(\mathbf{w}\,a)|=n-7$, $|\mathcal{X}^{11}_{aa}|=|D_1(0\,\mathbf{v})\cap D_{1}(0\,\mathbf{w})|=2$, $|\mathcal{X}^{11}_{a\overline{a}a}|=|D_2(0\,\mathbf{v})\cap D_{2}(0\,\mathbf{w})|\leq 2(n-7)-5$, and $|\mathcal{X}^{11}_{a\overline{a}\,\overline{a}}|\leq |D_{1}(0\,\mathbf{w}_{[1,n-9]})|=n-9$.
So, if $s=t=0$ then by $(\ref{eq14})$ and $(\ref{eq15})$, it follows that
\begin{align*}
|\mathcal{X}|=|\mathcal{X}^{0}_{a}|+|\mathcal{X}^{0}_{\overline{a}}|+|\mathcal{X}^{1}_{a}|+|\mathcal{X}^1_{\overline{a}}|\leq 6n-54+8n-64+6n-48=20n-166.
\end{align*}

\textbf{Case B): $b=a$}. Then, we have
$$|\mathcal{X}^{0}_{a}|=|D_3(\mathbf{x}^{0}_{a})\cap D_3(\mathbf{y}^{0}_{a})|=|D_3(1\,0\,\mathbf{v}\,\overline{a}\,a\,\overline{a})\cap D_3(1\,1\,0\,\mathbf{w}\,\overline{a}\,a)|,$$
$$|\mathcal{X}^{0}_{\overline{a}}|=|D_2(\mathbf{x}^{0}_{\overline{a}})\cap D_4(\mathbf{y}^{0}_{\overline{a}})|=|D_2(1\,0\,\mathbf{v}\,\overline{a}\,a)\cap D_4(1\,1\,0\,\mathbf{w}\,\overline{a}\,a\,a)|,$$
$$|\mathcal{X}^{1}_{a}|=|D_4(\mathbf{x}^{1}_{a})\cap D_3(\mathbf{y}^{1}_{a})|=|D_4(0\,1\,0\,\mathbf{v}\,\overline{a}\,a\,\overline{a})\cap D_2(1\,0\,\mathbf{w}\,\overline{a}\,a)|,$$
$$|\mathcal{X}^{1}_{\overline{a}}|=|D_3(\mathbf{x}^{1}_{\overline{a}})\cap D_3(\mathbf{y}^{1}_{\overline{a}})|=|D_3(0\,1\,0\,\mathbf{v}\,\overline{a}\,a)\cap D_3(1\,0\,\mathbf{w}\,\overline{a}\,a\,a)|.$$
Moreover, if $d_L(\mathbf{x}^{0}_{a},\mathbf{y}^{0}_{a})\leq 1$ then $d_L(\mathbf{x},\mathbf{y})\leq 2$. Thus, $d_L(\mathbf{x}^{0}_{a},\mathbf{y}^{0}_{a})\geq 2$. Similarly, it is easily verified that $\mathbf{x}^{0}_{\overline{a}} \notin D_2(\mathbf{y}^{0}_{\overline{a}}), \mathbf{y}^{1}_{a} \notin D_2(\mathbf{x}^{1}_{a}),$ and $d_L(\mathbf{x}^{1}_{\overline{a}},\mathbf{y}^{1}_{\overline{a}})\geq 2$.

By Lemma $\ref{lm6}$, we have
\begin{align}
|\mathcal{X}^{0}_{a}|\leq 6(n-3)-33=6n-51,\label{eq25}\\
|\mathcal{X}^{1}_{\overline{a}}|\leq 6(n-3)-33=6n-51.\label{eq26}
\end{align}

Then, $\mathcal{X}^{1}_{a}$ can be decomposed as follows:
\begin{align}
|\mathcal{X}^{1}_{a}|&=|\mathcal{X}^{10}_{aa}|+|\mathcal{X}^{10}_{a\overline{a}}|+|\mathcal{X}^{11}_{aa}|+|\mathcal{X}^{11}_{a\overline{a}}|,\nonumber\\
|\mathcal{X}^{10}_{aa}|&=|D_3(\mathbf{x}^{10}_{aa})\cap D_{1}(\mathbf{y}^{10}_{aa})|=|D_3(1\,0\,\mathbf{v}\,\overline{a})\cap D_{1}(\mathbf{w}\,\overline{a})|,\nonumber\\
|\mathcal{X}^{10}_{a\overline{a}}|&=|D_4(\mathbf{x}^{10}_{a\overline{a}})\cap D_{0}(\mathbf{y}^{10}_{a\overline{a}})|=|D_4(1\,0\,\mathbf{v}\,\overline{a},a)\cap D_{0}(\mathbf{w})|,\nonumber\\
|\mathcal{X}^{11}_{aa}|&=|D_2(\mathbf{x}^{11}_{aa})\cap D_{2}(\mathbf{y}^{11}_{aa})|=|D_2(0\,\mathbf{v}\,\overline{a})\cap D_{2}(0\,\mathbf{w}\,\overline{a})|,\nonumber\\
|\mathcal{X}^{11}_{a\overline{a}}|&=|D_3(\mathbf{x}^{11}_{a\overline{a}})\cap D_{1}(\mathbf{y}^{11}_{a\overline{a}})|=|D_3(0\,\mathbf{v}\,\overline{a})\cap D_{1}(0\,\mathbf{w})|.\label{eq27}
\end{align}
Moreover, $\mathcal{X}^{0}_{\overline{a}}$ can be decomposed as follows:
\begin{align}
|\mathcal{X}^{0}_{\overline{a}}|&=|\mathcal{X}^{00}_{\overline{a}a}|+|\mathcal{X}^{00}_{\overline{a}\,\overline{a}}|+|\mathcal{X}^{010}_{\overline{a}aa}|+|
\mathcal{X}^{010}_{\overline{a}a\overline{a}}|+|\mathcal{X}^{011}_{\overline{a}aa}|+|\mathcal{X}^{011}_{\overline{a}a\overline{a}}|
+|\mathcal{X}^{01}_{\overline{a}a}|+|\mathcal{X}^{01}_{\overline{a}\,\overline{a}}|,\nonumber\\
|\mathcal{X}^{00}_{\overline{a}a}|&=|D_1(\mathbf{x}^{00}_{\overline{a}a})\cap D_{2}(\mathbf{y}^{00}_{\overline{a}a})|=|D_1(\mathbf{v}\,\overline{a})\cap D_{2}(\mathbf{w}\,\overline{a}\,a)|,\nonumber\\
|\mathcal{X}^{00}_{\overline{a}\,\overline{a}}|&=|D_0(\mathbf{x}^{00}_{\overline{a}\,\overline{a}})\cap D_{0}(\mathbf{y}^{00}_{\overline{a}\,\overline{a}})|=|D_0(\mathbf{v})\cap D_{0}(\mathbf{w})|,\nonumber\\
|\mathcal{X}^{010}_{\overline{a}aa}|&=|D_{1-\ell_2}(\mathbf{x}^{010}_{\overline{a}aa})\cap D_{3}(\mathbf{y}^{010}_{\overline{a}aa})|=|D_{1-\ell_2}(\mathbf{v}_{[1,n-9-\ell_2]})\cap D_{3}(\mathbf{w}\,\overline{a})|,\nonumber\\
|\mathcal{X}^{010}_{\overline{a}a\overline{a}}|&=|D_2(\mathbf{x}^{010}_{\overline{a}a\overline{a}})\cap D_{2}(\mathbf{y}^{010}_{\overline{a}a\overline{a}})|=|D_2(\mathbf{v})\cap D_{2}(\mathbf{w})|,\nonumber\\
|\mathcal{X}^{011}_{\overline{a}aa}|&=|D_{0-\ell_1-\ell_2}(\mathbf{x}^{011}_{\overline{a}aa})\cap D_{4}(\mathbf{y}^{011}_{\overline{a}aa})|=|D_{0-\ell_1-\ell_2}(\mathbf{v}_{[2+\ell_1,n-9-\ell_2]})\cap D_{4}(0\,\mathbf{w}\,\overline{a})|,\nonumber\\
|\mathcal{X}^{011}_{\overline{a}a\overline{a}}|&=|D_{1-\ell_1}(\mathbf{x}^{011}_{\overline{a}a\overline{a}})\cap D_{3}(\mathbf{y}^{011}_{\overline{a}a\overline{a}})|=|D_{1-\ell_1}(\mathbf{v}_{[2+\ell_1,n-8]})\cap D_{3}(0\,\mathbf{w})|,\nonumber\\
|\mathcal{X}^{01}_{\overline{a}\,\overline{a}}|&=|D_1(\mathbf{x}^{01}_{\overline{a}\,\overline{a}})\cap D_{2}(\mathbf{y}^{01}_{\overline{a}\,\overline{a}})|=|D_1(0\,\mathbf{v})\cap D_{2}(1\,0\,\mathbf{w})|,\label{eq28}
\end{align}
where $\ell_1,\ell_2\geq 0$. If $\ell_1=0$ then $v_{1}=1$; if $\ell_2=0$ then $v_{n-8}=a$.

When $\ell_1\geq 1$ or $\ell_2\geq 1$, it is easily verified that $|\mathcal{X}^{1}_{a}|+|\mathcal{X}^{0}_{\overline{a}}|\leq 8n-64$ for any $n\geq 13$.

Consider $\ell_1=\ell_2=0$. If $|D_1(\mathbf{v})\cap D_1(\mathbf{w})|=1$ then we also obtain $|\mathcal{X}^{1}_{a}|+|\mathcal{X}^{0}_{\overline{a}}|\leq 8n-64$ for any $n\geq 13$.
When $|D_1(\mathbf{v})\cap D_1(\mathbf{w})|=2$, by Lemma $\ref{lm13}$, then $\mathbf{v}$ and $\mathbf{w}$ are Type-A-confusable. For convenience, let $\mathbf{v}=\mathbf{v}_1\,\mathbf{u}\,\mathbf{w}_2$ and $\mathbf{w}=\mathbf{v}_1\,\overline{\mathbf{u}}\,\mathbf{w}_2$, where $|\mathbf{v}_1|=s$, $|\mathbf{u}|=l$, $|\mathbf{w}_2|=t$, and $s+l+t=n-8$, where $\mathbf{u}$ is an alternating sequence. For any $n\geq 13$, we can also obtain that $|\mathcal{X}^{1}_{a}|+|\mathcal{X}^{0}_{\overline{a}}|\leq 8n-64$ based on different values of $s$ and $t$ by Lemma $\ref{lm16}$. Thus, when $b=a$, we can prove that $|\mathcal{X}|\leq 20n-166$ for any $n\geq 13$.
\end{proof}

\section{Conclusion}
\label{sec5}

In this paper, we study the sequence reconstruction problem over deletion channels. We proposed the lower bound on $N(n,3,t)$, that is, $N(n,3,t)\geq M(n,t)$ for $n\geq \max\{13,t+8\}$. Furthermore, we determined the values of $N(n,3,4)$ for any $n\geq 5$. For any $t\geq 5$ and  $n\geq t+8$, we propose the following conjecture.

\begin{conjecture}
For any $t\geq 5$ and $n\geq t+8$, we have
$$N(n,3,t)=M(n,t).$$
\end{conjecture}

\appendices
\section{The case of $n=5,6,7,8,9,10,11,12$}\label{APP-A}
The purpose of this appendix is to give two length-$n$ sequences $\mathbf{x}$ and $\mathbf{y}$ such that the cardinality of the intersection of their $4$-deletion balls is $N(n,3,4)$.

For $n=5,6,7,8$, we have $N(n,3,4)=2^{n-4}$. Specifically, we can find $\mathbf{x}$, $\mathbf{y}$ such that $d_L(\mathbf{x},\mathbf{y})=3$ and $N(n,3,4)=|D_4(\mathbf{x})\cap D_4(\mathbf{y})|$ for all $5\leq n\leq 12$ as follows.
\begin{itemize}
    \item When $n=5$, if $\mathbf{x}=0\,0\,0\,1\,0$ and $\mathbf{y}=1\,1\,1\,0\,1$ then $N(5,3,4)=|D_4(\mathbf{x})\cap D_4(\mathbf{y})|=2$.
    \item When $n=6$, if $\mathbf{x}=0\,1\,0\,1\,1\,0$ and $\mathbf{y}=1\,1\,0\,0\,0\,1$ then $N(6,3,4)=|D_4(\mathbf{x})\cap D_4(\mathbf{y})|=4$.
    \item When $n=7$, if $\mathbf{x}=0\,1\,0\,1\,1\,1\,0$ and $\mathbf{y}=1\,1\,0\,0\,1\,0\,1$ then $N(7,3,4)=|D_4(\mathbf{x})\cap D_4(\mathbf{y})|=8$.
    \item When $n=8$, if $\mathbf{x}=0\,1\,1\,0\,1\,0\,0\,1$ and $\mathbf{y}=1\,0\,0\,1\,0\,1\,1\,0$ then $N(8,3,4)=|D_4(\mathbf{x})\cap D_4(\mathbf{y})|=16$.
    \item When $n=9$, if $\mathbf{x}=0\,1\,1\,0\,0\,1\,0\,1\,0$ and $\mathbf{y}=1\,0\,1\,0\,1\,1\,0\,0\,1$ then $N(9,3,4)=|D_4(\mathbf{x})\cap D_4(\mathbf{y})|=26.$
    \item When $n=10$, if $\mathbf{x}=0\,1\,1\,0\,0\,1\,1\,0\,0\,1$ and $\mathbf{y}=1\,0\,1\,0\,1\,0\,1\,0\,1\,0$ then $N(10,3,4)=|D_4(\mathbf{x})\cap D_4(\mathbf{y})|=40$.
    \item When $n=11$, if $\mathbf{x}=0\,1\,1\,0\,0\,1\,1\,0\,1\,0\,1$ and $\mathbf{y}=1\,0\,1\,0\,1\,0\,1\,0\,1\,1\,0$ then $N(11,3,4)=|D_4(\mathbf{x})\cap D_4(\mathbf{y})|=57$.
    \item When $n=12$, if $\mathbf{x}=0\,1\,1\,0\,0\,1\,0\,1\,0\,1\,0\,1$ and $\mathbf{y}=1\,0\,1\,0\,1\,0\,1\,0\,0\,1\,1\,0$ then $N(12,3,4)=|D_4(\mathbf{x})\cap D_4(\mathbf{y})|=75$.
\end{itemize}

\section{Proof of Lemma \ref{lm3}}\label{APP-C}
The purpose of this appendix is to give the proof of Lemma \ref{lm3}.

\begin{proof}{ \rm (The proof of Lemma $\ref{lm3}$)}
We let $\mathbf{x}=x_1\,\cdots\,x_{n_1}$ and $y_1\,\cdots\,y_{n_2}$. 

First consider the case of $n_1=n_2+1$ and $\mathbf{y}\notin D_1(\mathbf{x})$. Let $\mathbf{u}$ denote the longest common prefix of $\mathbf{x}$ and $\mathbf{y}$, then $\mathbf{x}=\mathbf{u}\,\mathbf{x'}, \mathbf{y}=\mathbf{u}\,\mathbf{y'}$ such that $\mathbf{y}'\notin D_1(\mathbf{x}')$. We denote $\mathcal{X}=D_2(\mathbf{x})\cap D_1(\mathbf{y})$. 
\begin{itemize}
  \item If $|\mathbf{u}|=0$, then $\mathcal{X}^{x_1}\subset D_1(\mathbf{y})^{x_1}\subset D_0(y_3\,\cdots\,y_{n_2})$ and $\mathcal{X}^{\overline{x_1}}\subset \big(D_1(x_2\,\cdots\, x_{n_1})\cap D_1(\mathbf{y})\big)$. Thus, $|\mathcal{X}|=|\mathcal{X}^{x_1}|+|\mathcal{X}^{\overline{x_1}}|\leq 1+2\leq 3$ because of $x_2\,\cdots\, x_{n_1}\neq \mathbf{y}$.
  \item If $|\mathbf{u}|>0$, by Lemma $\ref{lm4}$, the intersection can be decomposed as
  \begin{align*}
    \mathcal{X}&=\bigcup\limits_{p\in\{0,1\}} D_p(\mathbf{u})\circ \big(D_{2-p}(\mathbf{x}')\cap D_{1-p}(\mathbf{y}')\big)\\
    &=\mathbf{u}\circ \big(D_{2}(\mathbf{x}')\cap D_{1}(\mathbf{y}')\big).
  \end{align*}
  Since $|D_{2}(\mathbf{x}')\cap D_{1}(\mathbf{y}')|\leq 3$, then $|\mathcal{X}|\leq 3$.
\end{itemize}
Denote $\mathcal{Y}=D_3(\mathbf{x})\cap D_2(\mathbf{y})$.
\begin{itemize}
  \item If $|\mathbf{u}|=0$, then $\mathcal{Y}^{x_1}\subset D_2(\mathbf{y})^{x_1}\subset D_1(y_3\,\cdots\,y_{n_2})$. Moreover, $|\mathcal{X}^{\overline{x}_1}|=|D_{2-\ell^*}(x_{3+\ell^*}\,\cdots\,x_{n_1})\cap D_{2}(y_{2}\,\cdots\,y_{n_2})|$ such that $x_{3}\,\cdots\,x_{n_1}\neq y_{2}\,\cdots\,y_{n_2}$ for $\ell=0$, where $\ell^*\geq 0$ is an integer. Then, $|\mathcal{Y}^{\overline{x_1}}|\leq \max\{|D_2(x_3\,\cdots\, x_{n_1})\cap D_2(y_2\,\cdots\,y_{n_2})|,n_2-2\}$. Thus, $|\mathcal{Y}|=|\mathcal{Y}^{x_1}|+|\mathcal{Y}^{\overline{x_1}}|\leq |D_1(y_3\,\cdots\,y_{n_2})|+N(n_2-1,1,2)=n_2-2+2(n_2-3)=3n_2-8$ for $n_2\geq4$.
  \item If $|\mathbf{u}|>0$, by Lemma $\ref{lm4}$, the intersection can be decomposed as
  \begin{align*}
    \mathcal{Y}&=\bigcup\limits_{p\in\{0,1,2\}} D_p(\mathbf{u})\circ \big(D_{3-p}(\mathbf{x}')\cap D_{2-p}(\mathbf{y}')\big)\\
    &=\bigg(\mathbf{u}\circ \big(D_{3}(\mathbf{x}')\cap D_{2}(\mathbf{y}')\big)\bigg)\cup \bigg(D_1(\mathbf{u})\circ \big(D_{2}(\mathbf{x}')\cap D_{1}(\mathbf{y}')\big)\bigg).
  \end{align*}
  Given that $|D_{3}(\mathbf{x}')\cap D_{2}(\mathbf{y}')|\leq 3(n_2-|\mathbf{u}|)-8, D_1(\mathbf{u})\leq |\mathbf{u}|, |D_{2}(\mathbf{x}')\cap D_{1}(\mathbf{y}')|\leq 3$, one can derive $|\mathcal{Y}|\leq 3n_2-8$.
\end{itemize}

Second consider the case of $n_1=n_2+2$ and $\mathbf{y}\notin D_2(\mathbf{x})$. Similarly, we have $|D_3(\mathbf{x})\cap D_1(\mathbf{y})|\leq 4$, and $|D_4(\mathbf{x})\cap D_2(\mathbf{y})|\leq 4n_2-13$ for any $n_2\geq 5$.

\end{proof}

\section{Proof of Lemmas $\ref{lm6}$ and $\ref{lm7}$}\label{APP-D}
The purpose of this appendix is to give the proof of Lemmas $\ref{lm6}$ and $\ref{lm7}$.

\begin{proof}{ \rm (The proof of Lemma $\ref{lm6}$)}
Let $\mathbf{x}=x_1\,x_2\,\cdots\,x_n$, $\mathbf{y}=y_1\,y_2\,\cdots\,y_n$, $\mathbf{u}=u_1\,\cdots\,u_s$, $\mathbf{w}=w_1\,\cdots\,w_t$, $\mathbf{v}=v_1\,\cdots\, v_l$, $\mathbf{\widetilde{v}}=\widetilde{v}_1\,\cdots\,\widetilde{v}_l$, $\mathcal{X}=D_3(\mathbf{x})\cap D_3(\mathbf{y})$. By Lemma $\ref{lm4}$, the intersection can be decomposed as
\begin{equation}
D_3(\mathbf{x})\cap D_3(\mathbf{y})=\bigcup\limits_{p+q\leq 3}D_p(\mathbf{u})\circ \big(D_{3-p-q}(\mathbf{v})\cap D_{3-p-q}(\mathbf{\widetilde{v}})) \circ D_{q}(\mathbf{w}).\nonumber
\end{equation}
Since $d_L(\mathbf{x},\mathbf{y})\geq 2$, then $d_L(\mathbf{v},\mathbf{\widetilde{v}})\geq 2$ such that
\begin{align*}
D_3(\mathbf{x})\cap D_3(\mathbf{y})&=\big(D_1(\mathbf{u})\circ \big(D_{2}(\mathbf{v})\cap D_{2}(\mathbf{\widetilde{v}})) \circ \mathbf{w}\big)\cup
 \big(\mathbf{u}\circ \big(D_{2}(\mathbf{v})\cap D_{2}(\mathbf{\widetilde{v}})) \circ D_1(\mathbf{w})\big)\cup \big(\mathbf{u}\circ \big(D_{3}(\mathbf{v})\cap D_{3}(\mathbf{\widetilde{v}})) \circ \mathbf{w}\big)\\
&=\bigcup\limits_{\mathbf{u}_1\in D_1(\mathbf{u})}\mathcal{X}_{\mathbf{w}}^{\mathbf{u}_1}\cup \mathcal{X}_{\mathbf{w}}^{\mathbf{u}}\cup \bigcup\limits_{\mathbf{w}_1\in D_1(\mathbf{w})}\mathcal{X}_{\mathbf{w}_1}^{\mathbf{u}},
\end{align*}
where $\bigcup\limits_{\mathbf{u}_1\in D_1(\mathbf{u})}\mathcal{X}_{\mathbf{w}}^{\mathbf{u}_1}=D_1(\mathbf{u})\circ \big(D_{2}(\mathbf{v})\cap D_{2}(\mathbf{\widetilde{v}})) \circ \mathbf{w}$, $\mathcal{X}_{\mathbf{w}}^{\mathbf{u}}=\mathbf{u}\circ \big(D_{3}(\mathbf{v})\cap D_{3}(\mathbf{\widetilde{v}})) \circ \mathbf{w}$, $\bigcup\limits_{\mathbf{w}_1\in D_1(\mathbf{w})}\mathcal{X}_{\mathbf{w}_1}^{\mathbf{u}}=\mathbf{u}\circ \big(D_{2}(\mathbf{v})\cap D_{2}(\mathbf{\widetilde{v}})) \circ D_1(\mathbf{w})$.

Observe that $\mathbf{u}_{[1,s-1]}\in D_1(\mathbf{u})$ if $s\geq 1$, and $\mathbf{w}_{[2,t]}\in D_1(\mathbf{w})$ if $t\geq 1$. Next, we will prove that $|\mathcal{X}_{\mathbf{w}}^{\mathbf{u}_{[1,s-1]}}\backslash \mathcal{X}_{\mathbf{w}}^{\mathbf{u}}|\leq 3$ and $|\mathcal{X}_{\mathbf{w}_{[2,t]}}^{\mathbf{u}}\backslash \mathcal{X}_{\mathbf{w}}^{\mathbf{u}}|\leq 3$. Without loss of generality, we only consider the value of $|\mathcal{X}_{\mathbf{w}}^{\mathbf{u}_{[1,s-1]}}\backslash \mathcal{X}_{\mathbf{w}}^{\mathbf{u}}|$. By Proposition $\ref{prp1}$, it follows that $|\mathcal{X}_{\mathbf{w}}^{\mathbf{u}_{[1,s-1]}}|=|\mathcal{X}_{\mathbf{w}}^{\mathbf{u}}|+|\mathcal{X}_{\mathbf{w}}^{\mathbf{u}_{[1,s-1]}\overline{u_s}}|$. Since $v_1\neq \widetilde{v}_1$, then $\overline{u_s}=v_1$ or $\widetilde{v}_1$. Without loss of generality, we consider the case where $v_1=\overline{u_s}$. Then, $\mathcal{X}_{\mathbf{w}}^{\mathbf{u}_{[1,s-1]}\overline{u_s}}=D_2(v_2\,\cdots\, v_l) \cap D_{1-\ell}(\widetilde{v}_{3+\ell}\,\cdots\,\widetilde{v}_l)$ with $\ell\geq 0$.
Thus, 
$$|\mathcal{X}_{\mathbf{w}}^{\mathbf{u}_{[1,s-1]}\overline{u_s}}|\leq \max\{ |D_2(v_2\,\cdots\, v_l) \cap D_1(\widetilde{v}_3\,\cdots\,\widetilde{v}_l)|,1\}.$$
Since $d_L(\mathbf{v},\mathbf{\widetilde{v}})\geq 2$, then $\widetilde{v}_3\,\cdots\,\widetilde{v}_l\notin D_1(v_2\,\cdots\, v_l)$ for $\ell=0$. By Lemma $\ref{lm3}$, we have $|\mathcal{X}_{\mathbf{w}}^{\mathbf{u}_{[1,s-1]}}\backslash \mathcal{X}_{\mathbf{w}}^{\mathbf{u}}|=|\mathcal{X}_{\mathbf{w}}^{\mathbf{u}_{[1,s-1]}\overline{u_s}}|\leq 3$ for $s\geq1$. Similarly, we also obtain that $|\mathcal{X}_{\mathbf{w}_{[2,t]}}^{\mathbf{u}}\backslash \mathcal{X}_{\mathbf{w}}^{\mathbf{u}}|\leq 3$ for $t\geq1$.

If $s,t\geq1$, by the decomposition of $\mathcal{X}$, then 
it follows that   
\begin{align*}
|\mathcal{X}|&=|\bigcup\limits_{\mathbf{u}_1\neq \mathbf{u}_{[1,s-1]}, \mathbf{u}_1\in D_1(\mathbf{u})}\mathcal{X}_{\mathbf{w}}^{\mathbf{u}_1}|+|\mathcal{X}_{\mathbf{w}}^{\mathbf{u}}|+|\bigcup\limits_{\mathbf{w}_1\neq \mathbf{w}_{[2,t]}, \mathbf{w}_1\in D_1(\mathbf{w})}\mathcal{X}_{\mathbf{w}_1}^{\mathbf{u}}|+|\mathcal{X}_{\mathbf{w}_{[2,t]}}^{\mathbf{u}}\backslash \mathcal{X}_{\mathbf{w}}^{\mathbf{u}}|+|\mathcal{X}_{\mathbf{w}}^{\mathbf{u}_{[1,s-1]}}\backslash \mathcal{X}_{\mathbf{w}}^{\mathbf{u}}|\\
&\leq (|D_1(\mathbf{u})|+|D_1(\mathbf{w})|-2)|D_{2}(\mathbf{v})\cap D_{2}(\mathbf{\widetilde{v}})|+|D_{3}(\mathbf{v})\cap D_{3}(\mathbf{\widetilde{v}})|+6\\
&\leq (s+t-2)|D_{2}(\mathbf{v})\cap D_{2}(\mathbf{\widetilde{v}})|+|D_{3}(\mathbf{v})\cap D_{3}(\mathbf{\widetilde{v}})|+6.
\end{align*}
By using a computerized search, we can obtain $N(2,2,2)=1$, $N(3,2,2)=2$, $N(4,2,2)=4$, $N(5,2,2)=4$, $N(2,2,3)=0$, $N(3,2,3)=1$, $N(4,2,3)=2$, $N(5,2,3)=4$, $N(6,2,3)=8$, and $N(7,2,3)=13$. By Theorem $\ref{thm3}$ and Corollary $\ref{cor1}$, it follows that $N(n,2,2)=6$ for $n\geq 6$ and $N(n,2,3)=6n-30$ for $n\geq 8$. When $s,t\geq1$, by the value of $l$, we have the following results:
\begin{enumerate}
    \item If $l=2$ then $$|\mathcal{X}|\leq (s+t)|D_{2}(\mathbf{v})\cap D_{2}(\mathbf{\widetilde{v}})|+|D_{3}(\mathbf{v})\cap D_{3}(\mathbf{\widetilde{v}})|\leq n-2+N(2,2,3)=n-2$$ because of $|D_3(\mathbf{v})\cap D_3(\mathbf{\widetilde{v}})|\leq N(2,2,3)=0$ and $|D_2(\mathbf{v})\cap D_2(\mathbf{\widetilde{v}})|\leq N(2,2,2)=1$;
    \item If $l=3$ then $$|\mathcal{X}|\leq (s+t)|D_{2}(\mathbf{v})\cap D_{2}(\mathbf{\widetilde{v}})|+|D_{3}(\mathbf{v})\cap D_{3}(\mathbf{\widetilde{v}})|\leq 2(n-3)+N(3,2,3)=2n-5$$ because of $|D_3(\mathbf{v})\cap D_3(\mathbf{\widetilde{v}})|\leq N(3,2,3)=1$ and $|D_2(\mathbf{v})\cap D_2(\mathbf{\widetilde{v}})|\leq N(3,2,2)=2$;
    \item If $l=4$ then $$|\mathcal{X}|\leq 4(n-4-2)+6+2=4n-16$$ because of $|D_3(\mathbf{v})\cap D_3(\mathbf{\widetilde{v}})|\leq N(4,2,3)=2$ and $|D_2(\mathbf{v})\cap D_2(\mathbf{\widetilde{v}})|\leq N(4,2,2)=4$;
    \item If $l=5$ then $$|\mathcal{X}|\leq 4(n-5-2)+6+4=4n-18$$ because of $|D_3(\mathbf{v})\cap D_3(\mathbf{\widetilde{v}})|\leq N(5,2,3)=4$ and $|D_2(\mathbf{v})\cap D_2(\mathbf{\widetilde{v}})|\leq N(5,2,2)=4$;
    \item If $l=6$ then $$|\mathcal{X}|\leq 6(n-6-2)+6+N(6,2,3)=6n-34$$ because of $|D_3(\mathbf{v})\cap D_3(\mathbf{\widetilde{v}})|\leq N(6,2,3)=8$ and $|D_2(\mathbf{v})\cap D_2(\mathbf{\widetilde{v}})|\leq N(6,2,2)=6$;
    \item If $l=7$ then $$|\mathcal{X}|\leq 6(n-7-2)+6+N(7,2,3)=6n-35$$ because of $|D_3(\mathbf{v})\cap D_3(\mathbf{\widetilde{v}})|\leq N(7,2,3)=13$ and $|D_2(\mathbf{v})\cap D_2(\mathbf{\widetilde{v}})|\leq N(7,2,2)=6$;
    \item If $l\geq 8$ then $$|\mathcal{X}|\leq 6(n-l-2)+6+6l-30=6n-36$$ because of $|D_3(\mathbf{v})\cap D_3(\mathbf{\widetilde{v}})|\leq N(l,2,3)=6l-30$ and $|D_2(\mathbf{v})\cap D_2(\mathbf{\widetilde{v}})|\leq N(l,2,2)=6$ for $l\geq 8$.
\end{enumerate}
By the above discussion, when $2\leq l\leq 5$, we have $$|D_3(\mathbf{x})\cap D_3(\mathbf{y})|\leq 6n-37,$$ for any $n\geq 10$ (By using a computerized search, it follows that $|D_3(\mathbf{x})\cap D_3(\mathbf{y})|\leq 22$ for $n=10$ and $l=4$). Moreover, if $\mathbf{u},\mathbf{w}$ are not both alternating sequences then $|D_3(\mathbf{x})\cap D_3(\mathbf{y})|\leq 6n-40$ for any $n\geq 10$. Consider $l=6$, we let $\mathbf{v}=1\,0\,1\,0\,1\,0$ and $\mathbf{\widetilde{v}}=0\,1\,1\,0\,0\,1$ with $d_L(\mathbf{v},\mathbf{\widetilde{v}})=2$, and let $\mathbf{u}=\mathbf{a}_s$ and $\mathbf{w}=\mathbf{a}_t$ be alternating sequences. Then, we compute that $|\mathcal{X}|=6n-34$. When $l=7$, we let $\mathbf{v}=1\,0\,1\,0\,1\,1\,0$, $\mathbf{\widetilde{v}}=0\,1\,1\,0\,1\,0\,1$, $\mathbf{u}=\mathbf{a}_s$ and $\mathbf{w}=\mathbf{a}_t$. Then we have $|\mathcal{X}|=6n-35$.

Similarly, when $s=0, t\geq 1$ or $s\geq 1, t=0$, we can prove that  $|D_3(\mathbf{x})\cap D_3(\mathbf{y})|\leq 6n-31$; if $l=7$ then we have $|D_3(\mathbf{x})\cap D_3(\mathbf{y})|\leq 6n-32$; if $l\geq 8$ then we have $|D_3(\mathbf{x})\cap D_3(\mathbf{y})|\leq 6n-33$. Furthermore, when $s=0, t\geq 1$ or $s\geq 1, t=0$, if $\mathbf{u}$ or $\mathbf{w}$ is not an alternating sequence then $|D_3(\mathbf{x})\cap D_3(\mathbf{y})|\leq 6n-37$.
\end{proof}

\begin{proof}{ \rm (The proof of Lemma $\ref{lm7}$)}
Let $n_1=n$, $n_2=n+2$, $\mathbf{x}=x_1\,x_2\,\cdots\,x_n$, $\mathbf{y}=y_1\,y_2\,\cdots\,y_{n+2}$, $\mathbf{u}=u_1\,\cdots\,u_s$, $\mathbf{w}=w_1\,\cdots\,w_t$, $\mathbf{v}=v_1\,\cdots\, v_l$, $\mathbf{\widetilde{v}}=\widetilde{v}_1\,\cdots\,\widetilde{v}_{l+2}$, $\mathcal{X}=D_2(\mathbf{x})\cap D_4(\mathbf{y})$. By Lemma $\ref{lm4}$, the intersection can be decomposed as
\begin{equation}
D_2(\mathbf{x})\cap D_4(\mathbf{y})=\bigcup\limits_{p+q\leq 2}D_p(\mathbf{u})\circ \big(D_{2-p-q}(\mathbf{v})\cap D_{4-p-q}(\mathbf{\widetilde{v}})) \circ D_{q}(\mathbf{w}).\nonumber
\end{equation}
Since $\mathbf{x}\notin D_2(\mathbf{y})$, then $\mathbf{v}\notin D_2(\mathbf{\widetilde{v}})$ such that
\begin{align*}
D_2(\mathbf{x})\cap D_4(\mathbf{y})&=\big(D_1(\mathbf{u})\circ \big(D_{1}(\mathbf{v})\cap D_{3}(\mathbf{\widetilde{v}})) \circ \mathbf{w}\big)\cup
 \big(\mathbf{u}\circ \big(D_{1}(\mathbf{v})\cap D_{3}(\mathbf{\widetilde{v}})) \circ D_1(\mathbf{w})\big)\cup \big(\mathbf{u}\circ \big(D_{2}(\mathbf{v})\cap D_{4}(\mathbf{\widetilde{v}})) \circ \mathbf{w}\big)\\
&=\bigcup\limits_{\mathbf{u}_1\in D_1(\mathbf{u})}\mathcal{X}_{\mathbf{w}}^{\mathbf{u}_1}\cup \mathcal{X}_{\mathbf{w}}^{\mathbf{u}}\cup \bigcup\limits_{\mathbf{w}_1\in D_1(\mathbf{w})}\mathcal{X}_{\mathbf{w}_1}^{\mathbf{u}},
\end{align*}
where $\bigcup\limits_{\mathbf{u}_1\in D_1(\mathbf{u})}\mathcal{X}_{\mathbf{w}}^{\mathbf{u}_1}=D_1(\mathbf{u})\circ \big(D_{1}(\mathbf{v})\cap D_{3}(\mathbf{\widetilde{v}})) \circ \mathbf{w}$, $\mathcal{X}_{\mathbf{w}}^{\mathbf{u}}=\mathbf{u}\circ \big(D_{2}(\mathbf{v})\cap D_{4}(\mathbf{\widetilde{v}})) \circ \mathbf{w}$, $\bigcup\limits_{\mathbf{w}_1\in D_1(\mathbf{w})}\mathcal{X}_{\mathbf{w}_1}^{\mathbf{u}}=\mathbf{u}\circ \big(D_{1}(\mathbf{v})\cap D_{3}(\mathbf{\widetilde{v}})) \circ D_1(\mathbf{w})$.

Observe that $\mathbf{u}_{[1,s-1]}\in D_1(\mathbf{u})$ if $s\geq 1$, and $\mathbf{w}_{[2,t]}\in D_1(\mathbf{w})$ if $t\geq 1$. Next, we will prove that $|\mathcal{X}_{\mathbf{w}}^{\mathbf{u}_{[1,s-1]}}\backslash \mathcal{X}_{\mathbf{w}}^{\mathbf{u}}|\leq 3$ and $|\mathcal{X}_{\mathbf{w}_{[2,t]}}^{\mathbf{u}}\backslash \mathcal{X}_{\mathbf{w}}^{\mathbf{u}}|\leq 3$. Without loss of generality, we only consider the value of $|\mathcal{X}_{\mathbf{w}}^{\mathbf{u}_{[1,s-1]}}\backslash \mathcal{X}_{\mathbf{w}}^{\mathbf{u}}|$. By Proposition $\ref{prp1}$, it follows that $|\mathcal{X}_{\mathbf{w}}^{\mathbf{u}_{[1,s-1]}}|=|\mathcal{X}_{\mathbf{w}}^{\mathbf{u}}|+|\mathcal{X}_{\mathbf{w}}^{\mathbf{u}_{[1,s-1]}\overline{u_s}}|$. Since $v_1\neq \widetilde{v}_1$, then $\overline{u_s}=v_1$ or $\widetilde{v}_1$. Consider the case where $v_1=\overline{u_s}$. Then, $\mathcal{X}_{\mathbf{w}}^{\mathbf{u}_{[1,s-1]}\overline{u_s}}=D_1(v_2\,\cdots\, v_l) \cap D_{2-\ell}(\widetilde{v}_{3+\ell}\,\cdots\,\widetilde{v}_l)$ with $\ell\geq 0$.
Thus, 
$$|\mathcal{X}_{\mathbf{w}}^{\mathbf{u}_{[1,s-1]}\overline{u_s}}|\leq \max\{ |D_1(v_2\,\cdots\, v_l) \cap D_2(\widetilde{v}_3\,\cdots\,\widetilde{v}_{l+2})|,|D_1(v_2\,\cdots\, v_l) \cap D_1(\widetilde{v}_4\,\cdots\,\widetilde{v}_{l+2})|,1\}.$$
Since $\mathbf{v}\notin D_2(\mathbf{\widetilde{v}})$, then $v_2\,\cdots\, v_l\notin D_1(\widetilde{v}_3\,\cdots\,\widetilde{v}_{l+2})$ for $\ell=0$ and $v_2\,\cdots\, v_l\neq \widetilde{v}_4\,\cdots\,\widetilde{v}_{l+2}$ for $\ell=1$. By Lemma $\ref{lm3}$, we have $|\mathcal{X}_{\mathbf{w}}^{\mathbf{u}_{[1,s-1]}}\backslash \mathcal{X}_{\mathbf{w}}^{\mathbf{u}}|=|\mathcal{X}_{\mathbf{w}}^{\mathbf{u}_{[1,s-1]}\overline{u_s}}|\leq 3$ for $s\geq1$. Consider the case where $\widetilde{v}_1=\overline{u_s}$. Then, $\mathcal{X}_{\mathbf{w}}^{\mathbf{u}_{[1,s-1]}\overline{u_s}}=D_{0-\ell}(v_2\,\cdots\, v_l) \cap D_{3}(\widetilde{v}_{2}\,\cdots\,\widetilde{v}_l)$ with $\ell\geq 0$.
Thus,  $|\mathcal{X}_{\mathbf{w}}^{\mathbf{u}_{[1,s-1]}\overline{u_s}}|\leq 1.$ So, we have $|\mathcal{X}_{\mathbf{w}}^{\mathbf{u}_{[1,s-1]}}\backslash \mathcal{X}_{\mathbf{w}}^{\mathbf{u}}|=|\mathcal{X}_{\mathbf{w}}^{\mathbf{u}_{[1,s-1]}\overline{u_s}}|\leq 3$ for $s\geq1$.
Similarly, we also obtain that $|\mathcal{X}_{\mathbf{w}_{[2,t]}}^{\mathbf{u}}\backslash \mathcal{X}_{\mathbf{w}}^{\mathbf{u}}|\leq 3$ for $t\geq1$.

If $s,t\geq1$, by the decomposition of $\mathcal{X}$, then 
it follows that   
\begin{align*}
|\mathcal{X}|&=|\bigcup\limits_{\mathbf{u}_1\neq \mathbf{u}_{[1,s-1]}, \mathbf{u}_1\in D_1(\mathbf{u})}\mathcal{X}_{\mathbf{w}}^{\mathbf{u}_1}|+|\mathcal{X}_{\mathbf{w}}^{\mathbf{u}}|+|\bigcup\limits_{\mathbf{w}_1\neq \mathbf{w}_{[2,t]}, \mathbf{w}_1\in D_1(\mathbf{w})}\mathcal{X}_{\mathbf{w}_1}^{\mathbf{u}}|+|\mathcal{X}_{\mathbf{w}_{[2,t]}}^{\mathbf{u}}\backslash \mathcal{X}_{\mathbf{w}}^{\mathbf{u}}|+|\mathcal{X}_{\mathbf{w}}^{\mathbf{u}_{[1,s-1]}}\backslash \mathcal{X}_{\mathbf{w}}^{\mathbf{u}}|\\
&\leq (|D_1(\mathbf{u})|+|D_1(\mathbf{w})|-2)|D_{1}(\mathbf{v})\cap D_{3}(\mathbf{\widetilde{v}})|+|D_{2}(\mathbf{v})\cap D_{4}(\mathbf{\widetilde{v}})|+6\\
&\leq (s+t-2)|D_{1}(\mathbf{v})\cap D_{3}(\mathbf{\widetilde{v}})|+|D_{2}(\mathbf{v})\cap D_{4}(\mathbf{\widetilde{v}})|+6.
\end{align*}
For convenience, let $f(n)$ be the maximal value of $|D_2(\mathbf{x})\cap D_4(\mathbf{y})|$, where $\mathbf{x}\in \mathbb{F}_2^n$, $\mathbf{y}\in \mathbb{F}_2^{n+2}$, and $\mathbf{x}\notin D_2(\mathbf{y})$. By using a computerized search, we can obtain $f(2)=1$, $f(3)=2$, $f(4)=4$, $f(5)=7$, $f(6)=11$. When $s,t\geq1$, by the value of $l$, we have the following results:
\begin{enumerate}
    \item If $l=2$ then $$|\mathcal{X}|\leq (s+t)|D_{1}(\mathbf{v})\cap D_{3}(\mathbf{\widetilde{v}})|+f(2)\leq 2(n-2)+1=2n-3$$ because of $|D_1(\mathbf{v})\cap D_3(\mathbf{\widetilde{v}})|\leq 2$ and $|D_2(\mathbf{v})\cap D_2(\mathbf{\widetilde{v}})|\leq f(2)=1$ by using a computerized search;
    \item If $l=3$ then $$|\mathcal{X}|\leq (s+t)|D_{1}(\mathbf{v})\cap D_{3}(\mathbf{\widetilde{v}})|+f(3)\leq 3(n-3)+2=3n-7$$ because of $|D_1(\mathbf{v})\cap D_3(\mathbf{\widetilde{v}})|\leq 3$ and $|D_2(\mathbf{v})\cap D_2(\mathbf{\widetilde{v}})|\leq f(3)=2$ by using a computerized search;
    \item If $l=4$ then $$|\mathcal{X}|\leq 4(n-4-2)+6+4=4n-14$$ because of $|D_1(\mathbf{v})\cap D_3(\mathbf{\widetilde{v}})|\leq 4$ and $f(4)=4$;
    \item If $l\geq 5$ then $$|\mathcal{X}|\leq 4(n-l-2)+6+4l-13=4n-15$$ because of $|D_1(\mathbf{v})\cap D_3(\mathbf{\widetilde{v}})|\leq 3$ and $|D_2(\mathbf{v})\cap D_4(\mathbf{\widetilde{v}})|\leq 4l-13$ for $l\geq 5$.
\end{enumerate}
Therefore, for $n \geq 9$, we have $|\mathcal{X}| \leq 4n-14$, with equality possible only when $l=4$.

Consider the case where $s=0$ and $t\geq 1$, or $s\geq1$ and $t=0$. Similarly, if $l\geq 5$ then we can prove that $|D_2(\mathbf{x})\cap D_4(\mathbf{y})|\leq 4n-14$. Moreover, if $\mathbf{u}$ or $\mathbf{w}$ is not an alternating sequence then $|D_2(\mathbf{x})\cap D_4(\mathbf{y})|\leq 4n-16$ for any $l\geq 5$; $|D_2(\mathbf{x})\cap D_4(\mathbf{y})|\leq 4n-15$ for $l=4$.
\end{proof}

\section{Proof of Lemma $\ref{lm16}$}\label{APP-E}
The purpose of this appendix is to give the proof of Lemma $\ref{lm16}$.
\begin{proof}
Since $d_L(\mathbf{x},\mathbf{y})=1$, we have $|D_1(\mathbf{x})\cap D_1(\mathbf{y})|=1$ or $2$. When  $|D_1(\mathbf{x})\cap D_1(\mathbf{y})|=1$, by Lemma $\ref{lm15}$, we have $|D_2(\mathbf{x})\cap D_2(\mathbf{y})|\leq n$ for any $n\geq 4$.

For any $n\geq 7$, if $|D_2(\mathbf{x}) \cap D_2(\mathbf{y})|=2n-4,2n-5,~\text{or}~2n-6$ then we have $|D_1(\mathbf{x})\cap D_1(\mathbf{y})|=2$ because of $|D_2(\mathbf{x})\cap D_2(\mathbf{y})|\leq n$ with $|D_1(\mathbf{x})\cap D_1(\mathbf{y})|=1.$ By Lemma $\ref{lm13}$, if $|D_1(\mathbf{x})\cap D_1(\mathbf{y})|=2$ then $\mathbf{x}$ and  $\mathbf{y}$ are Type-$A$ confusable. For convenience, we let $\mathbf{x}=\mathbf{(u,a,v)}$ and $\mathbf{y}=\mathbf{(u,\overline{a},v)}$ with $|\mathbf{u}|=s,|\mathbf{a}|=l,|\mathbf{v}|=t$, where $s\geq 0,t\geq 0,l\geq 2, s+l+t=n$, and $\mathbf{a}$ is an alternating sequence. It is easily verified that $D_1(\mathbf{a})\cap D_1(\overline{\mathbf{a}})=\{\mathbf{a}(l-1),\overline{\mathbf{a}(l-1)}\}$ and $|D_2(\mathbf{a})\cap D_2(\overline{\mathbf{a}})|=2(l-2)$. Moreover, for any $l\geq 3$ we have $D_2(\mathbf{a})\cap D_2(\overline{\mathbf{a}})=\{\mathbf{a}(l-2),\overline{\mathbf{a}(l-2)},\mathbf{w}_1,\cdots,\mathbf{w}_{2l-6}\}$, where $\mathbf{w}_i$ is a binary sequence with one run of length $2$ for any $i\in [2l-6]$.
By Lemma $\ref{lm4}$, it follows that
\begin{align}
D_2(\mathbf{x})&\cap D_2(\mathbf{y})=\bigcup\limits_{p+q\leq 2}D_p(\mathbf{u})\circ \big(D_{2-p-q}(\mathbf{a})\cap D_{2-p-q}(\overline{\mathbf{a}})\big) \circ D_{q}(\mathbf{v})\nonumber\\
&=\big(D_1(\mathbf{u})\circ \{\mathbf{a}(l-1),\overline{\mathbf{a}(l-1)}\} \circ\mathbf{v}\big)
\cup \big(\mathbf{u} \circ \{\mathbf{a}(l-1),\overline{\mathbf{a}(l-1)}\} \circ D_1(\mathbf{v})\big)\cup \big(\mathbf{u}\circ (D_2(\mathbf{a})\cap D_2(\overline{\mathbf{a}}))\circ \mathbf{v})\big).\nonumber
\end{align}
Thus, $|D_2(\mathbf{x})\cap D_2(\mathbf{y})|\leq 2(|D_1(\mathbf{u})|+|D_1(\mathbf{v})|)+|D_2(\mathbf{a})\cap D_2(\overline{\mathbf{a}})|$.

When $\mathbf{u}$ and $\mathbf{v}$ are not both alternating sequences, without loss of generality, we assume $\mathbf{u}$ is not an alternating sequence, where $\mathbf{u}=(u_1,u_2,\cdots,u_s)$, $|D_1(\mathbf{u})|\leq s-1$, and $s\geq 2$. For convenience, let $\mathbf{a}^{b}(s)$ be an alternating sequence of length $s$ that starts with $b$, where $b\in \mathbb{F}_2$ and $s>0$. We consider a sequence $\mathbf{z}=(\mathbf{u}_{[1,s-1]},\mathbf{a}^{u_s}(l-1),\mathbf{v})=(\mathbf{u},\mathbf{a}^{\overline{u_s}}(l-2),\mathbf{v}),$ then $\mathbf{z}\in D_1(\mathbf{u})\circ \{\mathbf{a}(l-1),\overline{\mathbf{a}(l-1)}\} \circ\mathbf{v}$ and $\mathbf{z}\in \mathbf{u}\circ (D_2(\mathbf{a})\cap D_2(\overline{\mathbf{a}})) \circ\mathbf{v}$. Thus, $|D_2(\mathbf{x})\cap D_2(\mathbf{y})|\leq 2(|D_1(\mathbf{u})|+|D_1(\mathbf{v})|)+|D_2(\mathbf{a})\cap D_2(\overline{\mathbf{a}})|-1=2n-7$. So, if $|D_2(\mathbf{x})\cap D_2(\mathbf{y})|\geq 2n-6$ then $\mathbf{u}$ and $\mathbf{v}$ are both alternating sequences. 

Next, we discuss the value of $|D_2(\mathbf{x})\cap D_2(\mathbf{y})|$ in some cases of $l=n,n-1,n-2,2,$ or $3\leq l\leq n-3$. Let $\mathbf{u,v}$ be alternating sequences and $\mathcal{X}=D_2(\mathbf{x})\cap D_2(\mathbf{y})$. Then $\mathbf{u}=\mathbf{a}^{u_1}(s)$ and $\mathbf{v}=\mathbf{a}^{v_1}(t)$, where $u_1,v_1\in \mathbb{F}_2$.

If $l=n$ then we have $|D_2(\mathbf{x})\cap D_2(\mathbf{y})|=2n-4$. When $l=n-1$, with loss of generality, let $s=1$ and $\mathbf{u}=u_1$. By Proposition $\ref{prp1}$, we have $\mathcal{X}=\mathcal{X}^{u_1}+\mathcal{X}^{\overline{u_1}}$. Since $\mathbf{x}=u_1\,\mathbf{a}$ and $\mathbf{y}=u_1\,\overline{\mathbf{a}}$ we can obtain $|\mathcal{X}^{u_1}|=|D_2(\mathbf{a})\cap D_2(\mathbf{\overline{a}})|$ and $|\mathcal{X}^{\overline{u_1}}|=|D_1(\mathbf{a}^{u_1}(n-2))\cap D_0(\mathbf{a}^{\overline{u_1}}(n-3))|$. Thus, $|\mathcal{X}|=|\mathcal{X}^{u_1}|+|\mathcal{X}^{\overline{u_1}}|=2n-6+1=2n-5$. If $l=n-2$ and $\{s,t\}=\{0,2\}$ then we have $|D_2(\mathbf{x})\cap D_2(\mathbf{y})|=2n-5$. If $l=n-2$ and $\{s,t\}=\{1\}$ then by Proposition $\ref{prp1}$ it is easily verified that $|D_2(\mathbf{x})\cap D_2(\mathbf{y})|=2n-6$.

When $l=2$, if $\min\{s,t\}=0$ then $|D_2(\mathbf{x})\cap D_2(\mathbf{y})|=2n-4$. Consider $l=2$ and $s,t\geq 1$. If $\mathbf{u,v}$ are alternating sequences then $D_1(\mathbf{u})=\{\mathbf{a}^{u_1}(s-1),\mathbf{a}^{\overline{u_1}}(s-1),\mathbf{w}_1^{(1)},\cdots,\mathbf{w}_{s-2}^{(1)}\}$ and $D_1(\mathbf{v})=\{\mathbf{a}^{v_1}(t-1),\mathbf{a}^{\overline{v_1}}(t-1),\mathbf{w}_1^{(2)},\cdots,\mathbf{w}_{t-2}^{(2)}\}$, where $\mathbf{w}_i^{(j)}$ is a binary sequence with one run of length $2$ for some $i=1$ and $j\in [s-2]$ or some $i=2$ and $j\in [t-2]$. Specially, when $s=1$ or $t=1$, we have $D_1(\mathbf{u})\circ \mathcal{C}=\mathcal{C}$ or $\mathcal{C}\circ D_1(\mathbf{v})=\mathcal{C}$ for any subset $\mathcal{C}\in \mathbb{F}_2^{m}$, where $m\geq 1$ is an integer. We have
\begin{align*}
D_1(\mathbf{u})\circ \{\mathbf{a}^{1}(l-1),\mathbf{a}^{0}(l-1)\} \circ\mathbf{v}&=\{\mathbf{a}^{i}(s-1)\circ \mathbf{a}^{j}(l-1)\circ \mathbf{a}^{v_1}(t):i,j\in\{0,1\}\}\\
&~~~~~~~\bigcup \{\mathbf{w}^{(1)}_{i}\circ \mathbf{a}^{j}(l-1)\circ \mathbf{a}^{v_1}(t):i\in [s-2],j\in\{0,1\}\}=A_1\cup B_1,
\end{align*}
where $A_1=\{\mathbf{a}^{i}(s-1)\circ \mathbf{a}^{j}(l-1)\circ \mathbf{a}^{v_1}(t):i,j\in\{0,1\}\}$ and $B_1= \{\mathbf{w}^{(1)}_{i}\circ \mathbf{a}^{j}(l-1)\circ \mathbf{a}^{v_1}(t):i\in [s-2],j\in\{0,1\}\}$.
\begin{align*}
\mathbf{u}\circ \{\mathbf{a}^{1}(l-1),\mathbf{a}^{0}(l-1)\} \circ D_1(\mathbf{v})&=\{\mathbf{a}^{u_1}(s)\circ \mathbf{a}^{j}(l-1)\circ \mathbf{a}^{i}(t-1):i,j\in\{0,1\}\}\\
&~~~~~~~\bigcup \{\mathbf{a}^{u_1}(s)\circ \mathbf{a}^{j}(l-1)\circ \mathbf{w}^{(2)}_{i}:i\in [t-2],j\in\{0,1\}\}=A_2\cup B_2,
\end{align*}
where $A_2=\{\mathbf{a}^{u_1}(s)\circ \mathbf{a}^{j}(l-1)\circ \mathbf{a}^{i}(t-1):i,j\in\{0,1\}\}$ and $B_2=\{\mathbf{a}^{u_1}(s)\circ \mathbf{a}^{j}(l-1)\circ \mathbf{w}^{(2)}_{i}:i\in [t-2],j\in\{0,1\}\}$.
\begin{align*}
\mathbf{u}\circ (D_2(\mathbf{a})\cap D_2(\overline{\mathbf{a}})) \circ \mathbf{v}&=\{\mathbf{a}^{u_1}(s)\circ \mathbf{a}^{j}(l-2)\circ \mathbf{a}^{v_1}(t):j\in\{0,1\}\}\\
&~~~~~~~\bigcup \{\mathbf{a}^{u_1}(s)\circ \mathbf{w}_{i}\circ \mathbf{a}^{v_1}(t):i\in [2l-6]\}=A_3\cup B_3,
\end{align*}
where $A_3=\{\mathbf{a}^{u_1}(s)\circ \mathbf{a}^{j}(l-2)\circ \mathbf{a}^{v_1}(t):j\in\{0,1\}\}$ and $B_3=\{\mathbf{a}^{u_1}(s)\circ \mathbf{w}_{i}\circ \mathbf{a}^{v_1}(t):i\in [2l-6]\}$.
By comparison, we have $A_i\cap B_j=\emptyset$ for any $i,j\in[3]$, and $B_i\cap B_j=\emptyset$ for any $i,j\in [3]$ and $i\neq j$. Hence, we only consider a set $A_1\cup A_2\cup A_3$. If $l=2$ then $A_1=\{\mathbf{a}^{u_1}(s-1)\circ 0 \circ\mathbf{a}^{v_1}(t),\mathbf{a}^{u_1}(s-1)\circ 1 \circ\mathbf{a}^{v_1}(t),\mathbf{a}^{\overline{u_1}}(s-1)\circ 0 \circ\mathbf{a}^{v_1}(t),\mathbf{a}^{\overline{u_1}}(s-1)\circ 1 \circ\mathbf{a}^{v_1}(t)\}$,
$A_2=\{\mathbf{a}^{u_1}(s)\circ 0\circ \mathbf{a}^{v_1}(t-1),\mathbf{a}^{u_1}(s)\circ 1\circ \mathbf{a}^{v_1}(t-1),\mathbf{a}^{u_1}(s)\circ 0\circ \mathbf{a}^{\overline{v_1}}(t-1),\mathbf{a}^{u_1}(s)\circ 1\circ \mathbf{a}^{\overline{v_1}}(t-1)\}$,
$A_3=\{\mathbf{a}^{u_1}(s)\circ \mathbf{a}^{v_1}(t)\}$, and $B_3=\emptyset$. If $\mathbf{z}\in A_1\cap A_2$ then $\mathbf{z}_{[1,s]}=\mathbf{a}^{v_1}$ and $\mathbf{z}_{[s+1,s+t]}=\mathbf{a}^{v_1}(t)$. It is easily verified that $\mathbf{a}^{u_1}(s)\in \{\mathbf{a}^{u_1}(s-1)\circ 0,\mathbf{a}^{u_1}(s-1)\circ 1\}$ and $\mathbf{a}^{v_1}(t)\in \{0\circ \mathbf{a}^{\overline{v_1}}(t-1),1\circ \mathbf{a}^{\overline{v_1}}(t-1)\}$. Hence, $\mathbf{a}^{u_1}(s)\circ \mathbf{a}^{v_1}(t)\in A_i$ for any $i\in [3]$. So, $|A_1\cup A_2\cup A_3|=7$. Therefore, we have
\begin{align*}
D_2(\mathbf{x})&\cap D_2(\mathbf{y})=|A_1\cup A_2\cup A_3|+|B_1\cup B_2\cup B_3|=7+2(s-2+t-2)=2n-5.
\end{align*}
Similarly, when $l\geq 3$, $\min\{s,t\}=0$, and $\max\{s,t\}\geq 1$ then $D_2(\mathbf{x})\cap D_2(\mathbf{y})=2n-5$. Furthermore, when $l\geq 3$, and $s,t\geq 1$ then $D_2(\mathbf{x})\cap D_2(\mathbf{y})=2n-6$.
\end{proof}

\section*{Acknowledgments}
This research is supported in part by the National Key Research and Development Program of China under Grant Nos. 2022YFA1005000 and 2022YFA1004900, the National Natural Science Foundation of China under Grant Nos. 12001134, 62371259, and 62571301, and the Fundamental Research Funds for the Central Universities of China (Nankai University), and the Nankai Zhide Foundation.


\begin{thebibliography}{1}
\bibliographystyle{IEEEtran}


\bibitem{Church}
G.~M. Church, Y. Gao, and S. Kosuri, ``Next-generation digital information storage in DNA,'' \textit{Science}, vol. 337, no. 6102, pp. 1628--1628, 2012.\par

\bibitem{Yazdi}
S. Yazdi, H.M. Kiah, E.R. Garcia, J. Ma, H. Zhao, and O. Milenkovic, ``DNA-based storage: Trends and methods,'' \textit{IEEE Trans. Molecular, Biological, Multi-Scale Commun.}, vol. 1, no. 3, pp. 230--248, 2015. \par

\bibitem{Lenz}
A. Lenz, P. H. Siegel, A. Wachter-Zeh, and E. Yaakobi, ``Coding over sets for DNA storage,'' \textit{IEEE Trans. Inf. Theory}, vol. 66, no. 4, pp. 2331--2351, 2020.\par

\bibitem{Parkin}
S. S. Parkin, M. Hayashi, and L. Thomas, ``Magnetic domain-wall racetrack memory,'' \textit{Science}, vol. 320, pp. 190--194, 2008.\par

\bibitem{Chee}
Y. M. Chee, H. M. Kiah, A. Vardy, E. Yaakobi, and V. K. Vu, ``Coding for racetrack memories,'' \textit{IEEE Trans. Inf. Theory}, vol. 64, no. 11, pp. 7094--7112, 2018.\par

\bibitem{K1}
E. Konstantinova, ``Reconstruction of permutations distorted by single reversal errors,'' \textit{Discrete Applied Mathematics}, vol. 155, pp. 2426--2434, 2007.\par

\bibitem{L1}
V. I. Levenshtein, ``Efficient reconstruction of sequences,'' \textit{IEEE Trans. Inf. Theory}, vol. 47, no. 1, pp. 2--22, 2001.\par

\bibitem{L2}
V. I. Levenshtein, ``Efficient reconstruction of sequences from their subsequences or supersequences,'' \textit{Journal of Combin. Theory, Ser. A}, vol. 93, no.2, pp. 310--332, 2001.\par

\bibitem{L3}
V. I. Levenshtein, E. Konstantinova, E. Konstantinov, and S. Molodtsov, ``Reconstruction of a graph from $2$-vicinities of its vertices,'' \textit{Discrete Applied Mathematics}, vol. 156, pp. 1399--1406, 2008.\par

\bibitem{L4}
V. I. Levenshtein and J. Siemons, ``Error graphs and the reconstruction of elements in groups,'' \textit{Journal of Combin. Theory, Ser. A}, vol. 116, pp. 795--815, 2009.\par

\bibitem{Sala1}
F. Sala, R. Gabrys, C. Schoeny, and L. Dolecek, ``Exact reconstruction from insertions in synchronization codes,'' \textit{IEEE Trans. Inf. Theory}, vol. 63, no. 4, pp. 2428--2445, 2017.\par

\bibitem{Gabrys}
R. Gabrys and E. Yaakobi, ``Sequence reconstruction over the deletion channel,'' \textit{IEEE Trans. Inf. Theory}, vol. 64, no. 4, pp. 2924--2931, 2018.\par

\bibitem{Pham}
V. L. P. Pham, K. Goyal, and H. M. Kiah, ``Sequence reconstruction problem for deletion channels: a complete asymptotic solution,'' In \textit{Proc. IEEE Int. Symp. Inform. Theory}, pp. 992--997, 2022.\par

\bibitem{Pham2}
V. L. P. Pham, K. Goyal, and H. M. Kiah, ``Sequence reconstruction problem for deletion channels: a complete asymptotic solution,'' \textit{Journal of Combin. Theory, Ser. A}, https://doi.org/10.1016/j.jcta.2024.105980, 2025.\par

\bibitem{Yaakobi}
E. Yaakobi, M. Schwartz, M. Langberg, and J. Bruck, ``Sequence reconstruction for Grassmann graphs and permutations,'' In \textit{Proc. Int. Symp. Inform. Theory}, pp. 874--878, 2013.\par

\bibitem{Wang1}
X. Wang, ``Reconstruction of permutations distorted by single Kendall $\tau$-errors,'' \textit{Cryptogr. Commun.}, vol. 15, pp. 131--144, 2023. \par

\bibitem{Wang2}
X. Wang and E. V. Konstantinova, ``The sequence reconstruction problem for permutations with the Hamming distance,'' \textit{Cryptogr. Commun.}, vol. 16, pp. 1033--1057, 2024. \par

\bibitem{Wang3}
X. Wang, F.-W. Fu, and E. V. Konstantinova, ``The sequence reconstruction problem for permutations with the Hamming distance,'' \textit{Des. Codes Cryptogr.}, vol. 93, no. 1, pp. 11--37, 2025. \par

\bibitem{Calabi}
L. Calabi, ``On the computation of Levenshtein's distances,'' \textit{TN-9-0030, Parke Math. Labs., Inc., Carlisle, MA.}, 1967.

\bibitem{Chrisnata}
J. Chrisnata, H. M. Kiah, and E. Yaakobi, ``Correcting deletions with multiple reads,'' \textit{IEEE Trans. Inf. Theory}, vol. 68, no. 11, pp. 7141--7158, 2022.

\bibitem{Lan}
Z. Lan, Y. Sun, W. Yu, and G. Ge, ``Sequence reconstruction under channels with multiple bursts of insertions or deletions,''  \textit{IEEE Trans. Inf. Theory}, vol. 72, no. 1, pp. 315-330, 2026. 

\bibitem{Sun1}
Y. Sun, Y. Xi, and G. Ge, ``Sequence reconstruction under single-burst insertion/deletion/edit channel,'' \textit{IEEE Trans. Inf. Theory}, vol. 69, no. 7, pp. 4466-4483, 2023.

\bibitem{Zhang}
D. Zhang, G. Ge, and Y. Zhang, ``Sequence reconstruction over 3-deletion channels,'' In \emph{Proc. IEEE Int. Symp. Inf. Theory}, pp. 891-896, 2024.

\bibitem{Song}
W. Song, K. Cai, and T. Q. S. Quek, ``Sequence reconstruction for the single-deletion single-substitution channel,''  \textit{IEEE J. Sel. Areas Inf. Theory}, vol. 6, pp. 232-247, 2025.

\bibitem{Cai} K. Cai, H. M. Kiah, T. T. Nguyen, and E. Yaakobi, ``Coding for sequence reconstruction for single edits,'' \textit{IEEE Trans. Inf. Theory}, vol. 68, no. 1, pp. 66-79, 2022.

\bibitem{Sun2} Y. Sun and G. Ge, ``Correcting two-deletion with a constant number of reads,'' \textit{IEEE Trans. Inf. Theory}, vol. 69, no. 5, pp. 2969-2982, 2023.

\bibitem{Sun3} Y. Sun and G. Ge, ``Bounds and constructions of $\ell$-read codes under the Hamming metric,'' \textit{IEEE Trans. Inf. Theory}, vol. 71, no. 8, pp. 5868-5883, 2025.

\bibitem{Sun4} Y. Sun and G. Ge, ``Codes for correcting a burst of edits using weighted-summation VT sketch,'' \textit{IEEE Trans. Inf. Theory}, vol. 71, no. 3, pp. 1631-1646, 2025.

\bibitem{Ye} Z. Ye, X. Liu, X. Zhang, and G. Ge, ``Reconstruction of sequences distorted by two insertions,'' \textit{IEEE Trans. Inf. Theory}, vol. 69, no. 8, pp. 4977-4992, 2023.

\bibitem{Wu} R. Wu and X. Zhang, ``Balanced reconstruction codes for single edits,'' \textit{Des. Codes Cryptogr.,} vol. 92, pp. 2011-2029, 2024.

\end{thebibliography}
\end{document}